\documentclass[10pt]{article}
\usepackage{amsmath}
\usepackage{dsfont}
\usepackage{amsfonts}
\usepackage{amssymb}
\usepackage{amsthm,xcolor}
\usepackage[normalem]{ulem}
\theoremstyle{plain}
\begingroup
\newtheorem{theorem}{Theorem}[section]
\newtheorem{lemma}[theorem]{Lemma}
\newtheorem{proposition}[theorem]{Proposition}
\newtheorem{corollary}[theorem]{Corollary}

\newtheorem{rmk}[theorem]{Remark}

\newtheorem{remark}[theorem]  {Remark} 
\endgroup

\theoremstyle{remark}
\begingroup
\endgroup

\mathchardef\emptyset="001F

\numberwithin{equation}{section}

\newcommand{\op}[1]{{\rm{#1}}}

\newcommand{\R}{\mathbb R}

\newcommand{\N}{\mathbb{N}}

\newcommand{\calP}{{\mathcal{P}}}

\newcounter{margcount} 
\title{Next-order asymptotic expansion for $N$-marginal optimal transport with Coulomb and Riesz costs}

\author{
Codina Cotar
\footnote{
University College London, Statistical Science Department,
London,
United Kingdom,
\texttt{c.cotar@ucl.ac.uk}}
\, and
Mircea Petrache
\footnote{Pontificia Universidad Catolica de Chile, Santiago, Chile,
\texttt{decostruttivismo@gmail.com}}}
\begin{document}

\maketitle

\textit{Abstract:} Motivated by a problem arising from Density Functional Theory, we provide the sharp next-order asymptotics for a class of multimarginal optimal transport problems with cost given by singular, long-range pairwise interaction potentials. More precisely, we consider an $N$-marginal optimal transport problem with $N$ equal marginals supported on $\mathbb R^d$ and with cost of the form $\sum_{i\neq j}|x_i-x_j|^{-s}$. In this setting we determine the second-order term in the $N\to\infty$ asymptotic expansion of the minimum energy, for the long-range interactions corresponding to all exponents $0<s<d$. We also prove a small oscillations property for this second-order energy term. Our results can be extended to a larger class of models than power-law-type radial costs, such as non-rotationally-invariant costs.
The key ingredient and main novelty in our proofs is a robust extension and simplification of the Fefferman-Gregg decomposition \cite{Feff85},  \cite{Gregg89}, extended here to our class of kernels, and which provides a unified method valid across our full range of exponents. Our first result generalizes a recent work of Lewin, Lieb and Seiringer \cite{LewLiebSeir17}, who dealt with the second-order term for the Coulomb case $s=1,d=3$.
\vspace{0.2in}

\section{Introduction}\label{secintro}
\subsection{The next-order asymptotics}\label{ssecnextorderasymptotics} 
Let $\mathsf{c}(x,y)=g(x-y):\mathbb R^d\times \mathbb R^d\to\mathbb R\cup\{+\infty\}$ be a ``pairwise interaction'' cost function, and consider the space $\mathcal P_{sym}^N(\mathbb R^d)$ of probability measures on $(\mathbb R^d)^N$ which are invariant under the action induced by the permutation group $S_N$ acting on the $N$ coordinates $(x_1,\ldots, x_N)\in(\mathbb R^d)^N$. We consider the following $N$-marginals optimal transport (OT) minimization problem, for given $\mu\in\mathcal P(\mathbb R^d)$ and $N\ge 2$:
\begin{equation}\label{mainmin}
F_{N, \mathsf{c}}^\mathrm{OT}(\mu):=\inf\left\{\int_{(\mathbb{R}^d)^N}\sum_{1\le i\neq j\le N}\mathsf{c}(x_i,x_j)d \gamma_N(x_1,\ldots,x_N):\ \gamma_N\in\mathcal P^N_{sym}(\mathbb{R}^d), \gamma_N\mapsto \mu\right\}\ .
\end{equation}
The notation $\gamma_N\mapsto\mu$ means that $\gamma_N$ has one-body density $\mu$ (physics terminology) or equivalently equal
$\mathbb{R}^d$-marginals $\mu$ (probability terminology),
\begin{equation} \label{marginals}
    \gamma_N(\underbrace{\mathbb{R}^{d}\times\ldots\times\mathbb{R}^d}_\textrm{i-1 times}\times A_i \times \underbrace{\mathbb{R}^{d}\times\ldots\times\mathbb{R}^d}_\textrm{N-i times}) = \int_{A_i}{\mu(x)}\,dx \mbox{ for all Borel}\ A_i\subseteq \mathbb{R}^d
    \mbox{ and all }i=1,\ldots,N\ .
\end{equation}
Even though our techniques work for more general costs, we will consider especially the case of
\begin{equation}\label{minineq}
\mathsf{c}(x,y)=|x-y|^{-s}\quad\mbox{if}\quad 0<s<d,~~x,y\in \mathbb{R}^d\ ,
\end{equation}
and we will prove the following asymptotic expansion for $F_{N,\mathsf{c}}^\mathrm{OT}(\mu)$. For all $0<s<d$, if $\mu$ has density $\rho\in L^{1+s/d}(\mathbb{R}^d)$, then we have (see Theorem \ref{upboundcont} below)
\begin{equation}\label{expnextord}
F_{N,\mathsf{c}}^\mathrm{OT}(\mu) = N^2 \int_{\mathbb R^d\times \mathbb R^d}\frac{\rho(x)\rho(y)}{|x-y|^s} dx\ dy - C(s,d) N^{1+s/d}\int_{\mathbb R^d}\rho^{1+s/d}(x) dx + o(N^{1+s/d})\quad\quad\text{ as }\quad N\to\infty\ .
\end{equation}
Moreover we will prove that the \textit{strictly positive} constant $C(s,d)$ is independent of the choice of the marginal $d\mu(x)=\rho(x)dx$ and therefore can be interpreted as arising in an independent model problem. Furthermore, we will derive a small oscillations principle (with respect to $N$) for 
$$N^{-1-s/d}\bigg(F_{N,\mathsf{c}}^\mathrm{OT}(\mu) - N^2 \int_{\mathbb R^d\times \mathbb R^d}\frac{\rho(x)\rho(y)}{|x-y|^s} dx\ dy\bigg),
$$
which could also be interpreted as a rough third order asymptotic bound for $F_{N,\mathsf{c}}^\mathrm{OT}(\mu)$ with $\mathsf{c}$ as in \eqref{minineq}.

Our methods are extendable to more general costs, as non-rotationally-invariant ones. See Remark \ref{morecosts} below.
\par For $d=3,s=1$, the $F_{N,\mathsf{c}}^\mathrm{OT}(\mu)$ was introduced in the physics literature in the context of Density Functional Theory (DFT) by Seidl, Perdew, Levy, Gori-Giorgi, and Savin \cite{Seidl99, SPL99, SGS07}, without them being aware of its meaning in optimal transport. Namely, for $s=1,d=3$, \eqref{mainmin}  is a natural semiclassical limit to the famous \emph{Hohenberg-Kohn (HK)} functional from quantum mechanics, originally introduced by Hohenberg-Kohn in \cite{HK64}, and rigorously proved by Levy and Lieb in \cite{Le79}, \cite{Li83}. The connection to optimal transport was mathematically established later by \cite{CoFriKlu11} and \cite{ButdePasGG12} for $N=2$,  later further extended to $N=3$ in \cite{BinddePasc17}, and recently independently proved for all $N\ge 2$ by \cite{CoFriKlu17} and \cite{Lewin17}.
\par In the process of establishing our two main results, we were required to prove, as key new tools for them, a set of additional secondary results of independent interest and of possible use to other settings, leading to generalized versions of the Fefferman-Gregg decomposition and positive definiteness criteria, as described in Section \ref{ssecfefferman}. What we use is a decomposition of positive definite kernels following the strategy established for $s=1,d=3$ by Fefferman \cite{Feff85}, and extended by Gregg \cite{Gregg89}, in $d=3$, with methods which work there for $0<s<3$ (but which in general dimension work only for $0<s<2+[(d-1)/2]$, where $[\cdot]$ denotes the integer part). The decomposition is based on a separation of the kernel into contributions from different ranges, coupled with a good packing strategy for the domain. We extend the range of validity of the Fefferman-Gregg decomposition, simplify the proofs, and apply the strategy to our new problem.
\par We note here that, together with our main result \eqref{expnextord}, the extended Fefferman-Gregg decomposition also allows to establish the equality of the next-order term for the two minimization problems below:
\begin{enumerate}
\item The energy-minimization for Coulomb and Riesz gases
\begin{equation}\label{gasmin}
\mathcal E_{N, \mathsf{c}}(V):=\inf\left\{\sum_{1\le i\neq j\le N}\textsf{c}(x_i-x_j) + N\sum_{i=1}^NV(x_i)\text{ for }\ x_1,\ldots,x_N\in\mathbb R^d: x_1,\ldots,x_N\in\mathbb R^d\right\},
\end{equation}
for $\mathsf{c}$ as in (\ref{minineq}), and where $V:\mathbb R^d\to\mathbb R$ is a suitable external "confining" potential assumed to be bounded below, lower semicontinuous, such that $\{x: V(x)<\infty\}$ has non-zero $\textsf{c}$-capacity, and such that $V(x)\to\infty$ as $x\to\infty$. The asymptotic expansion of $\mathcal E_{N, \mathsf{c}}(V)$ can be written as
$$\mathcal E_{N, \mathsf{c}}(V)=N^2 \mathcal I_{\mathsf{c},V}(\mu_V)-C_\mathrm{Jel}(s,d)N^{1+s/d}\int\rho^{1+s/d}(x)dx+o(N^{1+\frac{s}{d}}) \quad\quad\text{ as }\quad N\to\infty,$$ 
where $\mu_V$ is the unique minimizer of 
$$\mathcal I_{\mathsf{c},V}(\mu):=\int_{\mathbb R^d}\int_{\mathbb R^d}\textsf{c}(x-y)d\mu(x)d\mu(y)+\int_{\mathbb R^d}V(x)d\mu(x),$$
and where the constant $C_\mathrm{Jel}(s,d)>0$ corresponds to the minimum energy of a unit density Jellium-type problem, as described for example in \cite{ps}.
\item The minimum of the optimal transport problem considered in \eqref{expnextord}. Here, if we write $C_\mathrm{UEG}(s,d):=C(s,d)$, the next-order term can then be written as 
$$C_\mathrm{UEG}(s,d)N^{1+s/d}\int\rho^{1+s/d}(x)dx,$$
where $C_\mathrm{UEG}(s,d)$ can be interpreted as the energy of a Uniform Electron Gas, as described in \cite{CotPet17}. 

As a remark about terminology, note that the model of classical electrons corresponds to the Coulomb-type interaction potentials with $s=1, d\le 3$, while other more general ``Uniform Gas'' problems appear for pairwise interactions mediated by general kernels $\mathsf{c}(x-y)$ in general dimensions.
\end{enumerate}
This equality of the above constants, $C_\mathrm{Jel}(s,d)=C_\mathrm{UEG}(s,d)$, is proved for $d-2<s<d$ in \cite{CotPet17}. 
\par The Fefferman-Gregg decomposition also allows to find the precise next-order term for the Jellium problem for $d-2 \le s<d$, giving an alternative proof for the main result in \cite{ps}, as we will prove in future work.

\subsection{Comments about the decomposition techniques used in this paper}

\subsubsection{Relevance of the decomposition in the proof} 
The importance of positive-definiteness in our problem is crucial, and was emphasized in \cite{Petr15}: the property of our kernels of being balanced-positive-definite is \emph{equivalent} to the convergence of the renormalized minimum energies $N^{-2}F_N^\mathrm{OT}[\rho]$ to the mean field $\int_{\mathbb{R}^d}\int_{\mathbb{R}^d} \mathsf{c}(x,y)d\mu(x)d\mu(y)$.
\par This ingredient was not present in the problem considered by \cite{Feff85} and \cite{Gregg89}, which was of different nature, and is thus specific to our optimal-transport problem. 
\par In the decomposition of our power-law kernels, we can only hope to have next-order bounds for the $E_N^\mathrm{xc}(\rho)$-energies of the decomposed kernels when the first-order contribution is cancelled by subtracting the mean field. By the above result \cite{Petr15} we thus have the binding constraint that \emph{all the kernels in our decomposition must be positive definite}. This explains why we put much emphasis in the positive-definite kernel decomposition below. 

\subsubsection{Link to previous work on the Fefferman-Gregg technique}

As described in the abstract of \cite{Gregg89}, that paper treats inverse-power kernels in dimension $d=3$, with $s$ in a range around $1$. This is based on a decomposition of the kernels into positive-definite finite-range parts, together with a very tame error, as done for the case $s=1, d=3$ by \cite{Feff85}, based on the ``Swiss cheese'' packings present in \cite{Lieblebowitz} and \cite{Hughes89}. The key observation in \cite{Gregg89} is that the kernels with $s\neq d-2$ do not satisfy a version of Newton's lemma. This forbids the spherical averaging methods from \cite{Lieblebowitz} to apply to $s\neq d-2$.

\par The construction in \cite{Gregg89} is formulated in $d=3$ for $0<s<2$, due to other constraints coming from the kinetic part of the energy considered there, but the kernel decomposition methods hold for $0<s<3$, and in general dimension they allow to treat exponents $0<s<2+[(d-1)/2]$, not covering the whole range $0<s<d$. We thus optimized here several steps of that strategy, to that aim, as described below. Note that the range $0<s<2$ appearing in \cite{Gregg89} is due to $-\Delta_{\R^{Nd}} + \sum_{1\le i\neq j\le N}|x_i-x_j|^{-s}$ being a self-adjoint operator only for $0<s<2$. This is also true in general dimension $d$, as follows from Reed-Simon Vol. II, p. 169 and Thm. X.19. (The last result is stated in $d=3$ but extends/can be adapted to general $d$, but the range of self-adjointness of the Hamiltonian, $0<s<2$, is the one that stays valid for general $d\ge 1$.)

\par A related property which seems to never have emphasized before, concerning the methods of \cite{Feff85} and \cite{Gregg89}, is that they work also for \emph{non-rotationally-invariant kernels}. The lack of rotation invariance of the \emph{background field} instead was one of the motivations for \cite{Feff85, Gregg89}, as compared to the previous work \cite{Lieblebowitz} where radial symmetrization was a crucial ingredient in the proofs.

\par The fact that the decomposition is available for non-rotationally-invariant interaction kernels, is a property that seems not to have been exploited in applications yet, and forms an important point of advantage of this type of decomposition, as compared to other methods such as \cite{Grafschenker95} used in \cite{LewLiebSeir17}.

\par The functions called $Q_s^i$ in \cite{Gregg89} do not provide bounds on $d$ partial derivatives, in general dimension $d$, and derivative bounds for self-convolutions of characteristic functions of balls $1_B*1_B$ can be improved by only one extra derivative by their method, in order to use the positive-definiteness criteria of Lemma \ref{fourierbdlemma}. This would only allow to use the lemma for $0<s<2+[(d-1)/2]$, and not for the whole range $0<s<d$. Therefore we introduce a better mollification to produce the functions which we denote by $Q_{i,\eta}$ here. The functions $Q_{i,\eta}$ then have the derivative bounds required in Lemma \ref{fourierbdlemma} in general dimension $d$.

\par In \cite{Gregg89} an averaging procedure on ball packings used in the decomposition (needed there to pass from bounds on $2$ derivatives of the kernels, to bounds on $3$ derivatives) is done by changing the radii of balls from a fixed packing. In \cite{Gregg89} a new packing is constructed by the Swiss cheese lemma for each dilation factor, and the volume-fractions covered by the balls from each family only have the rough bounds coming from that lemma. This produces some extra error terms in the final superposition of the kernels, error terms whose control is then not discussed in \cite{Gregg89}, and is made harder by the fact that explicit decomposition formulas such as \eqref{decompqi} below are not available in that case. 

\par To solve these problems, we use a different procedure: we dilate and contract the \emph{whole family} of balls given in the Swiss cheese lemma, thereby producing covers which can be fit to our mollification procedure directly, by the algebraic dependence of the $Q_{i,\eta}$ on the dilation parameter of our families. In particular, the dilated families automatically cover fixed volume-fractions, and we avoid some of the error terms from \cite{Gregg89}.

\subsection{Link to previous related works}
The result in \eqref{expnextord} has been very recently proved for the specific case of $s=1,d=3,$ in \cite{LewLiebSeir17}. Note that the method used therein, by means of the Graf-Schenker decomposition, does not extend past that specific $s=1,d=3$ situation to our general class of costs, and it is mentioned in that paper that it would be interesting to consider the general situation for exponents and dimensions $0<s<d$, treated here. As it turns out, two points are different between our methods:
\begin{itemize}
\item We replace the Graf-Schenker decomposition used in \cite{LewLiebSeir17} by the generalized Fefferman-Gregg-type approach based on positive definiteness, approach which can be seen as a natural continuation of the new understanding from \cite{CFP}, \cite{Petr15} of the crucial role of positive definiteness in asymptotics problems such as the one considered here.
 \item A general difference between this paper and \cite{LewLiebSeir17} is that we rely on and are inspired by Optimal Transport tools, such as duality and the method \cite{ButChampdePas16} for giving some separation condition on the points in the optimizer. We hope that the reader will profit from comparing our approach and the Statistical Physics framework of \cite{LewLiebSeir17}, both of which highlight different aspects of the theory.
\end{itemize}
\par Moreover, we rely in our proofs on probability ideas such as convergence approximations.

\par We note that our tools are also applicable, by similar arguments, to other models than power-law-type radial costs, such as \emph{anisotropic costs} and costs with \emph{radial oscillations} (see Remark \ref{morecosts}, Examples (b), (c)). We also can transfer the study done here to the case of natural interaction kernels in \emph{curved spaces}, e.g. on compact Riemannian manifolds (see Remark \ref{morecosts}, Example (a)).

\par We remark that in $d=1$ the next order term can be directly described by a very elegant computation for a very large class of costs (as was explained to us by Simone Di Marino \cite{DimarPer}), by using the explicit ``monotone rearrangement'' description of the optimal transport plan from \cite{ColdePascdiMar13}.  Further results and a detailed review of optimal transport results with Coulomb costs and other repulsive costs can be found in \cite{diMarGerNe15}.

\subsection{Range of validity and future work}\label{sssecrangeofvalidity}

\textbf{Screening and effective energy localization.} Concerning the use of statistical mechanics methods, the form of our problem \eqref{minineq} creates a certain number of complications. As a main comparison, recall that in the study of the next order for \emph{classical} Coulomb/Riesz gas energies (so far understood only for $s\in[d-2,d)$: see \cite{ps} and the references therein), the source of localization of energy needed for obtaining the next-order term was directly provided via the boundary value problems connected to the elliptic PDE's coming from the kernels under consideration. However the methods from \cite{ps} are at the moment not extendable outside $d-2\le s<d$, and the next-order Jellium-type energy considered in \cite{ps} may be infinite for $0<s<d-2$.

\par In order to achieve the full range of exponents $0<s<d$, we apply the more robust strategy, initiated by Fefferman in \cite{Feff85}, and which consists, roughly speaking, in \emph{truncating the kernels instead of truncating the solutions} of the equations, and this is done precisely in such a way as to preserve positive definiteness, at least up to an asymptotically negligible error. The careful positivity-preserving kernel truncation method used here can be then interpreted as a rough and robust counterpart of the above PDE methods, which is available for operators which are more general, namely \emph{positive definite}, rather than elliptic. In this precise sense, we could say that the Fefferman truncation method which we extend and apply here mainly gives some way of an answer to the basic question: \emph{What is the analogue of a Neumann boundary value problem, if we replace the Laplacian with a general positive definite operator?} This question seems to have appeared and to have been answered already before, in a quite different form, in the theory of finite range decompositions: see \cite{BrydgesTalar} and the references therein.

\par\textbf{Comparison to the Coulomb/Riesz gas asymptotics.} As noted in the appendix \cite{LewLieb15}, there seems to be a discrepancy in $s=1,d=3,$ between the computations for the value of the Jellium energy and that of the Uniform Electron Gas (UEG) energy, on the specific example of lattice-like configurations.

\par Extending this consideration to the \emph{minimizers} of the corresponding UEG and Jellium energies would lead to the interesting consequence that there would be, for $s=d-2$, a \emph{gap between the constant} $\mathsf{C}_\mathrm{Jel}(d-2,d)$ appearing in the next-order term expansions for the large-$N$ expansions of the Coulomb gas minimum energy (described in \cite{ss2d}, \cite{ps}, as a minimum jellium energy) and the constant $C(d-2,d)$ whose existence is found in our main result \eqref{expnextord} (which corresponds to an UEG energy minimization problem, as appearing e.g. in Proposition \ref{unif} for the case $\Lambda=[0,1]^d$).

\par For the Coulomb case $s=d-2$, the precise estimates needed for proving the presence of a gap (beyond computations for special examples such as lattices) seems to require a precise understanding of the Jellium and UEG minimizers, which goes beyond the state-of-the-art results currently available. However for the exponents $s\in(d-2,d)$, for which both the values $\mathsf{C}_\mathrm{Jel}(s,d)$ and $C(s,d)$ are computable due to our result \eqref{expnextord} and to \cite{ps}, these values agree as proved recently in \cite{CotPet17}.

\par\textbf{Towards a general localization theory for positive definite operators.} The first paper where positive-definite truncations for the Coulomb interaction potential were constructed/used seems to be \cite{Feff85}, then refined and extended to more general $d=3$ interactions by \cite{Gregg89} and \cite{Hughes89}. A different, simpler construction by means of the Yukawa potential was introduced by Conlon, Lieb and Yau \cite{ConLiebYau89}. Later on a much simpler construction, which is specific to the case $d=3$ and to the Coulomb cost appeared in \cite{Grafschenker95}.

\par Well-behaved truncations of operators of positive type are a key tool used in the recent renormalization group results such as \cite{david'sresult}, \cite{gordonwithroland}, \cite{AdamsKoteckyMuller}. In these cases the setting is often $\mathbb Z^d$ rather than $\mathbb R^d$, and the kernels that have to be localized are often the Coulomb kernels, interpreted as infinite-dimensional positive definite matrices. General decompositions have been studied for example in \cite{BrydgesTalar} and \cite{Bauerschmidt}.

\par These results so far profit of explicit representations of the operators involved, however one should be able to extend the truncation theory such as presented here to both discrete settings and more general positive definite operators. We plan to pursue this direction in future work.

\subsection{Main results}\label{ssecmainresult}

In our statements and proofs below, we define a generalized version $E_{N,\mathsf{c}}^\mathrm{xc}(\mu)$ of the so-called ``exchange-correlation'' energy for a probability measure $\mu\in\mathcal P(\R^d)$ and a cost function $\mathsf{c}:\mathbb R^d\times\mathbb R^d\to \R_{>0}\cup\{+\infty\}$ such that the integrals below are finite
\begin{equation}\label{defexc}
E_{N,\mathsf{c}}^\mathrm{xc}(\mu):= F_{N,\mathsf{c}}^\mathrm{OT}(\mu)-N^2\int_{\mathbb{R}^d}\int_{\mathbb{R}^d} \mathsf{c}(x,y)d\mu(x)d\mu(y)\ .
\end{equation}
The above definition, which is the starting point for our paper, could be stated for general finite positive measures $\mu$, under the condition that we apply the right normalization for the second term. Therefore the requirement $\mu\in \mathcal P(\mathbb R^d)$ constitutes our choice of normalization. For our costs of particular interest as in \eqref{minineq}, we will also use for $0<s<d$ the notation $F_{N,s}^\mathrm{OT}(\mu)$, respectively $E_{N,s}^\mathrm{xc}(\mu)$. 

We can extend the definition of $F_{N,\mathsf{c}}^\mathrm{OT}(\mu)$ also to $N\in\{0,1\}$, by taking $F_{N,\mathsf{c}}^\mathrm{OT}(\mu):=0$ in these two cases. Moreover, for all $N\in\mathbb{R}_{>0}, N\ge 2$, let us define the ``grand-canonical optimal transport"
\begin{subequations}\label{OT_relax}
\begin{equation}\label{OTGC}
F_{\mathrm{GC},N,\mathsf{c}}^\mathrm{OT}\left(\mu\right):=\inf \left\{\sum_{n=2}^\infty\alpha_n F_{n,\mathsf{c}}^\mathrm{OT}(\mu_n)\left|\begin{array}{c}\sum_{n=0}^\infty\alpha_n=1,\ \sum_{n=1}^\infty n \alpha_n\mu_n = N\mu,\\[3mm]\mu_n\in {\cal P}(\mathbb{R}^d),\ \alpha_n\ge 0,\quad\mbox{for all}\quad n\in\N\end{array}\right.\right\}\ ,
\end{equation}
and the ``grand-canonical exchange correlation energy"
\begin{equation}
\label{ExcGC}
E_{\mathrm{GC},N,\mathsf{c}}^\mathrm{xc}\left(\mu\right):=F_{\mathrm{GC},N,\mathsf{c}}^\mathrm{OT}\left(\mu\right)-N^2 \int_{{\mathbb{R}^d}\times {\mathbb{R}^d}}\mathsf{c}(x,y)d\mu(x)d\mu(y).
\end{equation}
\end{subequations}
Here again, the classical definition of \eqref{OT_relax} is usually given only for $N\ge 2$, though one can define the quantities for all $N>0$. The reason for introducing $\alpha_0$ and $\alpha_1$ in \eqref{OTGC}, even though they do not appear in the sums of optimal transport problems there, is that they naturally appear in the proof of Lemma \ref{subadd12gena01} below (see Step 4 therein). Furthermore, our convention for $F_{N,\mathsf{c}}^\mathrm{OT}$ for $N\in\{0,1\}$ implies that for $N\in(0,1]$, we have  $F_{\mathrm{GC},N}(\mu)=0$, as we can take a competitor of the form $\alpha_n=0,n\ge 2, \alpha_0=1-N,\alpha_1=N, \mu_n=\mu, n\ge 1$. Also, the same convention implies that in \eqref{OTGC} we may equivalently minimize over sums $\sum_{n=0}^\infty \alpha_nF_{n,\mathsf{c}}^\mathrm{OT}(\mu_n)$.
\begin{theorem}\label{upboundcont}
Fix $d\ge1$ and $0<s<d$ and let $\mathsf{c}$ be the corresponding cost as in \eqref{minineq}. Assume that $\mu\in\calP(\mathbb{R}^d)$ has density $\rho\in L^{1+s/d}(\mathbb{R}^d)$. Then there exists $C(s,d)>0$, which depends only on $s$ and $d$, such that we have (where both limits exist for the full sequence)
\begin{equation}
 \label{upb}
\lim_{\substack{N\to\infty\\\N\in\N_+}}N^{-1-\frac{s}{d}}E_{N,s}^\mathrm{xc}(\mu)= \lim_{N\to\infty\atop N\in\mathbb{R}_+}N^{-1-\frac{s}{d}}E_{\mathrm{GC},N,s}^\mathrm{xc}(\mu)=-C(s,d)\int_{\mathbb{R}^d} \rho^{1+\frac{s}{d}}(x)dx\ .
\end{equation}
\end{theorem}
\begin{rmk}
After a first draft of the present paper was completed, the very important case $s=1,d=3,$ of Theorem \ref{upboundcont}, appeared in \cite[Thm. 4.1]{LewLiebSeir17} (for the case of continuous slowly-varying densities in $L^{4/3}(\mathbb{R}^3)$). In \cite{LewLiebSeir17}, a positive function $\rho\in L^1(\mathbb R^d)$ is called slowly-varying,  if the oscillations at infinity are controlled in the sense that there holds $\sum_{k\in\mathbb Z^d}\max_{x\in [0,1)^d+k}\rho(x) <\infty$. The proof of the upper bound of Theorem \ref{upboundcont} is based on classical tools such as subadditivity, and is similar to \cite{LewLiebSeir17}, to which we often refer the reader for the proofs. The more difficult sharp lower bound from Section \ref{seclowerboundpwconstant} requires different and more robust kernel decomposition techniques compared to \cite{LewLiebSeir17}, making our methods extendable to non-isotropic and oscillating kernels. 
\end{rmk}
\begin{rmk}[extensions to more general costs]\label{morecosts}
We note the following possible extensions of Theorem \ref{upboundcont}:
\begin{itemize}
 \item [(a)] If $(M,g)$ is a compact Riemannian manifold, then we may consider the optimal transport problem for the interaction cost $G(x,y)$ given by the Green function of the Laplacian of $(M,g)$. This cost is positive definite and has the same homogeneity near zero as the Coulomb potential $|x-y|^{2-d}$ where $d=\mathrm{dim}(M)$. This type of minimization has been studied in \cite{BeCoCri17}. For the case of \emph{embedded submanifolds} $M\subset \mathbb R^D$ this fits in the same framework as the celebrated Smale's $7$th problem \cite{Sma98}. For applying our methods to this example it is not essential that the operator under consideration be of second order, and it is possible to consider also interactions given by the Green functions of other other fractional or higher order positive definite operators.
\item [(b)] The costs $\mathsf{c}(x,y)$ to which our method applies need not be rotation-invariant, as the decomposition method described here uses only translation-invariance unlike the Graf-Schenker approach \cite{Grafschenker95}. This allows to extend our results to the costs of the form 
\[
 \mathsf{c}(x,y):=g(x-y), \text{ with } g(x)=|x|^{-s} f (x/|x|), \text{ where }f\in C^0(\mathbb S^{d-1})\ .
\]
Non-rotation-invariant costs such as the above are not treatable by any previous methods, to our knowledge.
\item [(c)] Other kernels treatable by our methods are those of the form 
\[
 \mathsf{c}(x,y):=g(x-y), \text{ with }g(x)=l(|x|)f(x/|x|), \text{ with }l(r)=r^{-s-3}(\sin(r) - r \cos(r))\text{ and }f\in C^0(\mathbb S^{d-1})\ ,
\]
again for $0<s<d$, which are including radial oscillations too. Kernels with radial oscillations may occur in physics e.g. in the study of Friedel oscillations.
\item [(d)] With considerable more effort regarding the Fefferman-Gregg adaptation than the relatively straightforward cases (a)-(c) above, the proof could potentially also be adapted to other cost cases such as
\[
 \mathsf{c}(x,y):= \prod_{i=1}^d |x_i-y_i|^{-\alpha}, \quad \text{for }x=(x_1,\ldots,x_d),y=(y_1,\ldots,y_d)\in\mathbb R^d,\text{ and }0<\alpha<1\ ,
\]
and to a corresponding version including radial oscillations as well.
\end{itemize}
In the more general cases (a)-(c), we note that we have $\lim_{\epsilon\downarrow 0}\frac{\mathsf{c}(y+\epsilon x,y)}{g_0(\epsilon x)}=1$ for all $y$ (and this should be transferred to a local chart at $y$, in case (a)), locally uniformly for $x\neq 0$, where the function $g_0:\mathbb R^d\setminus\{0\}\to\mathbb R$ is homogeneous of degree $-s$ for some $0<s<d$, i.e. for all $x\neq 0$ and all $\lambda>0$ there holds $g_0(\lambda x)=\lambda^{-s}g_0(x)$. In that setting, we expect the next-order term limit treated here to be
\begin{equation}\label{pd}
 \lim_{N\rightarrow\infty}N^{-1-s/d}E_{N, \mathsf{c}}^\mathrm{xc}(\mu)= -C(g_0,s)\int_{\mathbb{R}^d}(\rho(x))^{1+\frac{s}{d}} dx,\quad\mbox{where}\quad C(g_0,s)>0\ .
\end{equation}
\end{rmk}
Moreover, $E_{\mathrm{GC},N, s}^\mathrm{xc}\left(\mu\right)$ satisfies the small oscillations property detailed in Theorem \ref{monincrsub} below, which is not only crucial for settling the continuity of $s\to C(s,d)$ in \cite{CotPet17}, but it is also interesting in its own right. As a counterpart for the Jellium problem, oscillation bounds for minimizers of the Coulomb gas energy appeared in \cite{rns} for the log-interaction in $d=2$ and were extended in \cite{rnp} to Coulomb interactions $s=d-2$ in general dimension $d\ge 2$, and to $d-2<s<d$ under an extra hypothesis. Note that in the setting of \cite{rns,rnp}, spatial oscillations and uniformity of the asymptotic spatial distribution of optimum configurations were controlled, whereas here we control oscillations of the energy values.
\begin{theorem}[Small oscillations property of $E^{\mathrm{xc}}_{\mathrm{GC}, N,s}(\mu)$]\label{monincrsub}
Fix $0<\epsilon<d/2$ and let $\epsilon\le s\le d-\epsilon$. Let $\mu\in \mathcal {P}(\mathbb{R}^d)$ be a probability measure with density of the form $\rho(x)=\sum_{i=1}^k\alpha_i1_{\Lambda_i}(x)$, where $\Lambda_1,\ldots,\Lambda_k$ are hyperrectangles with disjoint interiors, and $\alpha_i\in \mathbb{R}_{>0}, i=1,\ldots k,$ are such that $\sum_{i=1}^k\alpha_i=1$.

Then there exists $C(\rho,d,\epsilon)>0$ such that for all $N, \widetilde N\in \mathbb{R}_{>0},$ $N\ge\widetilde N\ge 2$, we have
\begin{equation}
\label{smalloscgc}
\left|\frac{E_{\mathrm{GC}, N,s}^{\mathrm{xc}}(\mu)}{N^{1+s/d}}-\frac{E_{\mathrm{GC},\widetilde N,s}^{\mathrm{xc}}(\mu)}{\widetilde N^{1+s/d}}\right|\le \frac{C(\rho, d,\epsilon)}{\log \tilde N}\ ,
\end{equation}
where $C(\rho,d,\epsilon)$ depends on $\rho, d, \epsilon$ but is independent of the choice of the parameter $s\in[\epsilon,d-\epsilon]$.
\end{theorem}
\begin{rmk}
Though for simplicity of calculations we only prove Theorem \ref{monincrsub} for the setting described therein, we may extend the above small-oscillations bound to general Borel sets $\Lambda_i$ with $\phi$-regular boundary (see Definition \ref{phiregular} below), and also to more general densities, such as any positive Riemann-integrable function, i.e. any function that can be approximated well by piecewise constant functions on hyper-rectangles. This includes in particular continuous densities with compact support and $\phi$-regular boundary.

However, with the present strategy, this may come at the cost of a less tame oscillation bound than the one of order $1/(\log\widetilde N)$ of \eqref{smalloscgc} above. Note that the order can be improved for $s=1,d=3,$ to $1/{\tilde N}^\alpha,\alpha>0,$ with an $s$-dependent constant, by using the Graf-Schenker decomposition \cite{Grafschenker95} instead of the Fefferman-Gregg one.
\end{rmk}
The above small oscillation result connects to the conjecture about the \emph{optimal Lieb-Oxford bounds}, initially formulated for the case $s=1, d=3,$ but actually open in the whole range $0\le s<d$, $d\ge 2$ (where for $s=0$ we consider the kernel $\mathsf{c}(x,y)=-\log|x-y|$ instead), regarding the precise value of the optimal constant $C_\mathrm{LO}(s,d)$ such that for all $N\ge 2$ and for all measures $\mu$ with density $\rho\in L^{1+s/d}(\R^d)$ there holds
\begin{equation}\label{conj-lo}
\frac{E_{N,s}^\mathrm{xc}(\mu)}{N^{1+s/d}}\ge - C_\mathrm{LO}(s,d)\int_{\R^d}\rho^{1+s/d}(x)dx.
\end{equation}
As mentioned in \cite{LewLiebSeir17}, the fact that the optimal constant is realized in the limit $N\to\infty$ would fit with numerical results.

\subsubsection{Fefferman-Gregg decomposition}

We next state the generalized Fefferman-Gregg decomposition, originally introduced by Fefferman \cite{Feff85} for $s=1,d=3$, extended by Gregg \cite{Gregg89} to $0<s<2+[(d-1)/2]$, and further extended by us to $0<s<d$. This is the main new tool for the proofs of Theorem \ref{upboundcont} and of Theorem \ref{monincrsub}, as explained in the next section.
\begin{proposition}[Fefferman-Gregg decomposition]\label{prop3ws}
Let $M\in \mathbb N_+$, $0<\epsilon<d/2$ and $\epsilon\le s\le d-\epsilon$. Then there exists a constant $C$ depending only on $d,\epsilon$, a family $\Omega$ of ball packings $F_\omega$ of $\mathbb R^d, \omega\in\Omega$, a radius $R_1>0$ and a probability measure $\mathbb P$ on $\Omega$ such that the cost $|x_1-x_2|^{-s}$ can be decomposed as follows:
\begin{equation}\label{decomptouse}
\frac{1}{|x_1-x_2|^s}=\frac{M}{M+C}\left\{\int_{\Omega}\left(\sum_{A\in F_\omega}\frac{1_A(x_1)1_A(x_2)}{|x_1-x_2|^s}\right)d\mathbb{P}(\omega)+w(x_1-x_2)\right\}\ ,
\end{equation}
where
\begin{enumerate}
\item $w$ is positive definite;
\item if $\mu\in\calP(\mathbb{R}^d)$ has density $\rho\in L^{1+\frac{s}{d}}(\mathbb R^d)$ then
\[
E_{N, w}^\mathrm{xc}(\mu)\ge E_{\mathrm{GC},N, w}^\mathrm{xc}(\mu)\ge -\frac{C(w,d,\epsilon)}{M}N^{1+s/d}\int_{\mathbb{R}^d}\rho^{1+s/d}(x)dx-\frac{C(w,d,\epsilon)}{M}R_1^{-s}(N-1)\ 
\]
for some $C(w,d, \epsilon)>0,$ which depends only on $w, d$ and $\epsilon$. 
\end{enumerate}
\end{proposition}

\subsubsection{Summary of constructions for Proposition \ref{prop3ws}}
While we stated here the proposition in a more self-contained way, we note some further parameters which will intervene and allow to fit the decomposition within the rest of the proof:
\begin{itemize}
 \item The parameters $M$ and $R_1>0$ above, as well as the further scale parameter $\ell>0$ introduced below are going to satisfy the relative constraints of the Swiss cheese Lemma \ref{cheeselemma}.
 \item The packing family $\Omega$ and $\mathbb{P}$ will therefore later depend on $l$ and will be denoted $\Omega_l$, respectively ${\mathbb{P}}_l$. Each family $\Omega_l$ will be composed of $(l\mathbb Z)^d$-periodic packings of balls $F_\omega^l$, obtained from Lemma \ref{cheeselemma}.
 \item We will perform the precise choice of $\Omega_l,\mathbb P_l$ in \eqref{lbrough3}, and the constant $C$ figuring in Proposition \ref{prop3ws} will be fixed in \eqref{errorw}.
\end{itemize}
\begin{remark}\label{contmorec1}
As already stated in the introduction, the decomposition \eqref{decomptouse} can be obtained for much more general costs, such as for example, the costs from Remark \ref{morecosts}. In this more general setting, the decomposition is of the form
\begin{equation}\label{decompgeneral}
\mathsf{c}(x_1,x_2)=\frac{M}{M+C}\left\{\int_{\Omega}\left(\sum_{A\in F_\omega}1_A(x_1)1_A(x_2)\mathsf{c}(x_1,x_2)\right)d\mathbb{P}(\omega)+w(x_1,x_2)\right\}\ ,
\end{equation}
with $w$ positive definite and (at least for the costs (a)-(c) from Remark \ref{morecosts}, and using the asymptotic profile $g(x)=|x|^{-s}f(x/|x|)$ as defined there)
\begin{equation}\label{roughboundgeneral}
E_{N, w}^\mathrm{xc}(\mu)\ge E_{\mathrm{GC},N, w}^\mathrm{xc}(\mu)\ge-\frac{C_\mathrm{rough}(f,s, d)}{M}{N^{1+\frac{s}{d}}\int_{\mathbb{R}^d} (\rho(x))^{1+\frac{s}{d}}}dx-\frac{C_\mathrm{rough}(f,s, d) }{M}{R_1^{-s}}(N-1)\ ,
\end{equation}
for some $C_\mathrm{rough}(f,s,d)$ depending only on $f,s,d,$ and $C$ depends only on $s,d$. For an explanation of the above, see Remark \ref{contmorec2}.
\end{remark}

\subsection{Strategy of the proof of Theorem \ref{upboundcont}}
\subsection*{The case of uniform marginals}
We start by individuating the constant $C(s,d)$ in \eqref{upb} in the main theorem. To determine its value, it suffices to consider the case $\rho(x)=1_\Lambda/|\Lambda|$ for a specific Borel-measurable set $\Lambda$ of positive measure. We treat this case in Section \ref{secpreliminaries}, and this is the result of Proposition \ref{unif}. The proof is by a classical method based on subadditivity and on the scaling properties of our functionals, which is direct consequence of the $(-s)$-homogeneity of our kernels.
\subsection*{Marginals with piecewise constant density}
The first main difficulty is to prove the sharp upper/lower bound from \eqref{upb} for the case of $\mu$ with  density $\rho$ which is the sum of characteristic functions of a finite union of disjoint hyperrectangles. (Such densities $\rho$ are called ``piecewise constant functions'' in an analysis terminology and the measures $\mu$ are called ``mixtures of uniform measures'' in a probability terminology; throughout the paper, we use the former terminology.)

While the proof of the upper bound will be an easy consequence of the strategy above, the lower bound is a lot more involved and requires to use Proposition \ref{prop3ws}. The construction and main principle leading to this proposition are as described at an informal level below:
\subsection*{Use of the Fefferman-Gregg decomposition}
\begin{enumerate}
\item We construct a family of packings, here denoted by $\{F_\omega^l\}_{\omega\in \Omega_l}$, each $F_\omega^l$ consisting of balls of small radii $R_1,\ldots,R_M$, forming a geometric series, and then periodized at a slightly larger scale $\sim l$. 

The parameters $M,l,R_1,\ldots, R_M$ are fixed in the Swiss cheese Lemma \ref{cheeselemma}.

For the purposes of our schematic explanation, only the parameter $l$ will be relevant, and the other parameters, as well as the specific choice of $F_\omega^l$, are used in the proof of Proposition \ref{prop3ws}. For the related explanation see Section \ref{ssecfefferman}.

\item Separately for each cover $F_\omega^l$, we perform the following decomposition of the kernel $\mathsf{c}$ (see also (\ref{decomptouse})):
\begin{equation}\label{followingdecomp}
c(x_1,\ldots, x_N):=\sum_{1\le i\neq j\le N}\mathsf{c}(x_i,x_j)=\sum_{A\in F_\omega^l}\sum_{1\le i\neq j\le N}1_A(x_1)1_A(x_2)\mathsf{c}(x_i,x_j) + err_\omega^l(x_1,\ldots,x_N),
\end{equation}
where $err_\omega^l$ is an error term, to be carefully estimated. The precise error estimate will be given later.
\item Equation \eqref{followingdecomp} translates to an ``averaged'' inequality, in which the expectation is taken with respect to a suitable probability measure on $\Omega_l$ (for the explicit description, see \eqref{lbrough3} below):
\begin{equation}\label{averagedconverse}
\int_{\Omega_l}\left(\sum_{A\in F_\omega^l} E_{N\mu(A)}^\mathrm{xc}\left(\frac{\mu|_A}{\mu(A)}\right)\right) d\mathbb P_l(\omega)\le E_N^\mathrm{xc}(\mu) + err,
\end{equation}
where in order to present the heuristics better we forget for a moment about the complications coming from the fact that $N\mu(A)$ may not be integer or that $\mu(A)$ may be zero.
\item Since we are considering the case that the density $\rho$ is constant on each of a disjoint union of regular enough sets (hyperrectangles, in this case), it turns out that the contribution of those $A\in F_\omega^l$ on which $\rho$ is not constant, becomes negligible in the limit $l\to0$. Suppose for simplicity of exposition that for all $A\in F_\omega^l$ we have $\rho|_A$ constant, equal to $|A|^{-1}\mu(A)$. Then for each $\omega\in\Omega_l$ the expression \eqref{averagedconverse} approximates a Riemann sum. Indeed for each term in the sum \eqref{averagedconverse} we find
\begin{equation}\label{rescale}
E_{N\mu(A)}^\mathrm{xc}\left(\frac{\mu|_A}{\mu(A)}\right) \simeq - (N\mu(A))^{1+s/d}C(s,d)|A|^{-s/d} = -C(s,d) N^{1+s/d}\int_A\left(\rho|_A\right)^{1+s/d} dx,
\end{equation}
and then, using the fact that our packings are asymptotically precise as $l\to 0$, we find that uniformly in $\omega\in\Omega_l$, the sums appearing \eqref{averagedconverse} are asymptotic to the desired integral:
\begin{equation}\label{tothedesiredintegral}
\int_{\Omega_l}\left(\sum_{A\in F_\omega^l} E_{N\mu(A)}^\mathrm{xc}\left(\frac{\mu|_A}{\mu(A)}\right) \right)d\mathbb P_l(\omega)=-(1+o(1))N^{1+s/d}\int_{\mathbb R^d}\rho^{1+s/d}(x)dx,\quad\text{ as }\quad l\to 0, N\to\infty.
\end{equation}
\end{enumerate}
The way to treat the error terms $err$ from \eqref{averagedconverse}, is to first note that they correspond to the same type of minimization problem as the starting one, but with cost $w$ as in Proposition \ref{prop3ws}, i.e. $err=E_{N,w}^\mathrm{xc}(\mu)$.

For the $err$-term to really give ``small error'' contribution with the next-order scaling $N^{1+s/d}$, and not for example a leading-order contribution of order $N^2$, one must use, in a sharp way, the \emph{screening} (or \emph{charge nullity}) behavior of minimizers $\gamma_N$, quantifying that $\gamma_N$,s $F_N^\mathrm{OT}$-type energy is locally \emph{cancelled to leading order} by the corresponding energy coming from the mean field.

\par The criteria for this cancellation to happen was recently formalized in \cite{CFP} in the study of the first order term in the asymptotics of $F_N^\mathrm{OT}$: the \emph{positive definiteness} of the kernel plays the main role, and in \cite{Petr15} it was proved that the \emph{necessary and sufficient} condition for the convergence to the mean field is that the kernel should be balanced positive definite. This necessary and sufficient criterion is precisely why it is a posteriori essential to require $w$ to be positive definite in Proposition \ref{prop3ws}, a property without which we would not be able to obtain the desired bounds in the second point of that proposition, which allow us to conclude \eqref{tothedesiredintegral} from \eqref{averagedconverse}.
\subsection*{The appearance of the functional $E_\mathrm{GC}$}
In the above description of the core of our proof, we did not consider the technical problem given by the fact that while using truncations of our cost to obtain localization, the number of marginals in the so-created ``truncated transport plans'' is actually not constant. This and the fact that $N\mu(A)$ appearing in \eqref{averagedconverse} is not necessarily integer, justifies here, like it did in \cite{LewLiebSeir17} and in previous works, the introduction of the grand-canonical version $E_{\mathrm{GC},N,\mathsf{c}}^\mathrm{xc}$ of $E^\mathrm{xc}_{N,\mathsf{c}}$, in which we fix the number of marginals only ``in average''. This functional is then used as a technical tool and as a replacement of $E^\mathrm{xc}_{N,\mathsf{c}}$ throughout the paper. Furthermore, the \textit{key} property of the grand canonical optimal transport problem $F_{\mathrm{GC},N,\mathsf{c}}^\mathrm{OT}$ that allows us to obtain (\ref{averagedconverse}), is that it splits the optimal transport problem on sums of costs defined on disjoint sets into sums of optimal transport problems on the costs on disjoint sets, as is detailed in Lemma \ref{subadd_gcb} from Section \ref{sssecsplittingthecost}.
\subsection*{Approximation of more general densities $\rho$ by piecewise constant ones}
The passage from piecewise constant marginals to our class of densities follows then by approximating the respective density by piecewise constant ones.

\subsection*{Small oscillations for $E_\mathrm{GC}$}
The proof of Theorem \ref{monincrsub} is based again on the Fefferman-Gregg decomposition, which in this case can be done for two distinct values of $N, \widetilde N$ at the same time. Here we look first once more to piecewise constant marginals as for Theorem \ref{upboundcont}, as the arising further approximation constants have a tame dependence on the $L^p$-norms of the marginals and can be controlled well. 

The core of the strategy consists of splitting our domain via Fefferman-Gregg decompositions. Then we ensure that all but an asymptotically negligible proportion of the small balls in the decomposition are completely contained in parts where the marginal is constant, and thus ``cut out'' uniform marginals. We next compare the uniform-case asymptotics separately to balls coming from the decompositions done at $N$ and at $\tilde N$ marginals, and use the precise asymptotic valid in that case to ensure the oscillation bound.

Note here that the small oscillation bound strategy cannot work for the initial functional $E^\mathrm{xc}$, because in order to bound a difference as appearing in \eqref{smalloscgc}, we require \emph{matching upper and lower bounds} of each of the two quantities, in terms of the same Fefferman-Gregg-sums. The (rigorous version of the) lower bound as in \eqref{averagedconverse} figures a sum of $E_{\mathrm{GC}}^\mathrm{xc}$-terms, and must be complemented by an (asymptotically) matching upper bound, which we can only obtain via the classical tool of subadditivity. As a sharp subadditivity bound for $E^\mathrm{xc}$ only furnishes upper bounds via $E^\mathrm{xc}$-terms, and a priori we have only the strict inequality $E^\mathrm{xc}_\mathrm{GC}<E^{\mathrm{xc}}$, this forces us in order to find a sharp \emph{matching upper bound} to work with $E_\mathrm{GC}^\mathrm{xc}$.

\subsubsection{Plan of the paper} 
In \textbf{Section \ref{secpreliminaries}} we establish some preliminary definitions and useful formulas.

\par The main result in \textbf{Section \ref{secoptimalconstant}} is the optimal sharp upper bound in the next-order asymptotics, for the case of probability measures $\mu$ whose density is piecewise constant on a union of hyperrectangles with disjoint interiors.

\par In \textbf{Section \ref{seclowerboundpwconstant}} we provide the sharp lower bound matching the one from Section \ref{secupperbound}. This is the core result of our paper. In \textbf{Section} \ref{sssecsplittingthecost} we introduce the ``grand-canonical version'' $E_{\mathrm{GC}, N}^\mathrm{xc}$ of $E_{N}^\mathrm{xc}$ from \eqref{ExcGC}, and we give some  preliminary lemmas, to be used in the rest of paper. We present the Fefferman-Gregg decomposition for our kernel in a self-contained way in \textbf{Section \ref{ssecfefferman}}. The proof of the sharp lower bound result of Proposition \ref{mixextlowb} for piecewise constant densities is the object of \textbf{Section \ref{sseclowerboundssecfeffermandecompositiontion}}. In \textbf{Section} \ref{ssecproofmainthm} we show the proof of Theorem \ref{upboundcont}.

\par In \textbf{Section} \ref{secsmalloscprop} we give the proof to our second main result, Theorem \ref{monincrsub}. 

\par In \textbf{Appendix \ref{appendixfefferman}} we provide detailed proofs of the results pertaining to our decomposition of kernels present in Section \ref{seclowerboundpwconstant}. \textbf{Appendix \ref{seccomput}} contains a brief remark to the sharp function spaces of densities to be expected in our generalized setting, in \textbf{Appendix \ref{secunifboundLO}} we give a statement of a Lieb-Oxford bound which is uniform in $s$. \textbf{Appendix \ref{Remark4.9}} presents the proofs of some key properties of $E_{\mathrm{GC}, N}^\mathrm{xc}$ stated in Section \ref{sssecsplittingthecost}, and Appendix \textbf{Section} \ref{ssecoptimaltransport} gives some helpful optimal transport results.
\section{Preliminaries and notation}\label{secpreliminaries}
We will use when convenient the notation $\mathsf{c}(x,y):=g(x-y)$. We assume from now on that $\gamma_N\in\calP_{sym}^{N}(\mathbb{R}^d)$ is a solution to \eqref{mainmin}; such a solution exists for our costs $\mathsf{c}$ of interest by standard arguments as
given e.g in Villani \cite{Vil}. We note that if $\mathsf{c}(x,y)=|x-y|^{-s}$ for $\mu$ with density $\rho\in L^1(\mathbb{R}^d)$, we have the following function space sharp requirements that arise naturally in relation with this type of kernel:
\begin{subequations}\label{finmu}
\begin{eqnarray}
\sup_{x\in \mathbb R^d}\int_{\mathbb R^d} |x-y|^{-s}\rho(y)dy<\infty&\mbox{ for }&\rho\in L^{\frac{d}{d-s},1}(\mathbb{R}^d),\label{finmu1}\\
\int_{\mathbb R^d}\int_{\mathbb R^d} |x-y|^{-s}\rho(x)\rho(y)dx\ dy<\infty&\mbox{ for }&\rho\in L^{\frac{2d}{2d-s},2}(\mathbb R^d),\label{finmu2}
\end{eqnarray}
\end{subequations}
where the spaces $L^{p,q}(\mathbb{R}^d)$ are the Lorentz spaces (see Appendix Section \ref{seccomput}). Note in particular that for $0<s<d$ there follows $\frac{d}{d-s}>1+\frac{s}{d}>\frac{2d}{2d-s}$, therefore the requirement \eqref{finmu1} is stronger than $\rho\in L^{1+s/d}(\mathbb R^d)$ and \eqref{finmu2} is less restrictive than $\rho\in L^{1+s/d}(\mathbb{R}^d)$, at least for $\rho\in L^1(\R^d)$.
\begin{remark}\label{boundrem}
Let $\mathsf{c}(x,y)=|x-y|^{-s}$ and set $0< s<d$. Then for all $N\ge 2$ and all $\mu\in\mathcal {P}(\mathbb R^d)$ with $\rho\in L^{1+\frac{s}{d}}(\mathbb R^d)$, we have for some $c_{\mathrm{LO}}(s,d)>0$ which does not depend on $N$ and $\mu$
\begin{equation}\label{loterm1}
-c_{\mathrm{LO}}(s,d)\int_{\mathbb{R}^d}\rho^{1+\frac{s}{d}}(x)dx \le N^{-1-\frac{s}{d}}\left[F_{N,s}^\mathrm{OT}(\mu)-N^2\int_{\mathbb{R}^d}\int_{\mathbb{R}^d} g(x-y)d\mu(x)d\mu(y)\right]\le 0\ .
\end{equation}
Due to \eqref{finmu}, the space $\rho\in L^1(\mathbb{R}^d)\cap L^{\frac{2d}{2d-s},2}(\mathbb R^d)$ gives a sharp condition ensuring that the mean field energy for $\rho$ is finite, and if also $\rho\in L^{1+\frac{s}{d}}(\mathbb R^d)$ then the next-order term is finite too and \eqref{loterm1} follows by Lemma 16 from \cite{LundNamPort}, which result extends the Lieb-Oxford inequality from $s=1,d=3,$ to $0<s<d$. 
\end{remark}
An optimal transport result that we will use in the paper is the following many-marginals Monge-Kantorovich duality result, proved in \cite{depascaleduality} for the Coulomb cost, and later extended to more general costs in \cite{ButChampdePas16}. Following \cite{ButChampdePas16}, we assume that
\begin{equation}\label{ass_c_main_intro}
\mathsf{c}(x,y)=g(x-y)=l(|x-y|)\quad\text{where}\quad l:[0,\infty)\to[0,+\infty]\quad\text{is}\quad\left\{\begin{array}{l}\mbox{continuous on $(0,\infty)$}\ ,\\\mbox{strictly decreasing}\ ,\\\mbox{such that }\lim_{t\to 0^+}l(t)=+\infty\ .\end{array}\right.
\end{equation}
\begin{proposition}\label{depascdual}
Let $\mathsf{c}$ satisfy conditions \eqref{ass_c_main_intro}. Then for any $\mu\in {\cal P}(\mathbb{R}^d)$ and any $N>1$, the equality
\[
F_{N,\mathsf{c}}^\mathrm{OT}(\mu)=\sup_{f\in {\cal I}_{\mu}}\left\{N\int_{\mathbb{R}^d} f(x)d\mu(x):\sum_{i=1}^N f(x_i)\le \sum_{1\le i\neq j\le N} \mathsf{c}(x_i,x_j)\right\}
\]
holds, where ${\cal I}_{\mu}$ denotes the set of $\mu$-integrable functions and the pointwise inequality is satisfied everywhere.
\end{proposition}
In preparation for the proof of Proposition \ref{unif} below, we first need the following subadditivity result,  proved in Lemma 2.5 from \cite{LewLiebSeir17}. We observe here that such subadditivity methods have been used previously in similar settings (see, e.g. \cite{HLSI09} or \cite{HLSII09}). Note that, unlike \cite{LewLiebSeir17}, we work in our statements with probability measures.

We recall here that our convention is that $F_{1,\mathsf{c}}^\mathrm{OT}(\mu)=0$ and thus $E_{1,\mathsf{c}}^\mathrm{xc}(\mu)=-\int_{{\mathbb{R}^d}\times {\mathbb{R}^d}}\mathsf{c}(x,y)d\mu(x)d\mu(y)$.
\begin{proposition}\label{subadd3}
Let $\mathsf{c}:\mathbb{R}^d\times\mathbb{R}^d\rightarrow \mathbb R\cup\{+\infty\}$. Consider $k$ probability measures $\mu_1, \ldots, \mu_k,$ with densities respectively equal to $\rho_1, \ldots,\rho_k$, such that the quantities below are well-defined and finite (for $\mathsf{c}(x,y)=|x-y|^{-s}$, let $\rho_i\in L^{1+\frac{s}{d}}(\mathbb{R}^d), i=1,\ldots,k$). Fix $M_1,\ldots, M_k\in\mathbb N_+$, and let $\mu$ be the probability measure with density $\frac{\sum_{i=1}^kM_i \rho_i}{\sum_{i=1}^kM_i}$. Then the following subadditivity relation holds:
\begin{equation}\label{subadd1gen}
E_{\sum_{i=1}^kM_i,\mathsf{c}}^\mathrm{xc}\left(\mu\right):=E_{\sum_{i=1}^kM_i,\mathsf{c}}^\mathrm{xc}\left(\frac{\sum_{i=1}^kM_i\mu_i}{\sum_{i=1}^kM_i}\right)\le \sum_{i=1}^kE_{M_i,\mathsf{c}}^\mathrm{xc}(\mu_i).
\end{equation}
\end{proposition}
\begin{proof}
The proof for $k=2$ can be found in \cite{LewLiebSeir17}. Given the result for $k=2$, the general case follows by induction on $k$: for the inductive step one can apply the $k=2$ case to the measures $\mu=\mu_k, \mu'=\tfrac{\sum_{i=1}^{k-1}M_i\mu_i}{\sum_{i=1}^{k-1}M_i}$ and the numbers $M=M_k, M'=\sum_{i=1}^{k-1}M_i$.
\end{proof}
We state in Lemma \ref{scalOT} the different scalings of $F_{N,\mathsf{c}}^\mathrm{OT}(\mu)$ under change of $\mu$ only. The proof follows directly via Proposition \ref{depascdual}, and will be omitted. See also Corollary 2.6 from \cite{CaffMcCann} for an adaptation to the case of finite measures of the usual two-marginals Monge-Kantorovich duality.
\begin{lemma}\label{scalOT}
Let $\mu\in\mathcal{P}(\mathbb{R}^d)$ be {a probability measure with a density $\rho$ such that the quantities defined below are all finite}. Let $\mathsf{c}:\mathbb{R}^d\times\mathbb{R}^d\rightarrow\mathbb{R}\cup\{+\infty\}$. 
\begin{itemize}
\item [(a)] If we replace $\mu$ by $\beta \mu$ for $\beta>0$, then we get
\begin{equation}\label{scalmu}
F_{N,\mathsf{c}}^\mathrm{OT}(\beta\mu)=\beta F_{N,\mathsf{c}}^\mathrm{OT}(\mu),~~\mbox{which converges to}~~\beta\int_{\mathbb{R}^\times\mathbb{R}^d} c(x,y)d\mu(x)d\mu(y)~~\mbox{as}~~N\to\infty.
\end{equation}
%
%
%
\item [(b)] Set $0<s<d$ and let $\mathsf{c}(x,y)=|x-y|^{-s}$. If we replace for $\alpha>0$, $\mu$ by $\mu_\alpha$ defined as $d\mu_\alpha(x)=\rho_\alpha(x)dx$ with $\rho_\alpha(x)=\rho(\alpha x)$, then we get for such $\mathsf{c}$
\begin{equation}\label{scalrho}
F_{N,s}^\mathrm{OT}(\mu_\alpha)=\alpha^{-d+s}F_{N,s}^\mathrm{OT}(\mu),~~\mbox{which converges to}~~\alpha^{-d+s}\int_{\mathbb{R}^\times\mathbb{R}^d} c(x,y)d\mu(x)d\mu(y)~~\mbox{as}~~N\to\infty.
\end{equation}
Note that the transformation $\mu\mapsto\alpha^d\mu_\alpha$ maintains the property of $\mu$ of being a probability measure. Moreover, the above checks that
\begin{equation}\label{scalingexc}
E_{N,s}^\mathrm{xc}(\alpha^d\mu_\alpha) = \alpha^{s}E_{N,s}^\mathrm{xc}(\mu)\ .
\end{equation}
\end{itemize}
\end{lemma}
The convergence to the mean field from (\ref{scalmu}) and (\ref{scalrho}) follows from \cite{CFP} and \cite{Petr15}. We emphasise that, in view of \eqref{scalmu}, it will be crucial for us to ensure that the correct normalizations are used in the definition of $E_N^\mathrm{xc}(\mu)$, a fact that is automatically ensured because we work with probability measures $\mu$ rather than with positive measures. 

Before we proceed, we need to introduce regularity of sets. As in \cite{Fisher}, we say that a set $\Lambda$ has a \emph{$\phi$-regular boundary} if for some $t_0>0$ and for a continuous function $\phi:[0,t_0)\rightarrow \mathbb{R}^+$ with $\phi(0)=0$ there holds
\begin{equation}\label{phiregular}
\forall 0\le t\le t_0\ ,\ \left|\left\{x:\ d(x,\partial\Lambda)\le |\Lambda|^{1/d}t\right\}\right|\le\phi(t) |\Lambda|\ .
\end{equation}
We will mostly work with bounded Borel-measurable sets with $\phi$-regular boundary for some $\phi$, to make our estimates quantitative. Note that if a set has $\phi$-regular boundary then it is Jordan-measurable. We state next the following result, proved via classical subadditivity reasonings in \cite[Thm. 2.6]{LewLiebSeir17}) (see also \cite[Prop. 2]{RoRu} and \cite[Prop. 7.2.4]{Ru}), and  pointed out to us by M. Lewin in January 2016 in IHP. 
\begin{proposition}[Uniform electron gas]\label{unif}
Fix $0<s<d$ and let $\mathsf{c}(x,y)=|x-y|^{-s}$. Let $\mu\in\calP(\mathbb{R}^d)$ be a uniform measure, with density $\rho(x):=|\Lambda|^{-1}1_\Lambda(x)$ supported on a Borel set with $\phi$-regular boundary $\Lambda\subset\mathbb{R}^d$. Then there exists
\begin{equation*}
-\infty<\lim_{N\to\infty}N^{-1-\frac{s}{d}}\left[F_{N,s}^\mathrm{OT}(\mu)-N^2\int_{\mathbb{R}^d}\int_{\mathbb{R}^d} \frac{1}{|x-y|^s}d\mu(x)d\mu(y)\right] =\lim_{N\to\infty}N^{-1-\frac{s}{d}} E_{N,s}^\mathrm{xc}(\mu)<0\ .
\end{equation*}
Moreover, there exists $0<C(s,d)<\infty$, which is independent of $\Lambda$, such that
\begin{equation}\label{upbu1}
\lim_{N\to\infty}N^{-1-\frac{s}{d}} E_{N,s}^\mathrm{xc}(\mu)= -{C(s,d)}\int_\Lambda \rho^{1+\frac{s}{d}}(x)dx=-C(s,d)|\Lambda|^{-s/d}\ .
\end{equation}
\end{proposition}
Even though Proposition \ref{unif} is stated for sets with $\phi$-regular boundary, in fact for some of our results we only need to consider the more restricted case of hyper-rectangles with disjoint interiors, like in Theorem \ref{optconstmixt} below.
\begin{rmk}
Note that Proposition \ref{unif} can be adapted, with some more work, and by means of Proposition \ref{depascdual}, to the costs examples from Remark \ref{morecosts}, to derive in this case a limiting result as in \ref{pd}.
\end{rmk}
\section{Optimal constant for piecewise constant $\rho$}\label{secoptimalconstant}
As already explained in the introduction, before we proceed to the proof of our main statement for the general marginals case, we will need to understand first the case of marginals with piecewise constant densities. Even though stated only for $0<s<d$, the result actually holds in much greater generality of the costs provided that one can show existence of the corresponding limit to the one in Proposition \ref{unif}, and that one can extend the Fefferman-Gregg decomposition to these costs. In particular,  Proposition \ref{unif} can be shown to hold for the costs from Remark \ref{morecosts}.
\begin{theorem}[Optimal constant for marginals with piecewise constant density]\label{optconstmixt}
Fix $0<s<d$ and let $\mathsf{c}(x,y)=|x-y|^{-s}$. Set $k\ge 1$. For $i=1,\ldots,k,$ let $\mu_i\in\calP(\mathbb{R}^d)$ be the uniform measure on $\Lambda_i\subset\mathbb{R}^d$ with density $\rho_i$, where the $\Lambda_i, i=1,\ldots,k$ are Borel sets with $\phi$-regular boundaries and disjoint interiors. Then there exists $C(s,d)>0$, depending only on $s$ and $d$, such that
 \begin{equation} \label{eqadd4aa}
 {\lim}_{N\rightarrow\infty}N^{-1-s/d}E_{N,s}^\mathrm{xc}\left(\sum_{i=1}^k\alpha_i\mu_i\right)= {\lim}_{N\rightarrow\infty}N^{-1-s/d}E_{\mathrm{GC},N,s}^\mathrm{xc}\left(\sum_{i=1}^k\alpha_i\mu_i\right)=C(s,d) \int_{\R^d} \left(\sum_{i=1}^k\alpha_i\rho_i(x)\right)^{1+s/d}dx\ .
 \end{equation}
\end{theorem}
The theorem follows immediately in view of Proposition \ref{mixextsub} and of Proposition \ref{mixextlowb} below.
\subsection{Optimal upper bound for piecewise constant $\rho$}\label{secupperbound}
\begin{proposition}[Optimal upper bound for piecewise constant marginals]\label{mixextsub}
Fix $0<s<d$ and let $\mathsf{c}(x,y)=|x-y|^{-s}$. Set $k\ge 1$. For $i=1,\ldots,k,$ let $\mu_i\in\calP(\mathbb{R}^d)$ be the uniform measure on $\Lambda_i$ with density $\rho_i(x)=\tfrac{1_{\Lambda_i}(x)}{|\Lambda_i|}$, where $\Lambda_i\subset\mathbb{R}^d, i=1,\ldots,k$ are Borel sets with $\phi$-regular boundary and disjoint interiors. Let $\alpha_i\in \mathbb{R}_{>0}, i=1,\ldots k,$ be such that $\sum_{i=1}^k\alpha_i=1$. Then if $C(s,d)>0$ is the optimal constant from Proposition \ref{unif}, we have
\begin{eqnarray}\label{subadd4}
 \limsup_{N\rightarrow\infty}N^{-1-s/d}E_{N,s}^\mathrm{xc}\left(\sum_{i=1}^k\alpha_i\mu_i\right)
 &\le& -C(s,d) \sum_{i=1}^k\alpha_i^{1+s/d}\left(\int_{\Lambda_i} \rho_i^{1+s/d}(x)dx\right)\nonumber\\
 &=& -C(s,d) \int_{\R^d} \left(\sum_{i=1}^k\alpha_i\rho_i(x)\right)^{1+s/d}dx\ .
 \end{eqnarray}
 \end{proposition}
\begin{rmk}The proof relies only on Propositions \ref{subadd3} and \ref{unif}, and thus extends to more general costs.
\end{rmk}
\begin{proof}
We would like to use the formula in Proposition \ref{subadd3}. However,  this will not be an immediate application to our case due to the form $\sum_{i=1}^k\alpha_i\mu_i$ of our measure, so we will first need to re-write the terms in our probability measure to bring them to a suitable form. In view of \eqref{scalmu}, and to preserve convergence to the mean field, we need to work with probability measures rather than just with positive measures.

\par We write $\alpha_i=\frac{q_i^N}{N}, q_i^N\in\mathbb{R}_{>0}, i=1,\ldots, k$, with $\sum_{i=1}^k q_i^N=N$. 
We then write, denoting $\bar N:=\sum_{i=1}^k[q_i^N]$, 
\begin{eqnarray}
\sum_{i=1}^k\alpha_i\mu_i&=&\sum_{i=1}^k\frac{q_i^N}{N}\mu_i=\sum_{i=1}^k\frac{[q_i^N]}{N}\mu_i+\sum_{i=1}^k\frac{q_i^N-[q_i^N]}{N}\mu_i=\frac{\bar N}{N}\sum_{i=1}^k\frac{[q_i^N]}{\bar N}\mu_i+\frac{N-\bar N}{N}\sum_{i=1}^k\frac{q_i^N-[q_i^N]}{N-\bar N}\mu_i\nonumber\\
&=&\frac{\bar N}{N}\mu'+\frac{N-\bar N}{N}\mu''\ ,
\end{eqnarray}
with the obvious definitions for $\mu'$ and $\mu''$.

\par From a double application of Proposition \ref{subadd3} we have
\begin{equation}\label{subadd4int}
E^\mathrm{xc}_{N,s}\left(\sum_{i=1}^k\alpha_i\mu_i\right)\le E^\mathrm{xc}_{\bar N,s}(\mu')+E^\mathrm{xc}_{N-\bar N,s}(\mu'')\le  E^\mathrm{xc}_{\bar N,s}(\mu')\le \sum_{i=1}^k E^\mathrm{xc}_{[q_i^N],s}(\mu_i)\ ,
\end{equation}
where for the second inequality we used that $E^\mathrm{xc}_{N-\bar N}(\mu'')\le 0$. Equation \eqref{subadd4int} leads to
\begin{eqnarray}\label{splitsub}
\limsup_{N\rightarrow\infty}N^{-1-\frac{s}{d}} E^\mathrm{xc}_{N,s}\left(\sum_{i=1}^k\alpha_i\mu_i\right)&\le& \sum_{i=1}^k\limsup_{N\rightarrow\infty}N^{-1-\frac{s}{d}}E^\mathrm{xc}_{[q_i^N],s}(\mu_i)\nonumber\\
&=&\sum_{i=1}^k\limsup_{N\rightarrow\infty}\bigg(\frac{[q_i^N]}{N}\bigg)^{1+\frac{s}{d}}([q_i^N])^{-1-\frac{s}{d}}E^\mathrm{xc}_{[q_i^N],s}(\mu_i)\nonumber\\
&=&\sum_{i=1}^k\alpha_i^{1+\frac{s}{d}}\limsup_{N\rightarrow\infty}([q_i^N])^{-1-\frac{s}{d}}E^\mathrm{xc}_{[q_i^N],s}(\mu_i)\ ,\nonumber
\end{eqnarray}
where for the second equality we utilised that $[q_i^N]/N\le\alpha_i\le \left([q_i^N]+1\right)/N, i=1,\ldots, k.$

\par Fix $\delta>0$. By Proposition \ref{unif} for large enough $N$ we have for all $i=1,\ldots,k,$
\[
([q_i^N])^{-1-\frac{s}{d}}E_{[q_i^N],s}^\mathrm{xc}(\mu_i)\le -C(s,d) \int_{\Lambda_i} \rho_i^{1+\frac{s}{d}}(x)dx+\delta\ ,
\]
where $C(s,d)>0$ is the constant from Proposition \ref{unif}. Thus, for any fixed $\delta>0$
\begin{eqnarray}\label{firstinsplita}
\limsup_{N\rightarrow\infty}N^{-1-\frac{s}{d}}\sum_{i=1}^kE_{[q_i^N],s}^\mathrm{xc}(\mu_i)\le\sum_{i=1}^k \alpha_i^{1+\frac{s}{d}}\left(-C(s,d)\,\int_{\Lambda_i} \rho_i^{1+\frac{s}{d}}(x)dx+\delta\right) .\nonumber
\end{eqnarray}
Taking $\delta\rightarrow 0$ in the last line in the above produces
\begin{equation}\label{firstinsplit}
\lim_{N\rightarrow\infty}N^{-1-\frac{s}{d}}\sum_{i=1}^kE_{[q_i^N],s}^\mathrm{xc}(\mu_i)\le -C(s,d)\,\sum_{i=1}^k \,\alpha_i^{1+\frac{s}{d}}\int_{\Lambda_i} \rho_i^{1+\frac{s}{d}}(x)dx=\int_{\R^d} \left(\sum_{i=1}^k\alpha_i\rho_i(x)\right)^{1+s/d}dx,
\end{equation}
where we used the fact that $\rho_i(\Lambda_j)=0$ for all $1\le i\neq j\le k$ due to the disjointness assumption on the $\Lambda_i$ and to the definition of $\rho_i$.
\end{proof}
\section{Optimal lower bound for piecewise constant $\rho$}\label{seclowerboundpwconstant}
The main result of this section is the following:
\begin{proposition}[Optimal lower bound for piecewise constant marginals]\label{mixextlowb}
Fix $0<s<d$ and let $\mathsf{c}(x,y)=|x-y|^{-s}$. Set $k\ge 1$. For $i=1,\ldots,k,$ let $\mu_i\in\calP(\mathbb{R}^d)$ be the uniform measure on $\Lambda_i\subset\mathbb{R}^d$ with density $\rho_i(x)=\tfrac{1_{\Lambda_i}(x)}{|\Lambda_i|}$, where $\Lambda_i$ are Borel sets with $\phi$-regular boundary and disjoint interiors. Let $\alpha_i\in \mathbb{R}_{>0}, i=1,\ldots k,$ be such that $\sum_{i=1}^k\alpha_i=1$. Then if $C(s,d)>0$ is the optimal constant from Proposition \ref{unif}, we have
\begin{eqnarray} \label{supadd4a}
 \liminf_{N\rightarrow\infty}N^{-1-s/d}E_{N,s}^\mathrm{xc}\left(\sum_{i=1}^k\alpha_i\mu_i\right)&\ge& -C(s,d) \sum_{i=1}^k\alpha_i^{1+s/d}\int_{\Lambda_i} \rho_i^{1+s/d}(x)dx\nonumber\\
 &=& -C(s,d)\int_{\cup_{i=1}^k\Lambda_i} \left(\sum_{i=1}^k\alpha_i\rho_i(x)\right)^{1+s/d}dx\ .
\end{eqnarray}
\end{proposition}
In order to prove the statement of Proposition \ref{mixextlowb}, we will use the result in Proposition \ref{unif} above, together with the Fefferman-Gregg decomposition introduced below. 
\subsection{Fefferman-Gregg decomposition and positive definiteness}\label{ssecfefferman}
The main aim of this section is to introduce the setting and tools required to prove Proposition \ref{prop3ws} stated in the introduction.
\subsubsection{The translations of $\mathbb R^d$ and their action}\label{sssecgroupaction}
We consider the group of translations of $\mathbb R^d$ and its action
\begin{equation}\label{defgrx}
\tau_yx:=x+y\ .
\end{equation}
This action extends to functions, as usual: if $f:(\mathbb R^d)^k\to V$ where $V$ is a vector space, then we define $\tau_y(f)(x_1,\ldots,x_k) := f(\tau_yx_1,\ldots,\tau_yx_k)$. In particular we have the following properties:
\begin{enumerate}
\item For $y\in \mathbb R^d, f_1,f_2:(\mathbb R^d)^k\to V$ we have $\tau_y(f_1\cdot f_2) = (\tau_yf_1)\cdot (\tau_yf_2)$ if $\cdot$ is a scalar product on $V$
\item With the same notations as above, $\tau_y(f_1 * f_2)= (\tau_yf_1)*(\tau_yf_2)$.
\end{enumerate}
The first property above uses just the definition of the action and is valid more in general, whereas the second property uses the fact that translations are linear and measure-preserving transformations.

\par For $f:(\mathbb R^d)^k\to V$, where $V$ is a vector space, we define, whenever the following integral converges (in particular for compactly supported $f$, for example),
\[
\langle f \rangle (x_1,\ldots, x_k) := \langle f (x_1,\ldots, x_k)\rangle := \int_{\mathbb R^d}\tau_y f(x_1,\ldots, x_k) dy = \int_{\mathbb R^d} f(x_1+y,\ldots,x_k+y) dy\ .
\]
In the case $k=1$, we find that the above integral is defined for any $f\in L^1(\mathbb R^d, V)$ and equals the integral of $f$, by using the fact that translations are measure-preserving:
\[
\langle f\rangle (x) = \int_{\mathbb R^d} f(x+y)dy  =\int_{\mathbb R^d} f(y') dy'\ .
\]
By using the property 2. of our group action, and denoting
\[
f_-(x) := f(-x)\ ,
\]
we then find that, for $f_1,f_2:\mathbb R^d\to V$ such that $\left(f_1\right)_-*f_2\in L^1(\mathbb R^d,V)$, the following holds:
\begin{eqnarray*}
\langle f_1(x_1) f_2(x_2)\rangle&=&\int_{\mathbb R^d} f_1(x_1+y) f_2 (x_2+y) dy =\int_{\mathbb R^d} f_1(x_1-x_2+ y') f_2(y') dy'\nonumber\\
&=&\left(\left(f_1\right)_-*f_2\right)(x_2-x_1)\ ,
\end{eqnarray*}
which gives a function depending only of $x_1-x_2$.
\par In the case of $f_1=f_2=1_A$, which is the indicator function of a set $A$, we also find, as $x_1-x_2= \tau_yx_1 - \tau_yx_2$ for all $y\in\mathbb R^d$ and $\left(1_A\right) - = 1_{(-A)}$, that for any function $g_0:\mathbb R^d\to \mathbb R$ there holds
\begin{eqnarray*}
\left\langle 1_A(x_1)1_A(x_2) g_0(x_1-x_2) \right\rangle &=& g_0(x_1-x_2) \left\langle1_A(x_1)1_A(x_2) \right\rangle= g_0(x_1-x_2)1_{(-A)}*1_A (x_2-x_1)\nonumber\\
&:=& g_0(x_1-x_2)h_A(x_1-x_2)\ .
\end{eqnarray*}
Note that in particular, as a consequence of the fact that we performed an averaging $\langle\cdot\rangle$, there holds $h_{\tau A}=h_A$, i.e. the above function does not change under translations of $A$. In case $A=B_r(x)$ is a ball, by also using the fact that for $B_r:=B_r(0)$ we have $B_r=-B_r$, there holds
\begin{equation}\label{caseaisball}
\langle 1_{B_r(x)}(x_1)1_{B_r(x)}(x_2)\rangle=\langle 1_{B_r}(x_1)1_{B_r}(x_2)\rangle=h_{B_r}(x_1-x_2)=1_{B_r}*1_{B_r}(x_1-x_2)\ .
\end{equation}
The main idea developed by Fefferman \cite{Feff85}, Gregg \cite{Gregg89} and \cite{ConLiebYau89}, and later perfected by Graf and Schenker \cite{Grafschenker95} (based on the Yukawa potentials decomposition of \cite{ConLiebYau89}) for the case where one averages also over rotations $SO(d)$, is to reorder the above integrals such that sums over good packings or over tilings occur. For this, consider a lattice $\mathcal L\subset \mathbb R^d$ and let $\Omega_{\mathcal L}$ be a fundamental domain for $\mathcal L$. Then we can write
\begin{equation}
\int_{\mathbb R^d} f(y) dy =  \sum_{p\in \mathcal L} \int_{\Omega_{\mathcal L}}f(p+y')dy'.
\end{equation}
We may use this principle for reordering the $\langle f\rangle$ integrals as follows:
\begin{eqnarray}
\langle f(x_1,\ldots,x_k)\rangle &=& \int_{\mathbb R^d} f(x_1+y, \ldots, x_k + y)dy =\sum_{p\in \mathcal L}\int_{\Omega_{\mathcal L}}f(x_1 + p +y', \ldots x_k+p+y')dy'\nonumber\\
&=&\int_{\Omega_{\mathcal L}}\sum_{p\in \mathcal L+y'}f(x_1+p,\ldots,x_k+p) dy'.\label{averagerot}
\end{eqnarray}
\subsubsection{Localization procedures on packings} \label{ssseclocalization}
Note that we don't need our kernel to be rotation-invariant, i.e. we don't require $\mathsf{c}$ to have the form $\mathsf{c}(x,y)=l(|x-y|)$: this stronger requirement would be necessary only for the Graf-Schenker decomposition, in which averages over rotations appear, and not for the one we describe here.

\medskip

In order to obtain the asymptotic lower bound of our energies for kernels of the form $\mathsf{c}(x,y)=g(x-y)$, we desire to find a way in which to ``localize'' the interaction energies. Roughly stated, this means that if $F$ is a packing of $\mathbb R^d$ by disjoint sets, then we would like to find a decomposition of our kernel of the form 
\begin{equation}\label{lbrough1}
g(x-y)=\sum_{A\in F}1_A(x)1_A(y)g(x-y)+err(x,y),
\end{equation}
where the error term $err(x,y)$ can be well controlled. The kernels $1_A(x)1_A(y)g(x-y)$ then detect only interactions between points $x,y\in A$ and therefore, provided that the error term is well-behaved, a decomposition like \eqref{lbrough1} allows to reduce a study of our energies over the whole $\mathbb R^d$ to studies done ``locally'', separately on each $A\in F$. In fact in the concrete situations we face we will rather consider decompositions of the following averaged form, which is slightly more complicated than \eqref{lbrough1}:
\begin{equation}\label{lbrough2}
g(x-y) = \int_\Omega \left(\sum_{A\in F_\omega}1_A(x)1_A(y)g(x-y)\right) d\mathbb{P}(\omega) + err(x,y),
\end{equation}
where $(\Omega, \mathbb{P})$ is a probability space, and for each $\omega\in \Omega$ we decompose $g$ along a separate packing $F_\omega$. To have that $\mathbb{P}$ is a probability measure is not essential, but it simplifies our situation, because it implies that bounds done separately for each $\omega\in\Omega$ directly give the same bound for the integral. \textit{Crucial} in our arguments will be to construct a decomposition such that both $err(x,y)$ and the integral term are positive (semi)-definite. This will allow us to cancel the mean field term in our calculations, for which a weaker form of positive definiteness is both a necessary and sufficient assumption, as shown in \cite{Petr15}.
\par To pass to more concrete calculations, we look at periodic packings which are all related to a basic one by isometries. Consider again a lattice $\mathcal L$ (which below will always be $\mathcal L=\left(\ell \mathbb Z\right)^d$ for some suitably chosen $\ell>0$) with fundamental domain $\Omega_{\mathcal L}$ as in the previous section, and suppose now that $A_1,\ldots,A_n$ are pairwise disjoint sets such that $A_i\subset \Omega_{\mathcal L}, 1\le i\le n$. Observe that the following is a packing of $\mathbb R^d$ by copies of these sets:
\begin{equation}\label{deff}
F:=\{A_i-p :\ p\in\mathcal L,\ 1\le i\le n\}.
\end{equation}
If we particularize the formula \eqref{averagerot} to the case $k=2$ and $f(x_1,x_2)=1_{A_i}(x_1)1_{A_i}(x_2)$, we suppose that {$\Omega_{\mathcal L}$ is symmetric so that} $\Omega_{\mathcal L}=-\Omega_{\mathcal L}$, and we sum over $i=1,\ldots,n$, then keeping in mind \eqref{deff}, we find: 
\begin{eqnarray}
\sum_{i=1}^n\langle 1_{A_i}(x_1)1_{A_i}(x_2)\rangle &\stackrel{\text{\eqref{averagerot}}}{=}&\sum_{i=1}^n\int_{\Omega_{\mathcal L}}\sum_{p\in \mathcal L}(\tau_y1_{A_i})(x_1+y)(\tau_y1_{A_i})(x_2+y) dy\nonumber\\
&\stackrel{\text{\eqref{deff}}}{=}&\int_{\Omega_{\mathcal L}}\sum_{A\in F}1_{A-y}(x_1)1_{A-y}(x_2)dy= \int_{\Omega_{\mathcal L}}\sum_{A\in F+y}1_{A}(x_1)1_{A}(x_2)dy\label{basicformula}
\end{eqnarray}
where {$F+y=\{A+y: \ A\in F\}$.} Note that if the $A_i$ are balls of the form $B_{r_i}(y_i)$ with $M$ possible values $r_i\in\{R_1,\ldots,R_M\}$ then, due to the formula \eqref{caseaisball} and grouping together the $n$ balls into families having same radius $R_i$ we find, with $Y_i:=\{y\in \mathbb R^d: B_r(y)\in F, r=R_i\}$:
\begin{equation}\label{localizebasic}
\begin{aligned}
\frac{1}{|\Omega_{\mathcal L}|}\int_{\Omega_{\mathcal L}}
&\sum_{A\in F+y}1_{A}(x_1)1_{A}(x_2)g(x_1-x_2)dy=\frac{1}{|\Omega_{\mathcal L}|}\sum_{i=1}^n\langle 1_{A_i}(x_1)1_{A_i}(x_2)g(x_1-x_2)\rangle \\
&=g(x_1-x_2)\frac{1}{|\Omega_{\mathcal L}|}\sum_{i=1}^M\sum_{y\in Y_i}\langle 1_{B_{R_i}(y)}(x_1)1_{B_{R_i}(y)}(x_2)\rangle\\
&\stackrel{\text{\eqref{caseaisball}}}{=} g(x_1-x_2)\sum_{i=1}^M\frac{\#Y_i}{|\Omega_{\mathcal L}|}1_{B_{R_i}}*1_{B_{R_i}}(x_1-x_2)=g(x_1-x_2)\sum_{i=1}^M c_i\frac{1_{B_{R_i}}*1_{B_{R_i}}}{|B_{R_i}|}(x_1-x_2)\ ,
\end{aligned}
\end{equation}
where $c_i$ is the proportion of $\Omega_{\mathcal L}$ covered by those balls in $F$ that have radius $R_i$, i.e.
\begin{equation}\label{coeffballpack}
c_i:=\frac{\# Y_i |B_{R_i}|}{|\Omega_{\mathcal L}|}=\frac{\left|\bigcup_{x\in Y_i}B_{R_i}(x)\right|}{|\Omega_{\mathcal L}|}\ .
\end{equation}
Following \eqref{basicformula} and its link to \eqref{lbrough2}, we see that in order to study the error introduced by the localization, in the case $g(x)=|x|^{-s}$ one faces the problem of how to bound a contribution of the form
\begin{equation}\label{kernelerror}
\frac{1}{|x_1-x_2|^s}\left(1 - \sum_{i=1}^n c_i \langle 1_{A_i}(x_1)1_{A_i}(x_2)\rangle\right)\ ,
\end{equation}
where $c_i\in\mathbb R$ are coefficients chosen depending on the precise details of the strategy that one follows. The two main approaches that proved successful in order do this are as follows:
\begin{itemize}
\item Fefferman \cite{Feff85} for $d=3,s=1$ (later extended in \cite{Gregg89} to $0<s<3$) used formulas based on \eqref{localizebasic}. At the same time the kernel error analogous to \eqref{kernelerror} is positive definite, and it has a bound by $M^{-1}|x_1-x_2|^{-s} + O(M^{-2})$, where $M$ is the total number of scales used in the technical part of the construction. The $A_i$'s are here chosen to form a so-called ``Swiss cheese'' decomposition of $\Omega_{\mathcal L}$. One of the first appearances of this type of decomposition seems to be in \cite[Sec. III]{Lieblebowitz}, but see also the references therein. The parameter $M$ appearing in the bound denotes the total number of scales used in this packing.
\item Graf and Schenker \cite{Grafschenker95} (see also \cite{HLSI09}, \cite{HLSII09}) consider the situation where the $A_i$ form a partition of $\Omega_{\mathcal L}$, and are all isometric to a fixed simplex, in which case the sum in \eqref{kernelerror} reduces to the case where only a single term appears, as the quantity $h_A$ is invariant under isometry. In this case, it becomes crucial that at the same time \eqref{kernelerror} is positive definite and bounded by a value that decreases to zero as the scale of $A$ decreases to zero. Unfortunately these conditions seem to hold only in the case $s\le1,d=3$ and are explicitly proven in \cite{Grafschenker95} only for $s=1,d=3$. On the other hand, the above Fefferman method will be extended here to all $d$ and all $0<s<d$.
\end{itemize}
The main structural differences between our setting and that of \cite{Feff85} and \cite{Gregg89} are the following:
\begin{itemize}
\item The role of positive definiteness of $g$ is more stringent here than in \cite{Feff85, Gregg89}, because  the mean field $\frac{\mu\otimes\mu}{\mu^2(\mathbb R^d)}$ is a minimizer of $F_{\infty,g}^\mathrm{OT}$ only in that case, by \cite{Petr15}. We don't know how to explicitly describe the minimizers of $F_{\infty,g}^\mathrm{OT}(\mu)$ for $g$ not positive definite. The energy $E_{N,g}^\mathrm{xc}(\mu) =F_{N,g}^\mathrm{OT}(\mu) - \frac{N(N-1)}{2}F_{\infty,g}^\mathrm{OT}(\mu)$ is thus giving actual next-order terms only if $g$ is positive definite: thus only by having positive definite error terms in our kernel decompositions can we expect the first order effects to be cancelled by the associated mean fields. In this way the result of \cite{Petr15} is a crucial ingredient to keep in mind in our strategy.
\item The competitors for $F_{\infty,g}^\mathrm{OT}(\mu)$ are automatically competitors for $F_N^\mathrm{OT}(\mu)$: therefore the next-order energy $E_{N,g}^\mathrm{xc}$ is negative, and we look for a lower bound for it. On the contrary, \cite{Feff85, Gregg89} had a positive energy, and were interested in an upper bound for it.
\end{itemize}
\subsubsection{The packing lemma and the construction of $\Omega_\ell, \mathbb P_\ell, F_\omega^\ell$}\label{sssecpacking}
In order to be more specific about the form that \eqref{kernelerror} will have, we recall the Swiss cheese lemma, first introduced in \cite{Lieblebowitz} and used in \cite{Feff85, Hughes89, Gregg89}, for $d=3$. The main idea therein was to decompose regions in space into sets of disjoint balls with geometrically increasing radii. We need to adapt the computations from \cite{Hughes89} to the case of general dimension $d$. In the previous works the authors work at scales $R_1>1$, whereas here we only require $R_1>0$, since the argument in \cite{Hughes89} easily adapts to this more general setting.
\begin{lemma}[``Swiss cheese'' lemma in general dimension]\label{cheeselemma}
For each $d\ge 2$ define $C_d:=\frac{2^{d+1}}{|B_1|}$, where $B_1$ is the unit ball centred at zero $\{x\in \mathbb R^d: |x|\le 1\}$. Consider a sequence of real numbers $0<R_1<\cdots<R_M$ such that $R_{k+1}>(1+4\sqrt d |B_1|) R_k$ for all $k$. Assume that $Q$ is a cube of side length $l> 8\sqrt d |B_1|(M+C_d) R_M$. Then there exists a family $\mathcal B$ of disjoint balls $B_r(x)\subset Q,$ with $x\in Q$, such that 
\begin{itemize}
\item The balls in $\mathcal B$ have radii $r\in \{R_1,\ldots,R_M\}$.
\item For each $i=1,\ldots, M$, if 
\begin{equation*}
Y_i:=\left\{x\in Q:\ \exists B_r(x)\in\mathcal B, \, r=R_i\right\}\ ,
\end{equation*}
then there holds
\begin{equation}\label{cheesebound}
\frac{1}{M+C_d+1}<c_i:={\frac{\left|\bigcup_{x\in Y_i}B_{R_i}(x)\right|}{|Q|}}<\frac{1}{M+C_d}\ .
\end{equation}
\end{itemize}
\end{lemma}
The proof of the lemma follows the same ideas as in \cite[Sec. 4]{Gregg89} or in \cite[Sec. 4.3]{Hughes89}, namely it is based on packing on a cube, then periodized, and it proceeds by induction on $i$. We refer to \cite{Gregg89} for the details.

\par We will use the ball packing given by Lemma \ref{cheeselemma} for the case $Q=\Omega_{\mathcal L}=[-\tfrac{l}2, \tfrac{l}2]^d$ in \eqref{localizebasic}, and then we define a packing of the whole of $\mathbb R^d$ by extending by $\mathcal L$-periodicity, for $\mathcal L=\left(\ell\mathbb Z\right)^d$. We use a special case of the definition \eqref{deff} for the case $\{A_1,\ldots,A_n\}=\mathcal B$. The dependence on $\mathcal L$ is not recorded in the notation for the sake of lightness, and we define
\begin{equation}\label{defff}
F_{\mathcal B}:=\{B - p: B\in\mathcal B, p\in\mathcal L\}\ .
\end{equation}
We note that the formulas \eqref{cheesebound} and \eqref{coeffballpack} giving the coefficients $c_i$ agree, for the choice of packing family $\mathcal B$ as in Lemma \ref{cheeselemma}.
\par However the basic formula \eqref{localizebasic}, in which $1_{B_{R_i}}*1_{B_{R_i}}(x)$ appears, cannot be directly used, for regularity reasons. Indeed, the relevant positive definiteness criterion, formulated in Corollary \ref{pdcriterion}, requires a control on derivatives up to order $d$ of our kernel approximants (see \eqref{boundgregg}). We consider first the expression
\begin{equation}\label{convolballs}
1_{B_r}*1_{B_r}(x) = \left\{\begin{array}{ll}\frac{\pi^{\frac{d-1}{2}}}{2^{d-1}\Gamma\left(\frac{d+1}{2}\right)}\int_{|x|}^{2r}(4r^2 - y^2)^{\frac{d-1}{2}}dy&\text{ for }|x|\le 2r\ ,\\ 0&\text{ for }|x|>2r\ ,\end{array}\right. 
\end{equation}
The above convolutions of balls have only $1+[(d-1)/2]$ continuous derivatives, a fact which forbids the desired bounds. Similarly to the procedure in \cite{Gregg89}, we therefore apply a mollification to \eqref{convolballs}. Note that the mollification in \cite{Gregg89} does not allow full regularity: the averaging over the radii allows only to add the control on one more extra derivative to the above-mentioned $1+[(d-1)/2]$ derivatives of $1_{B_r}*1_{B_r}$. While this is enough to cover the range $0<s<3$ in dimension $d=3$, this is unsatisfactory for general $d$ and therefore we proceed differently than \cite{Gregg89}, with a reasoning that covers all the cases $0<s<d$ at once.

\par To start with, fix a small parameter $\eta$, say $\eta\in (0,1/2]$, a choice that will play no special role in the computations below. For $t\in[1-\eta,1+\eta]$, consider a positive function $\rho_\eta\in C_0^\infty([1-\eta,1+\eta])$ such that $\int\rho_\eta(t) dt=1$, and set
\begin{equation}\label{mollifballs}
Q_{0,\eta}(x):=\int_{1-\eta}^{1+\eta}\frac{\left(1_{B_t}*1_{B_t}\right)(x)}{|B_t|}\rho_\eta(t)dt\ .
\end{equation}
Then $Q_{0,\eta}$ is smooth outside the origin, because it is radial and in radial coordinates we may use the smoothness of $\rho_\eta$ to control the radial partial derivatives. Also note that $Q_{0,\eta}(0)=1$ because $1_{B_t}*1_{B_t}(0)=|B_t|$ and $\int\rho_\eta(t)dt=1$. We then define, for $i=1,\ldots,M$, and $R_i$ as in Lemma \ref{cheeselemma} (but note that here $B_{tR_i}$ is for each $i$ the ball of radius $tR_i$ centered at $0$)
\begin{equation}\label{qieta}
Q_{i,\eta}(x) :=Q_{0,\eta}\left(\frac{x}{R_i}\right) = \int_{1-\eta}^{1+\eta}\frac{\left(1_{B_{tR_i}}*1_{B_{tR_i}}\right)(x)}{|B_{tR_i}|}\rho_\eta(t)dt\ .
\end{equation}
As a consequence of the bounds and support properties of the above integrands, for all $x\in\mathbb R^d$ we have
\begin{equation}\label{sptqieta}
Q_{i,\eta}(x)\le 1_{B_{2(1+\eta)R_i}}(x)\ .
\end{equation}
\par Lemma \ref{cheeselemma} with the above choices of radii $R_i, 1\le i\le M$, and $Q=[-\ell/2,\ell/2]^d$, produces a family of disjoint balls satisfying in particular estimate \eqref{cheesebound}, which for $\mathcal L:=\left(\ell \mathbb Z\right)^d$ is then extended by $\mathcal L$-periodicity to a disjoint packing $F$ as in \eqref{cheesebound}.

\par For $t\in[1-\eta, 1+\eta]$, we consider the dilation $t\mathcal B$ by $t$ of the packing $\mathcal B$, and the dilations $t\mathcal L, tQ$ of the lattice $\mathcal L$ and of the cube $Q$. Moreover if $B=B_r(x)$ is a ball in $\mathbb R^d$ then we denote $tB:=B_{tr}(tx)$. Then
\begin{equation*}\label{deffr}
t\mathcal B:=\{tB: B\in \mathcal B\}\qquad t\mathcal L:=\{tp:\ p\in\mathcal L\}\ ,\qquad tQ:=[-t\ell/2,t\ell/2]^d=\Omega_{t\mathcal L}\ .
\end{equation*}
Then we claim that the covered volume ratios $c_i$ appearing in \eqref{coeffballpack} and in \eqref{cheesebound} do not depend on the choice of $t>0$. Indeed, the packing $t\mathcal B$ obtained after the dilation by $t$ is still made of balls, this time with set of radii $tR_1,\ldots,tR_M$, which now cover $tQ=t\Omega_{\mathcal L}$. Moreover the dilation preserves the property of the packing of being made of disjoint balls, and for $t\mathcal B$ the ball center sets $Y_i$ are replaced by new center sets
\begin{eqnarray*}
tY_i:&=&\{tx: x\in Y_i\}=\{tx\in\mathbb R^d: \exists B_r(x)\in \mathcal B, r= R_i\}=\{tx\in\mathbb R^d: \exists B_r(tx)\in t\mathcal B, r= tR_i\}\nonumber\\
&=&\{x\in\mathbb R^d: \exists B_r(x)\in t\mathcal B, r= tR_i\}\ .
\end{eqnarray*}
After the above replacements, the values $c_i$ appearing in the formula \eqref{cheesebound} don't change:
\begin{equation*}\label{remainthesame}
\frac{\left|\bigcup_{x\in tY_i}B_{tR_i}(tx)\right|}{|tQ|}=\frac{\#(tY_i)\left|B_{tR_i}\right|}{|tQ|} =\frac{\#(Y_i)t^d\left|B_{R_i}\right|}{t^d|Q|}=\frac{\#(Y_i)\left|B_{R_i}\right|}{|Q|}=\frac{\left|\bigcup_{x\in Y_i}B_{R_i}(x)\right|}{|Q|}\ .
\end{equation*}
Then we define, analogously to \eqref{defff},
\begin{equation*}\label{deffff}
F_{t\mathcal B}:=\{B - p:\ B\in t\mathcal B, p\in t\mathcal L\} = \{tB - tp:\ B\in \mathcal B, p\in \mathcal L\}=tF_{\mathcal B}\ .
\end{equation*}
We now use formula \eqref{localizebasic} for the ball covers corresponding to the packing of Lemma \ref{cheeselemma} for the choice $Q=\Omega_{\mathcal L}$, then we superpose dilated versions of these packings in order to obtain a contribution like \eqref{qieta} in place of the contribution of each $B_{R_i}$ in \eqref{localizebasic}. We find the following formula, where we use the fact that $\Omega_{\mathcal L}=[-\ell/2,\ell/2]^d, |\Omega_{\mathcal L}|=l^d$ and $t\Omega_{\mathcal L} = [-t\ell/2,t\ell/2]^d, |t\Omega_\mathcal L|=(tl)^d$:
\begin{equation}\label{decompqi}
  \sum_{i=1}^Mg(x_1-x_2)c_i Q_{i,\eta}(x_1-x_2) = \int_{1-\eta}^{1+\eta}\left(\frac{1}{(lt)^d}\int_{[-\frac{lt}{2},\frac{lt}{2}]^d}\sum_{B\in tF_{\mathcal B}+y}1_B(x_1)1_B(x_2)g(x_1-x_2)dy\right)\rho_\eta(t)dt\ .
\end{equation}
We observe that the left hand side of \eqref{decompqi} is of the form that appears in \eqref{lbrough2}, namely
\begin{lemma}
We have
\begin{equation}\label{lbroughqi}
\sum_{i=1}^Mg(x_1-x_2)c_i Q_{i,\eta}(x_1-x_2)=\int_\Omega \left(\sum_{A\in F_\omega}1_A(x_1)1_A(x_2){g(x_1-x_2)}\right) d\,\mathbb{P}_l(\omega)\ ,
\end{equation}
for the choices 
\begin{equation}\label{lbrough3}
\begin{aligned}
\Omega =\Omega_l &:= \left\{(t,y)\in [1-\eta,1+\eta]\times\mathbb R^d:\ y\in\left[-\frac{lt}{2}, \frac{lt}{2}\right]^d\right\},\\
{d\, \mathbb P(t,y)=}d\,\mathbb{P}_l(t,y)&:=\frac{\rho_\eta(t)}{(lt)^d}1_{\Omega_l}(t,y)\ dt\ dy\ ,
\end{aligned}
\end{equation}
and to each $\omega:=(t,y)\in\Omega_l$ we associate an isometry $y$ and a packing $F_\omega$ given by
\begin{equation}\label{lbrough4}
F_\omega =F_\omega^l :=tF_{\mathcal B}+y=\{tA+y:\ A\in F_{\mathcal B}\}\ .
\end{equation}
\end{lemma}
In the above we use the notation \eqref{defgrx} regarding the definition and action of isometries, while the parameterization of isometries by $y$, and of dilations by $t$ is encoded via the set $\Omega_l$. Finally note that the measure $\mathbb P_l$ is defined precisely so as to equate \eqref{decompqi} to \eqref{lbroughqi}, and $l, \mathcal B$ are as defined from Lemma \ref{cheeselemma}.

\par Note that differently than in Gregg's work \cite{Gregg89}, here we fix once and for all the packing $\mathcal B$ depending only on the choices of $R_1,\ldots, R_M, l$ from that lemma, and do not need to produce a new packing separately anew for each dilation $t$. It is obvious that if $\mathcal B$ is a packing by disjoint balls then $g(t\mathcal B)$ is still a disjoint packing for all isometries $g$ and all dilations $t$.
\begin{rmk}[the case of simplices]
We may directly define $Q_{A,\eta}$ as in \eqref{qieta} for the case of more general sets $A$ rather than balls $B_{R_i}$, by replacing $B_r$ by $t A:=\{tx:\ x\in A\}$. In particular this can be done for the case of simplices $A=\triangle$ if we include also an averaging over $SO(d)$ when defining $\langle f\rangle$, as done in \cite{Grafschenker95}, where the singularity for the function $h_\triangle$ appearing therein, has the same kind of discontinuity as $1_B*1_B$. The packing from Lemma \ref{cheeselemma} could then be replaced by a tiling, again by families of simplices of sizes decreasing like a geometric series. However, it seems hard to prove positive definiteness and boundedness at zero of the error terms based on the less explicit formula for $h_\triangle$, outside the case $s=1,d=3$.
\par It would be interesting to have a more thorough investigation of the influence of general packing strategies on the study of positive definiteness, but that would go beyond the scope of the present work.
\end{rmk}

\subsubsection{Control of the error terms in the localization estimate}\label{ssseccontrolerror}
\textbf{Proof of Proposition \ref{prop3ws}}

\par With the decompositions from the previous section we come back to the study of the final formula of type \eqref{lbrough2} which will appear in our proof. We consider now only the case $c(x,y)=|x-y|^{-s}, 0<s<d$. We use the formula \eqref{lbroughqi} with choices \eqref{lbrough3}, and in this case we write, following the formalism \eqref{lbrough2}, with $\Omega=\Omega_l,\mathbb P=\mathbb P_l,F_\omega=F_\omega^l$ as in \eqref{lbrough3}, \eqref{lbrough4}
\begin{equation*}
err(x_1,x_2):= \frac{1}{|x_1-x_2|^s} - \int_{\Omega_l}\left(\sum_{A\in F_\omega^l}\frac{1_A(x_1)1_A(x_2)}{|x_1-x_2|^s}\right)d\mathbb P_l(\omega)\ .
\end{equation*}
In fact, we will find it useful in our computations to work with a slightly different error term from the one above, for which we can show both positive-definiteness, and also a rough next-order lower bound needed in our proof. More precisely, we introduce for $0<\epsilon<d/2$ and $\epsilon\le s\le d-\epsilon$
\begin{equation}
 \label{errorw}
 w(x_1-x_2):=\left(1+\frac{C}{M}\right)|x_1-x_2|^{-s}- \int_{\Omega_l}\left(\sum_{A\in F_\omega^l}\frac{1_A(x_1)1_A(x_2)}{|x_1-x_2|^s}\right)d\mathbb P_l(\omega)\ ,
 \end{equation}
with $C=C(\rho_\eta,d,\epsilon)>0$ a constant depending only on the choice of $\rho_\eta$, on $d$ and on $\epsilon$, the form of which will be made explicit in Lemma \ref{derivativebounds} in the Appendix.

\par Proposition \ref{prop3ws} becomes an immediate consequence of \eqref{lbrough2}, due to Lemma \ref{prop3wsa} which proves the positivity and boundedness properties of $w$, and due to Lemma \ref{prop3wsb} which proves asymptotic lower bounds for $w$.

\qed
\begin{rmk}[continuation of Remarks \ref{morecosts} and \ref{contmorec1}]\label{contmorec2}
We briefly mention here how the construction of the decomposition of Proposition \ref{prop3ws} should be adapted for the examples from Remark \ref{morecosts}.
\begin{itemize}
 \item [(a)] For the case of a general compact $d$-dimensional Riemannian manifold $(M,g)$ the main difference to our setting is that we don't have a group action by translations. However in this case we can still adapt Lemma \ref{cheeselemma} and find a good packing by metric balls of scale $l>0$. Then for $l\ll 1$ we may use the fact that at small scales near each $x_0\in M$, the manifold  $(M,g)$ is closer and closer to $(\mathbb R^d,g(x_0))$. This allows to reduce the packing construction to one on $\mathbb R^d$ and use translation averaging locally, up to an error depending on the modulus of continuity of $g$. Regarding Lemma \ref{prop3wsa}, in this case it can be directly implemented on manifolds, by using the decomposition in dyadic scales and the same type of positive definiteness criteria as in the appendix. Lemma \ref{prop3wsb} is provable by a similar localization argument, reducing it to a problem on $(\mathbb R^d, g(x_0))$ up to small error, and all the tools \cite{HS}, \cite{LundNamPort} used in the appendix can be extended to this case.
 \item [(b), (c)] For the non-homogeneous kernels as in these cases, the packing arguments go precisely as in the present section, allowing to obtain a kernel decomposition. Concerning Lemma \ref{prop3wsa}, we may still use the positive definiteness criteria as in Lemma \ref{fourierbdlemma}, by replacing the smoothed kernels $Q_{i,\eta}$ with non-isotropic analogues. For adapting Lemma \ref{prop3wsb} the main difference is that the proof of a rough lower bound as in Lemma \ref{prop3wsb} cannot be directly done by using the result of \cite{HS}. This however is a robust result, and in our case we could for example apply an ad-hoc version of \cite{HS} which uses decompositions of the product-type costs from (b), (c) which use non-isotropic elementary kernels modelled on the kernels at hand, rather than the isotropic $1_{B_r}*1_{B_r}$ as in \cite{HS}.
\end{itemize}
We also note that in all cases (a), (b), (c), the kernels $\mathsf{c}(x,y)=g(x-y)$ under consideration have asymptotic homogeneity $-s$ near the singularity at $x=y$ (i.e. there exists a function $g_0(x):\R^d\to \R$ such that $g_0(\lambda x)=\lambda^{-s}g_0(x)$ and $\lim_{x\to0}g(x)/g_0(x) = 1$), and this property can be quantified in order to allow the rough lower bound formula \eqref{roughboundgeneral}. 
\end{rmk}
\subsection{Optimal lower bound for Coulomb and Riesz costs by Fefferman decomposition}\label{sseclowerboundssecfeffermandecompositiontion}
Before proving Proposition \ref{mixextlowb}, we will need a number of helpful lemmas and corollaries.

\par Recall that we denoted for any cost function $\mathsf{c}:\mathbb{R}^d\times\mathbb{R}^d\rightarrow\mathbb{R}_{>0}\cup \{+\infty\}$ and for any $N\ge 2$ by 
\begin{equation*}
F_{N, \mathsf{c}}^\mathrm{OT}(\mu)=\inf\left\{\int_{(\mathbb{R}^d)^N}\sum_{i,j=1,i\neq j}^N\mathsf{c}(x_i,x_j)d \gamma_N(x_1,\ldots,x_N):\ \gamma_N\in\mathcal P^N_{sym}(\mathbb{R}^d),\gamma_N\mapsto \mu\right\}\ ,
\end{equation*}
and $F_{0,\mathsf{c}}^\mathrm{OT}(\mu)=F_{1,\mathsf{c}}^\mathrm{OT}(\mu)=0$, with a corresponding formula for $E_{N, \mathsf{c}}^\mathrm{xc}(\mu)$. Moreover, for the particular case of the cost $\mathsf{c}(x,y)=|x-y|^{-s}$ we use the notations $F_{N,s}^\mathrm{OT}(\mu)$ and $E_{N,s}^\mathrm{xc}(\mu)$. Note that for the purposes of the present section we don't need to impose that $\mathsf{c}$ be positive, and this is imposed just for consistency with the rest of the paper.
\subsubsection{Splitting the cost}\label{sssecsplittingthecost}
The following lemma uses in the Optimal Transport framework the ``grand-canonical'' formalism of \cite{LewLiebSeir17},  and will help us obtain our main result for very general densities, by means of our strategy involving the use of measures with piecewise constant density. We thank M. Lewin and S. Di Marino for pointing out at the end of May {2017 in Jyv\"askyl\"a} a flaw in the corresponding "splitting-the-cost" lemma which we had used in the preliminary version of the paper. In order to settle this flaw it turned out to be sufficient, as suggested by M. Lewin in Jyv\"askyl\"a, to use the functional $F_{\mathrm{GC},N,\mathsf{c}}^\mathrm{OT}$ allowing fluctuations in the number of marginals, rather than using $F_{N,\mathsf{c}}^\mathrm{OT}$ throughout the proof. We also introduce in \eqref{OT_relax1} the auxiliary minimization problem $F^\mathrm{OT}_{\mathrm{GCB},N,\mathsf{c}}$ in which we impose an upper bound on these fluctuations.

To the best of our knowledge, the results in Lemma \ref{subadd12gena01}, Lemma \ref{subadd_gcb} and Corollary \ref{subadd12gen11} below are new and of independent interest.

\par We will find it useful in some of the proofs below to work with another quantity, $F_{\mathrm{GCB},M,\mathsf{c}}^\mathrm{OT}(\mu,\bar M)$, due to its finite summation structure.  That is, for all $ M\in\mathbb{R}_{>0},\bar M\in \mathbb{N}_+$ such that $\bar M\ge M$, and for $\mu\in\mathcal{P}(\mathbb{R}^d)$, let

\begin{subequations}\label{OT_relax1}
\begin{equation}\label{OTGCB}
  F_{\mathrm{GCB},M,\mathsf{c}}^\mathrm{OT}\left(\mu, \bar M\right):=\inf \left\{\sum_{n=2}^{\bar M}\alpha_n F_{n,\mathsf{c}}^\mathrm{OT}(\mu_n)\left|\begin{array}{c}\sum_{n=0}^{\bar M}\alpha_n=1,\ \sum_{1\le n\le\bar M}n \alpha_n\mu_n = M\mu,\\[3mm]\mu_n\in {\cal P}(\mathbb{R}^d),\ \alpha_n\ge 0,\quad\mbox{for}\quad n=0,1,\ldots,\bar M\end{array}\right.\right\}\ ,
\end{equation}
or equivalently, as summing zero contributions has no effect on the sums,
\begin{equation}\label{OTGCBbis}
  F_{\mathrm{GCB},M,\mathsf{c}}^\mathrm{OT}\left(\mu, \bar M\right)=\inf \left\{\sum_{n=2}^\infty\alpha_n F_{n,\mathsf{c}}^\mathrm{OT}(\mu_n)\left|\begin{array}{c}\sum_{n=0}^\infty\alpha_n=1,\ \sum_{n=1}^\infty n \alpha_n\mu_n = M\mu,\\[3mm]\mu_n\in {\cal P}(\mathbb{R}^d),\ \alpha_n\ge 0,\quad\mbox{for}\quad n=0,1,\ldots,\bar M\\[3mm] \alpha_n=0,\quad\mbox{for}\quad n>\bar M\end{array}\right.\right\}\,.
\end{equation}
We then define
\begin{equation}
\label{ExcGCB}
E_{\mathrm{GCB},M,\mathsf{c}}^\mathrm{xc}\left(\mu, {\bar M}\right):=F_{\mathrm{GCB},M,\mathsf{c}}^\mathrm{OT}\left(\mu, {\bar M}\right)-M^2 \int_{{\mathbb{R}^d}\times {\mathbb{R}^d}}\mathsf{c}(x,y)d\mu(x)d\mu(y).
\end{equation}
\end{subequations}
\begin{rmk}
Definition \eqref{OTGCB} says that there are $M$ marginals ``on average'', and the total number of marginals usable in this decomposition is bounded above by $\bar M$. Removing this last constraint we defined $F_{\mathrm{GC},M,\mathsf{c}}^\mathrm{OT}(\mu)$ in \eqref{OTGC}, which is the precise analogue of \cite[eq. (3.1)]{LewLiebSeir17}. For the $GCB$-problem in \eqref{OTGCB}, as the possible number of marginals used is bounded by $\bar M$, the existence of minimizers can be proven like for the $F_{N,\mathsf{c}}^\mathrm{OT}(\mu)$-problems, based on the same techniques which work for the $N=2$ case (see \cite[Thm. 4.1]{Vil}), and therefore in case of lower semicontinuous cost $\mathsf{c}$ and $\mu$ as in the statement of Lemma \ref{subadd12gena01} the infimum is realized. The proof then extends to the $GC$-problem by using the fact that the measures given by $\lambda_k(\{n\}):=\lambda_{n,k}$ corresponding to the $GC$-optimizers form a tight sequence of probability measures on $\mathbb N$, as will be shown in detail in the proofs of Lemma \ref{gcotlim} and Lemma \ref{existgc}.
\end{rmk}
\begin{rmk}\label{rmk_gcb}
We note here the following special properties related to the problems \eqref{OT_relax} and \eqref{OT_relax1}:
\begin{enumerate}
\item For the case $\bar N=N\in\mathbb{N}_+, N\ge 2$, the relaxation \eqref{OTGCB} satisfies 
\begin{equation}
\label{relaxeqn}
F_{\mathrm{GCB},N,\mathsf{c}}^\mathrm{OT}(\mu,N)=F_{N,\mathsf{c}}^\mathrm{OT}(\mu),
\end{equation}
 as in this case the condition $\sum_{n=1}^Nn\alpha_n=N$ can be realized only if $\alpha_n=0$ for all $n<N$ and $\alpha_N=1$. The \eqref{relaxeqn} together with \eqref{OTGCB} imply in turn that for all $ \bar M,\bar N, N\in\mathbb{N}_+,\bar N\ge N\ge 2$, we have
\begin{equation}
\label{chainineq}
F_{\mathrm{GC},N,\mathsf{c}}^\mathrm{OT}\left(\mu\right)\le F_{\mathrm{GCB},N,\mathsf{c}}^\mathrm{OT}(\mu,\bar N+\bar M)\le F_{\mathrm{GCB},N,\mathsf{c}}^\mathrm{OT}(\mu,\bar N)\le F_{N,\mathsf{c}}^\mathrm{OT}\left(\mu\right).
\end{equation}
\item If $\mu=\tfrac{N'\mu'+N''\mu''}{N'+N''}$ where $\mu',\mu''$ are probability measures, and $N',N''\in\mathbb{R}_{>0}, N',N''>1$, then 
\begin{equation}\label{subadgcinit}
E_{\mathrm{GC},N'+N'',\mathsf{c}}^\mathrm{xc}(\mu)\le E_{\mathrm{GC},N',\mathsf{c}}^\mathrm{xc}(\mu')+E_{\mathrm{GC},N'',\mathsf{c}}^\mathrm{xc}(\mu'').
\end{equation}
Furthermore, assume that $\max(N',N'')>1$ and $\min(N',N'')\le 1$; for simplicity let $\max(N',N'')=N''$. Then
\begin{equation}\label{subadgcinit01}
E_{\mathrm{GC},N'+N'',\mathsf{c}}^\mathrm{xc}(\mu)\le E_{\mathrm{GC},N'',\mathsf{c}}^\mathrm{xc}(\mu'').
\end{equation}
Equation \eqref{subadgcinit} was stated in \cite[eqn. (3.2)]{LewLiebSeir17}; we provide a sketch of the proofs of \eqref{subadgcinit} and \eqref{subadgcinit01} in Appendix \ref{Remark4.9}.
\item We now consider the $\mathrm{GCB}$-problem, in the equivalent definition \eqref{OTGCBbis}. If $\mu=\tfrac{N'\mu'+N''\mu''}{N'+N''}$ where $\mu',\mu''$ are probability measures, and $N',N''\in\mathbb{R}_{>0},  \bar N',\bar N''\in\mathbb{N}_+, \bar N'\ge N'>1,\bar N''\ge N''>1$, then
\begin{equation}\label{subadd_gcbinit}
E_{\mathrm{GCB},N'+N'',\mathsf{c}}^\mathrm{xc}(\mu,\bar N'+\bar N'')\le E_{\mathrm{GCB},N',\mathsf{c}}^\mathrm{xc}(\mu',\bar N')+E_{\mathrm{GCB},N'',\mathsf{c}}^\mathrm{xc}(\mu'',\bar N'').
\end{equation}
Furthermore, assume that $\max(N',N'')>1$ and $\min(N',N'')\le 1$; for simplicity let $\max(N',N'')=N''$. Then
\begin{equation}\label{subadd_gcbinitle1}
E_{\mathrm{GCB},N'+N'',\mathsf{c}}^\mathrm{xc}(\mu,\bar N'+\bar N'')\le E_{\mathrm{GCB},N'',\mathsf{c}}^\mathrm{xc}(\mu'',\bar N'').
\end{equation}
The proofs of \eqref{subadd_gcbinit} and \eqref{subadd_gcbinitle1} can be found in Appendix \ref{Remark4.9}.
\item If $N\in\mathbb{R}_{>0}, \bar N\in\mathbb{N}_+, \bar N\ge N$ then
\begin{equation}\label{EGC_negative}
E_{\mathrm{GC},N,\mathsf{c}}^\mathrm{xc}(\mu)\le E_{\mathrm{GCB},N,\mathsf{c}}^\mathrm{xc}(\mu,\bar N)\le 0.
\end{equation}
For $N\in\mathbb{N}_+$ this follows from \eqref{chainineq} since $E_{N,\mathsf{c}}^\mathrm{xc}(\mu)
\le 0$. For the general $N\in\mathbb{R}_+$ case, if $N\le 1$ then we use our convention that $F_N^\mathrm{OT}=0$ for $N=0,1$, while for $N>1$ we apply \eqref{subadgcinit01} to $N'+N''=N, N''=[N], N'=N-[N]$, and $\mu'=\mu''=\mu$, to get $E_{\mathrm{GC},N,\mathsf{c}}^\mathrm{xc}(\mu)\le E_{\mathrm{GC},[N],\mathsf{c}}^\mathrm{xc}(\mu)\le 0$. A similar argument holds for $E_{\mathrm{GCB},N,\mathsf{c}}^\mathrm{xc}(\mu,\bar N)$.
\item For $\alpha>0$ let $\mu_\alpha$ be the measure given by $d\mu_\alpha(x)=\rho_\alpha(x)dx$ with $\rho_\alpha(x)=\rho(\alpha x)$. Then we have
\begin{equation}\label{scalinggc}
E_{\mathrm{GCB},N,s}^\mathrm{xc}(\alpha^d\mu_\alpha) = \alpha^{s}E_{\mathrm{GCB}, N,s}^\mathrm{xc}(\mu)\quad \mbox{and}\quad E_{\mathrm{GC},N,s}^\mathrm{xc}(\alpha^d\mu_\alpha) = \alpha^{s}E_{\mathrm{GC}, N,s}^\mathrm{xc}(\mu)\ .
\end{equation}
The proof follows similarly to the one of \eqref{scalingexc} from Lemma \ref{scalOT}, the key difference being that since the the optimizers $\mu_n$ are not assured to possess densities, one needs to use the change of measure 
$$
\alpha^d\mu_\alpha=F_\# \mu\quad\mbox{and}\quad\alpha^d\mu_{n,\alpha}= F_\#\mu_n,\quad\mbox{where}\quad F(x)=\alpha x,\quad x\in\mathbb{R}^d.
$$

\end{enumerate}
\end{rmk}
\begin{lemma}\label{subadd12gena01}
Let $\mathsf{c}:\mathbb{R}^d\times\mathbb{R}^d\rightarrow\mathbb{R}_{>0}\cup \{+\infty\}$. Take $\mu_1, \ldots, \mu_k\in\mathcal{P}(\mathbb{R}^d)$, with densities $\rho_1, \ldots,\rho_k,$ such that the quantities below are finite (e.g. if $\mathsf{c}(x,y)=|x-y|^{-s}, 0<s<d$, we can require $\rho_1,\ldots,\rho_k\in L^{\frac{2d}{2d-s},2}(\R^d)$), and assume that the $\mu_j$ have essentially disjoint supports, i.e. that $\mu_i(A)\mu_j(A)=0$ for every Borel set $A$ and for all $1\le i\neq j\le k$. Fix $M_1,\ldots, M_k\in\mathbb R_{>0}$ such that $\bar{M}:=\sum_{i=1}^k M_i\in\mathbb{N}_+,\bar M\ge 2,$ and denote by $\mathsf{c}_{ii}(x,y):=1_{\Lambda_i}(x)1_{\Lambda_i}(y)\mathsf{c}(x,y)$ for $i=1,\ldots,k$ and $x,y\in\mathbb R^d$. Let $\mu$ be the probability measure with density $\sum_{i=1}^kM_i \rho_i/\bar{M}$. 

Then the following holds: 
\begin{equation}\label{supadd1gen0eq}
F_{\bar{M},\sum_{i=1}^k \mathsf{c}_{ii}}^\mathrm{OT}\left(\mu\right)= F_{\bar{M}, \sum_{i=1}^k c_{ii}}^\mathrm{OT}\left(\frac{\sum_{i=1}^kM_i\mu_i}{\bar{M}}\right)\ge\sum_{i=1}^k F_{\mathrm{GCB},M_i,\mathsf{c}}^\mathrm{OT}\left(\mu_i, \bar M\right)\ge \sum_{i=1}^k F_{\mathrm{GC},M_i,\mathsf{c}}^\mathrm{OT}\left(\mu_i\right)\ .
\end{equation}
\end{lemma}
Before we prove the above lemma, note that
\begin{equation}\label{betterdual}
 F_{N,\mathsf{c}}^\mathrm{OT}(\mu) = \sup\left\{N\int f(x)d\mu(x)\left|\begin{array}{l}\text{for all }x_1,\ldots,x_N\in\rm{spt}(\mu),\\
 \text{there holds }\sum_{i=1}^Nf(x_i)\le \sum_{1\le i\neq j\le N}\mathsf{c}(x_i,x_j)\end{array}\right.\right\}\ .
\end{equation}
The above duality is a bit different than the usual version (see, for example, Theorem 5.9 from \cite{Vil}), due to the fact that we test only configurations that are all contained in the support of $\mu$ rather than general configurations in $(\mathbb R^d)^N$. This restriction can be applied by \textit{firstly} noting that any plan $\gamma_N$ such that $\gamma_N\mapsto\mu$ has support in $(\rm{spt}(\mu))^N$, therefore
\begin{eqnarray*}\label{betterprimal}
 F_{N,\mathsf{c}}^\mathrm{OT}(\mu)
 &=&\inf\left\{\int_{(\rm{spt}(\mu))^N}\sum_{1\le i\neq j\le N}\mathsf{c}(x_i,x_j)d\gamma_N(x_1,\ldots,x_N), \gamma_N\in\mathcal P_{sym}((\rm{spt}(\mu))^N), \ \gamma_N\mapsto\mu\right\}\ ,
\end{eqnarray*}
and \textit{secondly}, noting that the proof of duality in \cite{depascaleduality} carries through without changes once we replace the space $\mathbb R^d$ by the closed subset $\rm{spt}(\mu)\subset \mathbb R^d$, proving \eqref{betterdual}.

\par From \eqref{betterdual}, we immediately get
 \begin{equation}\label{optimnenew}
F_{M_i,\mathsf{c}}^\mathrm{OT}(\mu_i)=F_{M_i,\mathsf{c}_{ii}}^\mathrm{OT}(\mu_i)=F_{M_i,\sum_{j=1}^k\mathsf{c}_{jj}}^\mathrm{OT}(\mu_i)\ .
\end{equation}
\begin{proof}[Proof of Lemma \ref{subadd12gena01}:]
We consider below for simplicity of exposition just the case with $k=2$, and we denote $\Lambda_1=:A, \Lambda_2=:B$ and $\mathsf{c}_{11}=:\mathsf{c}_A, \mathsf{c}_{22}=:\mathsf{c}_B$, moreover we replace in the notation $M_1,M_2$ by $M,N$ and $\mu_1,\mu_2$ by $\mu,\nu$.

\medskip

\textbf{Step 1.} To begin with, we decompose the minimization problem defining $F_{\mathsf{c}_A+\mathsf{c}_B, M+N}^\mathrm{OT}\left(\tfrac{M\mu + N\nu}{M+N}\right)$ according to the cardinality of points in the configurations used in our transport plans, and belonging respectively to $A,B$, {in the spirit of \cite[Sec. 3]{LewLiebSeir17}}. Note that the Borel sets $\{B_{N_A, N_B}: N_A,N_B\in \mathbb N, N_A+N_B=M+N\}$ form a partition of such configurations, where 
\[
B_{N_A,N_B}:=\left\{(x_1,\ldots,x_{M+N})\in(A\cup B)^{M+N}\ \left|\begin{array}{l}\#\{i:\ 1\le i\le M+N:\ x_i\in A\}=N_A\\[3mm] \#\{i:\ 1\le i\le M+N:\ x_i\in B\}=N_B\end{array}\right.\right\}\ .
\]
These sets are symmetric under permutations of the $\R^d$-coordinates in $(\R^d)^N$. We observe here that there are $M+N+1$ sets $B_{N_A,N_B}$ such that $N_A+N_B=M+N$. As explained in more detail in Step 4 below, this fact will be crucial as to why the extra parameters $\alpha_0$ and $\alpha_1$ are introduced in \eqref{OTGCB} and \eqref{OTGC}.

Denote now the symmetric probability measures with fixed numbers of points in $A,B$ as follows:
\begin{equation}\label{pnanb}
 \mathcal P_{sym}((\mathbb R^d)^{N_A|N_B}):=\left\{\gamma_{N_A|N_B}\in \mathcal P_{sym}((\mathbb R^d)^{M+N}):\ \gamma_{N_A|N_B}(B_{N_A,N_B})=1\right\}\ .
\end{equation}
Note that for any $\gamma_{N_A|N_B}\in\mathcal P_{sym}((\mathbb R^d)^{N_A|N_B})$ its marginal is in the class
\begin{equation}\label{marginalnanb}
\mathcal P_{N_A|N_B}(\mathbb R^d):=\left\{\bar\mu\in\mathcal P(\mathbb R^d):\ \bar\mu(A)=\tfrac{N_A}{M+N},\ \bar\mu(B) =\tfrac{N_B}{M+N}\right\}\ .
\end{equation}
Thus for $\bar \mu\in\mathcal P_{N_A|N_B}(\mathbb R^d)$, the infimum in the following optimal transport problem is over a nonempty set:
\begin{equation*}
 F_{N_A|N_B,\mathsf{c}}^\mathrm{OT}(\bar \mu):= \inf\left\{\int_{(\mathbb R^d)^{M+N}}\sum_{1\le i\neq j\le M+N}\mathsf{c}(x_i,x_j) d\gamma_{N_A|N_B}(x_1,\ldots,x_{N_A+N_B})\ \left|\begin{array}{c}\gamma_{N_A|N_B}\in\mathcal P_{sym}((\mathbb R^d)^{N_A|N_B}),\\\gamma_{N_A|N_B}\mapsto\bar\mu\end{array}\right.\right\}\ .
\end{equation*}
By linearity of the marginal map $\gamma_{N_A|N_B}\mapsto \bar\mu$ given by $d\bar\mu(x):=\int d\gamma_{N_A|N_B}(x,x_2,\ldots,x_{M+N})$, we can generalize the above definition to marginals equal to general positive measures $\bar\mu$, such that for any $\alpha>0$ 
\begin{equation*}
 F_{N_A|N_B,\mathsf{c}_A+\mathsf{c}_B}^\mathrm{OT}(\alpha\bar\mu) = \alpha F_{N_A|N_B,\mathsf{c}_A+\mathsf{c}_B}^\mathrm{OT}(\bar\mu)\ .
\end{equation*}
Now we connect these definitions to the unconstrained problem. If $\gamma_{M+N}\in\mathcal P_{sym}((\R^d)^{M+N})$ is a symmetric transport plan with marginal $\mu_{M+N}$, then we can decompose
\begin{subequations}\label{decompose_gamma_mu}
\begin{equation}
\gamma_{M+N}=\sum_{\substack{N_A+N_B=M+N\\N_A,N_B\in\mathbb{N}}}\gamma_{M+N}|_{B_{N_A,N_B}}, \quad\mu_{M+N}=\sum_{\substack{N_A+N_B=M+N\\N_A,N_B\in\mathbb{N}}}\bar\mu_{N_A,N_B},
\end{equation}
where $\bar\mu_{N_A,N_B}$ is the marginal of the restriction $\gamma_{M+N}|_{B_{N_A,N_B}}$:
\begin{equation}
\bar\mu_{N_A,N_B}:=\int_{(\R^d)^{M+N-1}} d(\gamma_{M+N}|_{B_{N_A,N_B}})(\cdot,x_2,\ldots,x_{M+N})\ .
\end{equation}
\end{subequations}
Here $\gamma_{M+N}|_{B_{N_A,N_B}}$ are not probability measures, but they belong to the set $\mathcal M_{sym}^+((\R^d)^{N_A|N_B}):=\{\alpha\gamma_{N_A|N_B}:\ \gamma_{N_A|N_B}\in\mathcal P_{sym}((\R^d)^{N_A|N_B}),\alpha>0\}$ and similarly their marginals $\bar\mu_{N_A,N_B}$ belong to the set of measures $\mathcal M_{N_A|N_B}^+(\mathbb R^d):=\{\alpha\mu: \mu\in\mathcal P_{N_A|N_B}(\mathbb R^d), \alpha>0\}$. 

\medskip

\textbf{Step 2.} We now claim that
\begin{multline}\label{firstsplit}
 F_{M+N,\mathsf{c}}^\mathrm{OT}\left(\tfrac{M\mu + N\nu}{M+N}\right) \\=\min \left\{\sum_{\substack{N_A+N_B=M+N\\N_A,N_B\in\mathbb N}} F_{N_A|N_B,\mathsf{c}}^\mathrm{OT}(\mu_{N_A,N_B})\ 
 \left|\begin{array}{l}\mu_{N_A,N_B}\in \mathcal M_{N_A|N_B}^+(\mathbb R^d)\ ,\\[3mm] \sum_{\substack{N_A,N_B\in\mathbb{N}\\ N_A+N_B=M+N }}\mu_{N_A,N_B}=\frac{M\mu + N\nu}{M+N}\end{array}\right.\right\}\ .
\end{multline}
Indeed, the plans $\gamma_{M+N}|_{B_{N_A,N_B}}$ with marginals $\bar\mu_{N_A,N_B}$ as in the decomposition \eqref{decompose_gamma_mu} form a possible decomposition as in the minimum from \eqref{firstsplit}, and vice-versa, a set of minimizing measures $\mu_{N_A,N_B}$ on the right in \eqref{firstsplit} together with the optimal transport plans realizing $F_{N_A|N_B, \mathsf{c}}^\mathrm{OT}(\mu_{N_A,N_B})$ form by superposition a competitor to the minimization problem in $F_{M+N,\mathsf{c}}^\mathrm{OT}\left(\tfrac{M\mu + N\nu}{M+N}\right)$. By the linearity of the map
\begin{equation*}
 \gamma_{M+N}\mapsto \int_{(\mathbb R^d)^{M+N}}\sum_{1\le i\neq j\le M+N} \mathsf{c}(x_i,x_j) d\gamma_{M+N}(x_1,\ldots,x_{M+N})\ ,
\end{equation*}
we can then directly compare the costs, and conclude the proof of \eqref{firstsplit}.

\medskip

\textbf{Step 3.} If $\mu_{N_A,N_B}\in\mathcal P(\R^d)$ is the marginal of a plan $\gamma_{N_A|N_B}\in\mathcal P_{sym}((\R^d)^{N_A|N_B})$ with $N_A,N_B\neq 0$, then $\mu_{N_A,N_B}$ splits due to the definition \eqref{marginalnanb} as 
\begin{equation}\label{splitmuab}
\mu_{N_A,N_B}=\frac{N_A}{M+N}\mu_{N_A}+\frac{N_B}{M+N}\mu_{N_B},
\end{equation}
where $\mu_{N_A}\in\mathcal P(A), \mu_{N_B}\in\mathcal P(B)$. We then claim that 
%
%
\begin{equation}\label{split}
 F^\mathrm{OT}_{N_A|N_B,\mathsf{c}_A+\mathsf{c}_B}\left(\mu_{N_A,N_B}\right) =\left\{\begin{array}{l} F_{N_A,\mathsf{c}_A}^\mathrm{OT}(\mu_{N_A}) + F_{N_B,\mathsf{c}_B}^\mathrm{OT}(\mu_{N_B}),  \quad\mbox{if}\quad N_A,N_B\in\{2,\ldots, M+N-2\}\\[3mm]
 F_{N_A,\mathsf{c}_A}^\mathrm{OT}(\mu_{N_A}), \,\,\,\,\,\,\,\,\,\,\,\,\,\,\,\,\,\,\,\,\,\,\,\,\,\,\,\,\,\,\,\,\,\,\,\,\,\,\,\,\,\,\quad\mbox{if}\quad N_B\in\{0,1\}\\[3mm]
 F_{N_B,\mathsf{c}_B}^\mathrm{OT}(\mu_{N_B}),\,\,\,\,\,\,\,\,\,\,\,\,\,\,\,\,\,\,\,\,\,\,\,\,\,\,\,\,\,\,\,\,\,\,\,\,\,\,\,\,\,\,\quad\mbox{if}\quad N_A\in\{0,1\}.                              \end{array}\right.
\end{equation}
Indeed, define first the ``space of precisely split $(M+N)$-points configurations'' in $\mathbb R^d$ given by
\begin{equation}\label{defcnanb}
\mathcal C_{N_A|N_B}(\mathbb R^d):= \left\{C_{\vec x}:=\{x_1,\ldots,x_{M+N}\}\subset\mathbb R^d:\ \#(C_{\vec x}\cap A)=N_A, \#(C_{\vec x}\cap B)= N_B\right\}\ .
\end{equation}
Analogously to this, we define $\mathcal C_N(\mathbb R^d):=\{\{x_1,\ldots,x_N\}\subset \mathbb R^d\}$. Note that in all cases we consider the point configurations $\{x_1,\ldots,x_k\}$ as \emph{multisets}. 

\medskip 

Then, if under the isomorphism $\mathcal P_{sym}((\mathbb R^d)^{N_A|N_B})\simeq \mathcal P(\mathcal C_{N_A|N_B}(\mathbb R^d))$ induced by the map
\begin{equation*}
(\mathbb R^d)^{M+N}\ni\vec x:=(x_1,\ldots,x_{M+N})\mapsto \{x_1,\ldots,x_{M+N}\}=:C_{\vec x}\in\mathcal C_{M+N}(\mathbb R^d)\ ,
\end{equation*}
we identify a plan $\gamma_{N_A|N_B}$ having marginal $\mu_{N_A,N_B}$, to a measure $\bar \gamma_{N_A|N_B}\in \mathcal P(\mathcal C_{N_A|N_B}(\mathbb R^d))$, then we may take the operations $r_A:C_{\vec x}\mapsto C_{\vec x}\cap A$ and the similarly defined $r_B$ and consider the pushforward
\begin{equation}\label{gamma_na}
\gamma_{N_A}\simeq\bar \gamma_{N_A}:=(r_A)_\#\bar\gamma_{N_A|N_B}\in \mathcal P(\mathcal C_{N_A}(\mathbb R^d))\simeq\mathcal P_{sym}((\mathbb R^d)^{N_A})\ ,
\end{equation}
We next use the following lemma, which we prove later, and which translates the formula \eqref{gamma_na} more precisely. In particular, we can, and will, use \eqref{explicit_clumsy} as an equivalent definition of $\gamma_{N_A}$, respectively of $\gamma_{N_B}$, and we will not be  further using \eqref{gamma_na} below.
\begin{lemma}\label{lemma_na}
The measure $\gamma_{N_A}$ from \eqref{gamma_na} is given, in terms of duality with test functions $f\in C_b((\R^d)^{N_A})$, by the formula
\begin{multline}\label{explicit_clumsy}
\int f(x_1,\ldots,x_{N_A})d\gamma_{N_A}(x_1,\ldots,x_{N_A})\\= \frac{1}{N_A!N_B!}\sum_{\sigma\in S_{M+N}}\int f\left(x_{\sigma(1)},\ldots,x_{\sigma(N_A)}\right)\prod_{i=1}^{N_A}1_A\left(x_{\sigma(i)}\right)d\gamma_{N_A|N_B}(x_1,\ldots,x_{M+N}) \ .
\end{multline}
\end{lemma}
Next, we will use \eqref{explicit_clumsy} for a function $f(x_1)$ depending only on the first variable $x_1$. Note that in particular for any configuration $(x_1,\ldots,x_{M+N})\in B_{N_A,N_B}$ there exist \emph{precisely} $N_A!N_B!$ permutations $\sigma\in S_{M+N}$ such that $x_{\sigma(1)},\ldots,x_{\sigma(N_A)}\in A$, which shows that $\gamma_A$ is a probability measure. Then for each $\sigma\in S_{M+N}$ the set $\{\sigma(j): j=1,\ldots,M+N\}$ is in bijection with $\{1,\ldots,M+N\}$, and we have by the symmetry of $\gamma_{N_A|N_B}$ that 
\begin{multline}
\label{gammaNA}
\frac{1}{N_A!N_B!}\sum_{\sigma\in S_{M+N}}\int f\left(x_{\sigma(1)}\right)\prod_{i=1}^{N_A}1_A\left(x_{\sigma(i)}\right)d\gamma_{N_A|N_B}(x_1,\ldots,x_{M+N})\\
=\int \frac{1}{M+N}\sum_{j=1}^{M+N}\frac{1}{N_A!N_B!}\sum_{\substack{\sigma\in S_{M+N},\\ x_{\sigma(1)},\ldots,x_{\sigma(N_A)}\in A}}f(x_{\sigma(j)})d\gamma_{N_A|N_B}(x_1,\ldots,x_{M+N})\\
=\int \frac{1}{N_A!N_B!}\sum_{\substack{\sigma\in S_{M+N},\\ x_{\sigma(1)},\ldots,x_{\sigma(N_A)}\in A}}\frac{1}{M+N}\sum_{j=1}^{M+N}f(x_{\sigma(j)})d\gamma_{N_A|N_B}(x_1,\ldots,x_{M+N})\\
=\int \frac{1}{N_A!N_B!}\sum_{\substack{\sigma\in S_{M+N},\\ x_{\sigma(1)},\ldots,x_{\sigma(N_A)}\in A}}\frac{1}{M+N}\sum_{j=1}^{M+N}f(x_j)d\gamma_{N_A|N_B}(x_1,\ldots,x_{M+N})\\
=\int \frac{1}{M+N}\sum_{j=1}^{M+N}f(x_j)d\gamma_{N_A|N_B}(x_1,\ldots,x_{M+N})=\int f d\mu_{N_A,N_B}.
\end{multline}
The above chain of equalities is then an expression of the marginal $\mu_{N_A,N_B}$ of $\gamma_{N_A|N_B}$ in duality with an arbitrary test function $f$. The marginal for $\mu|_A$ is given by using the above for the modified test function $f(x)1_A(x)$. Using this choice instead of $f$, together with the self-evident fact that for each $(x_1,\ldots,x_{M+N})\in B_{N_A,N_B}$, a cyclic permutation of the indices $1,\ldots,N_A$, induces a bijection on the set of permutations $\{\sigma\in S_{M+N}: x_{\sigma(1)},\ldots,x_{\sigma(N_A)}\in A\}$, and we have from the second and last lines in \eqref{gammaNA}
\begin{multline}
\int f 1_Ad\mu_{N_A,N_B}=\int\frac{1}{M+N}\sum_{j=1}^{M+N}\frac{1}{N_A!N_B!}\sum_{\substack{\sigma\in S_{M+N},\\ x_{\sigma(1)},\ldots,x_{\sigma(N_A)}\in A}}f(x_{\sigma(j)})1_A(x_{\sigma(j)})d\gamma_{N_A|N_B}(x_1,\ldots,x_{M+N})
\\
=\int \frac{1}{N_A!N_B!}\sum_{\substack{\sigma\in S_{M+N},\\ x_{\sigma(1)},\ldots,x_{\sigma(N_A)}\in A}}\frac{1}{M+N}\sum_{j=1}^{N_A}f(x_{\sigma(j)})d\gamma_{N_A|N_B}(x_1,\ldots,x_{M+N})\\
=\int \frac{1}{N_A!N_B!}\sum_{\substack{\sigma\in S_{M+N},\\ x_{\sigma(1)},\ldots,x_{\sigma(N_A)}\in A}}\frac{N_A}{M+N}f(x_{\sigma(1)})d\gamma_{N_A|N_B}(x_1,\ldots,x_{M+N})\\
=\frac{N_A}{M+N}\int f(x_1)d\gamma_{N_A}(x_1,\ldots,x_{N_A})
=\frac{N_A}{M+N}\int f d\mu_{N_A}.
\end{multline}
This shows that $\mu_{N_A,N_B}|_A=\tfrac{N_A}{M+N}\mu_{N_A}$ and (together with the analogue statement for $\gamma_B$) we obtain that the marginals of $\gamma_A, \gamma_B$ are $\mu_{N_A},\mu_{N_B}$, as they appear in \eqref{splitmuab}, respectively. We now find, using the support and symmetry properties \eqref{pnanb} of $\gamma_{N_A|N_B}$, that
\begin{eqnarray}\label{firstdo}
 \lefteqn{\int f(x_1,\ldots,x_{M+N})d\gamma_{N_A|N_B}(x_1,\ldots,x_{M+N})}\nonumber\\
 &=&\frac{1}{N_A!N_B!}\sum_{\sigma\in S_{M+N}}\int f(x_{\sigma(1)},\ldots,x_{\sigma(M+N)})\prod_{i=1}^{N_A}1_A\left(x_{\sigma(i)}\right)d\gamma_{N_A|N_B}(x_1,\ldots,x_{M+N})\nonumber\\
 &=&\frac{1}{N_A!N_B!}\sum_{\sigma\in S_{M+N}}\int f(x_{\sigma(1)},\ldots,x_{\sigma(M+N)})\prod_{i=N_A+1}^{N+M}1_B\left(x_{\sigma(i)}\right)d\gamma_{N_A|N_B}(x_1,\ldots,x_{M+N}).
\end{eqnarray}
To prove \eqref{firstdo}, note that for any $(N+M)$-ple $(x_1,\ldots,x_{N+M})$ which belongs to any of the sets $B_{N_A,N_B}$ from \eqref{pnanb}, there exists precisely one partition $J_A:=\{i_1,\ldots,i_{N_A}\}, J_B:=\{i_{N_A+1},\ldots,i_{N+M}\}$ of $\{1,\ldots,N+M\}$ such that $x_i\in A$ if $i\in J_A$, and $x_i\in B$ if $i\in J_B$. Only the terms in the sum on the right in \eqref{firstdo} corresponding to permutations $\sigma\in S_{M+N}$ which send $\{1,\ldots,N_A\}, \{N_A+1,\ldots,N+M\}$ into respectively $J_A,J_B$ are nonzero, and the number of those $\sigma$ is $N_A!N_B!$. Together with the invariance under symmetrization of $\gamma_{N_A|N_B}$, this proves \eqref{firstdo}. Note that, even if this is not directly evident from the formula, the roles of $A$ and $B$ for the right hand side of formula \eqref{firstdo} are interchangeable.

Now by using \eqref{firstdo} for the choice $f(x_1,\ldots,x_{M+N})=\sum_{1\le i\neq j\le M+N}(\mathsf{c}_A(x_i,x_j)+\mathsf{c}_B(x_i,x_j))$ and \eqref{explicit_clumsy} (which equations extend to any function for which you can apply truncation and monotone convergence, e.g. for all positive $f$ for which the integrals are finite), we have
\begin{eqnarray}\label{cost_sym}
\lefteqn{\int\sum_{1\le i\neq j\le M+N}\left(\mathsf{c}_A(x_i,x_j) + \mathsf{c}_B(x_i,x_j)\right) d\gamma_{N_A|N_B}(x_1,\ldots,x_{M+N})}\nonumber\\
&=&\frac{1}{N_A!N_B!}\sum_{\sigma\in S_{M+N}}\int\sum_{1\le i\neq j\le M+N}\mathsf{c}_A(x_{\sigma(i)},x_{\sigma(j)})\prod_{i=1}^{N_A}1_A\left(x_{\sigma(i)}\right)d\gamma_{N_A|N_B}(x_1,\ldots,x_{M+N})\nonumber\\
&&+\frac{1}{N_A!N_B!}\sum_{\sigma\in S_{M+N}}\int\sum_{1\le i\neq j\le M+N}\mathsf{c}_B(x_{\sigma(i)},x_{\sigma(j)})\prod_{i=N_A+1}^{M+N}1_B\left(x_{\sigma(i)}\right)d\gamma_{N_A|N_B}(x_1,\ldots,x_{M+N})\nonumber\\
&=&\frac{1}{N_A!N_B!}\sum_{\sigma\in S_{M+N}}\int\sum_{1\le i\neq j\le N_A}\mathsf{c}_A(x_{\sigma(i)},x_{\sigma(j)})\prod_{i=1}^{N_A}1_A\left(x_{\sigma(i)}\right)d\gamma_{N_A|N_B}(x_1,\ldots,x_{M+N})\nonumber\\
&&+\frac{1}{N_A!N_B!}\sum_{\sigma\in S_{M+N}}\int\sum_{N_A+1\le i\neq j\le M+N}\mathsf{c}_B(x_{\sigma(i)},x_{\sigma(j)})\prod_{i=N_A+1}^{M+N}1_B\left(x_{\sigma(i)}\right)d\gamma_{N_A|N_B}(x_1,\ldots,x_{M+N})\nonumber\\
&=&\int\sum_{1\le i\neq j\le N_A}\mathsf{c}_A(x_i,x_j)d\gamma_{N_A}(x_1,\ldots,x_{N_A})+\int\sum_{1\le i\neq j\le N_B}\mathsf{c}_B(x_i,x_j)d\gamma_{N_B}(x_1,\ldots,x_{N_B}) ,
\end{eqnarray}
where for the first equality we have applied the interchangeable versions of \eqref{firstdo}, as well as the symmetry under permutations of the above integrands, for the second equality the support properties \eqref{pnanb} of $\mathsf{c}_A$ and $\mathsf{c}_B$, and for third equality we used \eqref{explicit_clumsy}.

Then \eqref{cost_sym} and \eqref{defcnanb}, directly prove \eqref{split}.

\medskip

\textbf{Step 4.} For the case of the split cost $\mathsf{c}=\mathsf{c}_A+\mathsf{c}_B$ we claim that 
\begin{multline}\label{secondsplit}
 F_{M+N,\mathsf{c}_A+\mathsf{c}_B}^\mathrm{OT}\left(\tfrac{M\mu + N\nu}{M+N}\right) \\=\min\left\{\sum_{n=2}^{M+N}\alpha_nF_{n,\mathsf{c}_A}^\mathrm{OT}(\mu_n) + \sum_{n=0}^{M+N-2}\alpha_nF_{M+N-n,\mathsf{c}_B}^\mathrm{OT}(\nu_n)\left|\begin{array}{l}\sum_{n=1}^{M+N}n\alpha_n\mu_n = M\mu\ ,\\[3mm] \sum_{n=0}^{M+N-1}(M+N-n)\,\alpha_n\nu_n= N\nu\ ,\\[2mm] \alpha_n\ge 0, \sum_{n=0}^{M+N}\alpha_n=1, \mu_n\in\mathcal P(A), \nu_n\in\mathcal P(B)\end{array}\right.\right\}\ .
\end{multline}
We now prove \eqref{secondsplit}. For each decomposition as in \eqref{firstsplit}, namely if
\begin{equation}\label{splitting}
\frac{M\mu + N\nu}{M+N}=\sum_{\substack{N_A+N_B=M+N\\N_A,N_B\in\mathbb N}}\mu_{N_A,N_B},\quad\mu_{N_A,N_B}\in\mathcal M_{N_A|N_B}^+(\mathbb R^d)\quad\text{for all}\quad N_A,N_B\ ,
\end{equation}
and taking into account that there are $M+N+1$ terms $(N_A,N_B)$ with $N_A+N_B=M+N$, $N_A,N_B\in\mathbb{N}$, we can perform a reparameterization as follows:
\begin{equation}\label{reparamnanb}
\text{There exist }\left\{\begin{array}{l}
\alpha_n\ge0,\ \sum_{n=0}^{M+N}\alpha_n=1\ ,\\[3mm]\mu_n\in\mathcal P(A), \nu_n\in\mathcal P(B)
\end{array}\right.\text{ such that }
\left\{\begin{array}{l}N_A=n,\quad N_B=M+N-n\ ,\\[2mm] \mu_{N_A,N_B}=\alpha_n \left(\frac{N_A}{M+N}\ \mu_n + \frac{N_B}{M+N}\ \nu_n\right)\ .
\end{array}\right.
\end{equation}
Assuming \eqref{reparamnanb}, we may apply to each $\mu_{N_A,N_B}$ equations (\ref{firstsplit}) and \eqref{split} of the previous step together with the essentially disjoint supports hypothesis on $\mu,\nu,$, to obtain \eqref{secondsplit}.

\medskip

To prove \eqref{reparamnanb}, with the notation in \eqref{splitting}, we first note that $N_A+N_B=M+N$ implies that if we write $N_A=n$ then $N_B=M+N-n$. 

Next, we define $\alpha_{N_A}:=\mu_{N_A,N_B}(A\cup B)\ge 0$, and by applying equation \eqref{splitting}, we find $\sum_{n=0}^{M+N}\alpha_n=1$. In case $\alpha_n\neq 0$, the measures $\mu_n, \nu_n$ are obtained from the renormalized measure $(\alpha_n)^{-1}\mu_{N_A,N_B}\in\mathcal P_{N_A|N_B}(\R^d)$ like in \eqref{splitmuab} from the previous step, via the definition \eqref{marginalnanb}. Then the bottom equation on the right of \eqref{reparamnanb} follows directly. This completes the proof of \eqref{reparamnanb}, and thus we completed the proof of \eqref{secondsplit} too.

\medskip

\textbf{Step 5.} Now comparing the right hand side of \eqref{secondsplit} to the decompositions coming from definition \eqref{OTGCB} of the $F_{\mathrm{GCB},\mathsf{c},N}$-problems, we see that $\{\alpha_n, \mu_n\}_{n\le M+N}$ and $\{\alpha_n,\nu_n\}_{n\le M+N}$ form competitors  for the minimum problems $F_{\mathrm{GCB},\mathsf{c}_A,M}(\mu,M+N)$ and $F_{\mathrm{GCB},\mathsf{c}_B,N}(\nu,M+N)$, respectively. This allows to bound the left hand side of \eqref{secondsplit} from below by the sum of these $GCB$-problems, which concludes the proof.
\end{proof}
\begin{proof}[Proof of Lemma~\ref{lemma_na}:]
\textbf{Step 1.} \textit{Notations.} We recall the notation $C_{\vec x}\stackrel{r_A}{\mapsto}C_{\vec x}\cap A$ for the restriction map and we recall that $\gamma_{N_A}=(r_A)_\#\gamma_{N_A|N_B}$. We further introduce a notation for the map that transforms an ordered $M+N$-ple into a multiset of cardinality $M+N$, which is defined by
\begin{equation}
i_{M+N}:(\R^d)^{M+N}\to\mathcal C_{M+N}(\R^d),\quad i_{M+N}(x_1,\ldots,x_{M+N})=\{x_1,\ldots,x_{M+N}\}.
\end{equation}
Then we use the notation
\begin{equation}
\bar \gamma:=\left(i_{M+N}\right)_\#\gamma,\quad\mbox{for}\quad\gamma\in\mathcal P_\mathrm{sym}((\R^d)^{M+N}).
\end{equation}
We introduce the following equivalence relation
\begin{equation}
(x_1,\ldots,x_{M+N})\sim(y_1,\ldots,y_{M+N})\quad\mbox{if}\quad\exists\sigma\in S_{M+N},\ \forall j=1,\ldots,M+N, x_j=y_{\sigma(j)}.
\end{equation}
We also denote $(y_{\sigma(1)},\ldots,y_{\sigma(M+N)})$ by $\sigma\vec y$ for brevity. Then for the quotient under the above equivalence relation $(\R^d)^{M+N}/\sim$ has a measurable set of representatives $\Omega_{M+N}\subset(\R^d)^{M+N}$, i.e. we can find a measurable set $\Omega_{M+N}$ such that for every $\vec x\in(\R^d)^{M+N}$ there exists a unique $\vec y\in\Omega_{M+N}$ for which there exist $\sigma\in S_{M+N}$ such that $\vec x=\sigma\vec y$. We define the map
\begin{equation}\label{defrho}
\rho_{M+N}:\mathcal C_{M+N}\to(\R^d)^{M+N},\quad \rho_{M+N}(\{x_1,\ldots,x_{M+N}\}):=\vec y\in\Omega_{M+N}\mbox{ s.t. }\vec x\sim\vec y.
\end{equation}
By the preceding discussion, the above choice of $\vec y$ is unique, and thus the map $\rho_{M+N}$ is well defined, and furthermore it is measurable.

\medskip

We next define the involution $\mathcal P((\R^d)^{M+N})\to\mathcal P((\R^d)^{M+N})$ with image $\mathcal P_\mathrm{sym}((\R^d)^{M+N})$, given by the symmetrization operation
\begin{equation}
\gamma\in\mathcal P((\R^d)^{M+N})\mapsto \gamma^\mathrm{sym}\in\mathcal P_\mathrm{sym}((\R^d)^{M+N}),
\end{equation}
given in duality with test functions $f\in C_b((\R^d)^{M+N})$ by the formula
\begin{equation}
\int f(\vec x)d\gamma^\mathrm{sym}(\vec x):= \frac{1}{(M+N)!}\sum_{\sigma\in S_{M+N}}f(\sigma\vec x)d\gamma(\vec x).
\end{equation}
\textbf{Step 2.} \textit{Claim:} There holds $\left(\rho_{M+N}\right)_\#\bar\gamma\in\mathcal P(\Omega_{M+N})\subset\mathcal P((\R^d)^{M+N})$, and for all $\gamma\in\mathcal P_\mathrm{sym}((\R^d)^{M+N})$ there holds, with $\rho_{M+N}$ as defined in \eqref{defrho},
\begin{equation}\label{claimi}
\left(\left(\rho_{M+N}\right)_\#\bar\gamma\right)^\mathrm{sym}=\gamma.
\end{equation}
It is enough to verify that for every test function $f\in C_b((\R^d)^{M+N})$ the duality with the above two measures gives the same value. Indeed, by repeatedly using the formula $\forall f\in C_b, \int fd g_\#\mu=\int f\circ g d\mu$ for the pushforward of a measure and the definitions introduced above, we have:
\begin{eqnarray}
\lefteqn{\int f(\vec x)d\left(\left(\rho_{M+N}\right)_\#\bar\gamma\right)^\mathrm{sym}(\vec x)}\nonumber\\
&=&\int\frac{1}{(M+N)!}\sum_{\sigma\in S_{M+N}}f(\sigma\vec x)d\left(\left(\rho_{M+N}\right)_\#\bar\gamma\right)(\vec x)\nonumber\\
&=&\frac{1}{(M+N)!}\sum_{\sigma\in S_{M+N}}\int f\left(\sigma\left(\rho_{M+N}\left(\{x_1,\ldots,x_{M+N}\}\right)\right)\right)d\bar\gamma\left(\{x_1,\ldots,x_{M+N}\}\right)\nonumber\\
&=&\frac{1}{(M+N)!}\sum_{\sigma\in S_{M+N}}\int f\left(\sigma\left(\rho_{M+N}\circ i_{M+N}(\vec x)\right)\right)d\gamma(\vec x)=\frac{1}{(M+N)!}\sum_{\sigma\in S_{M+N}}\int f(\sigma\vec y)d\gamma(\vec x)\nonumber\\
&=&\frac{1}{(M+N)!}\sum_{\sigma\in S_{M+N}}\int f(\sigma\vec x)d\gamma(\vec x)=\int f(\vec x) d\gamma(\vec x),
\end{eqnarray}
where, to justify the above equalities, we observe that $\{x_1,\ldots,x_{M+N}\}$ is $S_{M+N}$-invariant, then we observed that $\rho_{M+N}\circ i_{M+N}(\vec x)=\vec y$, the unique element in the $S_{M+N}$-orbit of $\vec x$ that lies in $\Omega_{M+N}$, and thus $f\circ\rho_{M+N}\circ i_{M+N}$ is constant on this orbit, which is equal to the $S_{M+N}$-orbit of $\vec y$, and finally we used the fact that $\gamma=\gamma^\mathrm{sym}$.

\medskip

\textbf{Step 3.} \textit{Conclusion of the proof.} We thus have proved \eqref{claimi}, and this gives an inverse operation to the pushforward $(i_{M+N})_\#:\mathcal P_\mathrm{sym}((\R^d)^{M+N})\to\mathcal P(\mathcal C_{M+N}(\R^d))$. We denote this inverse and characterize it as in \eqref{claimi}, as follows
\[
I_{M+N}:\mathcal P(\mathcal C_{M+N}(\R^d))\stackrel{\simeq}{\to}\mathcal P_\mathrm{sym}((\R^d)^{M+N}),\quad I_{M+N}(\bar\gamma):=((\rho_{M+N})_\#\bar\gamma)^\mathrm{sym}\quad\mathrm{for}\quad\bar\gamma\in\mathcal P(\mathcal C_{M+N}(\R^d)).
\]
The link expressed in \eqref{gamma_na}, between $\gamma_{N_A}$ and $\gamma_{N_A|N_B}$ is thus given via the maps $i_{M+N}$ and $I_{M+N}$. More precisely, \eqref{gamma_na} defines $\gamma_{N_A}$ as 
\begin{equation}\label{def_nonclumsy}
\gamma_{N_A}:=I_{N_A}\left((r_A)_\#(i_{M+N})_\#\gamma_{N_A|N_B}\right).
\end{equation}
We now prove that this $\gamma_{N_A}$ satisfies \eqref{explicit_clumsy}. Indeed, starting from the definition \eqref{def_nonclumsy} we have
\begin{eqnarray}
\lefteqn{\int f(\vec x_{N_A})d\gamma_A(\vec x_{N_A})=\int f(\vec x_{N_A}) d I_{N_A}\left((r_A)_\#(i_{M+N})_\#\gamma_{N_A|N_B}\right)}\nonumber\\
&=&\int f(\vec x_{N_A})d\left((\rho_{N_A})_\#(r_A)_\#(i_{M+N})_\#\gamma_{N_A|N_B}\right)^\mathrm{sym}\nonumber\\
&=&\frac{1}{N_A!}\sum_{\sigma\in S_{N_A}}\int f\left(\sigma\left(\rho_{N_A}\circ r_A\circ i_{M+N}(\vec x_{M+N})\right)\right)d\gamma(\vec x_{M+N})\nonumber\\
&=&\frac{1}{N_A!}\sum_{\sigma\in S_{N_A}}\int f\left(\sigma\left(\rho_{N_A}(\{x_1,\ldots,x_{M+N}\}\cap A)\right)\right)d\gamma(\vec x_{M+N}).\label{claimso}
\end{eqnarray}

Next, we claim that for each $\vec x_{M+N}\in B_{N_A,N_B}$ as defined in Step 1 of the proof of Lemma \ref{subadd12gena01} (and thus for $\gamma_{N_A|N_B}$-almost every $\vec x_{M+N}$) there holds 
\begin{equation}\label{claimsospeso}
\frac{1}{N_A!}\sum_{\sigma\in S_{N_A}}f\left(\sigma\left(\rho_{N_A}(\{x_1,\ldots,x_{M+N}\}\cap A)\right)\right)=\frac{1}{N_A!N_B!}\sum_{\sigma\in S_{M+N}}f(x_{\sigma(1)},\ldots, x_{\sigma(N_A)})\prod_{j=1}^{N_A}1_A(x_{\sigma(j)}).
\end{equation}
Indeed, for a fixed $\vec x_{M+N}\in B_{N_A,N_B}$ let $i_1<\ldots<i_{N_A}$ be the indices such that $x_{i_j}\in A$. Then consider first the right hand side of \eqref{claimsospeso}. There are precisely $N_A!N_B!$ possible permutations $\sigma\in S_{M+N}$ which send the set $\{i_1,\ldots,i_{N_A}\}$ into the set $\{1,\ldots, N_A\}$, and they form precisely the set of terms on which the left-hand side in \eqref{claimsospeso} is nonzero. This set of permutations can be written as $\sigma_0 S_{N_A}$, where $\sigma_0\in S_{M+N}$ is any fixed permutation such that $\sigma_0(k)=i_k$ for $1\le k\le N_A$ and $S_{N_A}$ are the permutations of $\{1,\ldots, N_A\}$ which leave fixed the remaining indices. This last subgroup acts transitively on the first $N_A$ indices, so the right hand side of \eqref{claimsospeso} equals, for $f:(\R^d)^{N_A}\to \R$ and with $\rho_{N_A}$ defined similarly to \eqref{defrho},
\begin{equation}\label{claimsosp2}
\frac{1}{N_A!}\sum_{\sigma\in S_{N_A}}f(x_{i_{\sigma(1)}},\ldots, x_{i_{\sigma(N_A)}}).
\end{equation}
Now, for the left hand side of \eqref{claimsospeso}, note that $\rho_{N_A}(\{x_1,\ldots,x_{M+N}\}\cap A)$ is of the form $(x_{k_1},\ldots,x_{k_{N_A}})$, and, by the transitivity action of $S_{N_A}$ on the indices, the average on left hand side of \eqref{claimsospeso} is identified to \eqref{claimsosp2}, as desired. As now \eqref{claimsospeso} is proved, we may substitute its right hand side into \eqref{claimso}, which then lends the right hand side of \eqref{gamma_na}, proving the claim \eqref{explicit_clumsy} of the lemma.
\end{proof}
Lemma \ref{subadd12gena01} allows to show the following useful result.
\begin{lemma}\label{subadd_gcb}
Let $\mathsf{c}:\mathbb{R}^d\times\mathbb{R}^d\rightarrow\mathbb{R}_{>0}\cup\{+\infty\}$. Consider $\mu_1, \ldots, \mu_k\in\mathcal{P}(\mathbb{R}^d)$ with densities $\rho_1, \ldots,\rho_k,$ such that the integrals needed to define the quantities below are finite, such that the $\mu_j$ have essentially disjoint supports, i.e. that for every Borel set $A$ and every $1\le i\neq j\le k$ there holds $\mu_i(A)\mu_j(A)=0$. Fix $M_1,\ldots, M_k\in\mathbb R_{>0}$, let $\bar{M}:=\sum_{i=1}^k M_i\in\mathbb{R}_{>0},$ and denote by $\mathsf{c}_{ii}(x,y):=1_{\Lambda_i}(x)1_{\Lambda_i}(y)\mathsf{c}(x,y)$ for $i=1,\ldots,k,$ and $x,y\in\mathbb R^d$. Let  $\mu\in \mathcal{P}(\mathbb{R}^d)$ have density $\sum_{i=1}^kM_i \rho_i/\bar{M}$. Set $M'\in\mathbb{N}_+,M'\ge \max\{\bar M,1\}$. Then
\begin{equation}\label{supadd1gen0eqGCB}
F_{\mathrm{GCB},\bar M, \sum_{i=1}^k\mathsf{c}_{ii}}^\mathrm{OT}\left(\mu,M'\right) 
\ge \sum_{i=1}^kF_{\mathrm{GCB},M_i,\mathsf{c}_{ii}}^\mathrm{OT}\left(\mu_i, M'\right)=\sum_{i=1}^kF_{\mathrm{GCB},M_i,\mathsf{c}}^\mathrm{OT}\left(\mu_i, M'\right)\ .
\end{equation}
and 
\begin{equation}\label{supadd1gen0eqGC}
F_{\mathrm{GC},\bar M, \sum_{i=1}^k\mathsf{c}_{ii}}^\mathrm{OT}\left(\mu\right) 
\ge \sum_{i=1}^kF_{\mathrm{GC},M_i,\mathsf{c}_{ii}}^\mathrm{OT}\left(\mu_i\right)=\sum_{i=1}^kF_{\mathrm{GC},M_i,\mathsf{c}}^\mathrm{OT}\left(\mu_i\right)\ .
\end{equation}
The inequality in (\ref{supadd1gen0eqGC}) becomes an equality if $M_i>1$, for all $i=1,\ldots,k$.
\end{lemma}
\begin{proof}
We start by proving (\ref{supadd1gen0eqGCB}). Note that the thesis is trivially valid if $M^\prime<2$ due to the convention that $F_N^\mathrm{OT}=0$ for $N=0,1$, and thus we assume from now on that $M^\prime\ge 2$. We prove the lemma for the case of a decomposition into two measures, and by applying the lemma repeatedly we can extend it to the case of $k$ measures.

By definition of $F_{\mathrm{GCB},\bar M,\mathsf{c}_{11} +\mathsf{c}_{22}}^{\mathrm{OT}}(\mu,M')$, there exist $\alpha_n\ge 0$ and $\mu_n\in\mathcal P(\mathbb R^d)$, with $0\le n\le M'$, with
\begin{equation}\label{coeff_gcb_mu}
\sum_{n=0}^{M'}\alpha_n=1,\quad \sum_{n=1}^{M'}n\alpha_n\mu_n=\bar M\mu,\quad F_{\mathrm{GCB},\bar M,\mathsf{c}_{11} +\mathsf{c}_{22}}^{\mathrm{OT}}(\mu,M')=\sum_{n=2}^{M'}\alpha_n F_{n,\mathsf{c}_{11}+\mathsf{c}_{22}}^{\mathrm{OT}}(\mu_n)\ .
\end{equation}
Since for each $n\in\{1,\ldots,M'\}$, we can write $n\mu_n=n\mu_n|_{\Lambda_1}+n\mu_n|_{\Lambda_2}$. Thus, by Lemma \ref{subadd12gena01} and (\ref{optimnenew}) there exist positive real numbers $M_n^1,M_n^2\ge 0,$ and measures $\mu_n^1\in\mathcal P(\Lambda_1), \mu_n^2\in\mathcal P(\Lambda_2),$ such that
\begin{equation}\label{apply_lem47}
M_n^1+M_n^2=n,\quad n\mu_n=M_n^1\mu_n^1+M_n^2\mu_n^2,\quad F_{n,\mathsf{c}_{11}+\mathsf{c}_{22}}^{\mathrm{OT}}(\mu_n)\ge F_{\mathrm{GCB},M_n^1,\mathsf{c}}^\mathrm{OT}\left(\mu_n^1, n\right)+F_{\mathrm{GCB},M_n^2,\mathsf{c}}^\mathrm{OT}\left(\mu_n^2, n\right)\ .
\end{equation}
Let $\alpha_{n,m}^1,\alpha_{n,m}^2\ge 0,$ and $\mu_{n,m}^1\in\mathcal P(\Lambda_1), \mu_{n,m}^2\in\mathcal P(\Lambda_2)$, $m\in\{0,\ldots,n\},$ be such that for $j=1,2$
\begin{equation}\label{coeff_gcb_munj}
\quad F_{\mathrm{GCB},M_n^j,\mathsf{c}}^{\mathrm{OT}}(\mu_n^j,n)=\sum_{m=2}^n\alpha_{n,m}^j F_{m,\mathsf{c}}^\mathrm{OT}(\mu_{n,m}^j),\quad \sum_{m=0}^n\alpha_{n,m}^j=1,\quad \sum_{m=1}^nm\alpha_{n,m}^j\mu_{n,m}^j=M_n^j\mu_n^j\ .
\end{equation}
Up to adding extra coefficients $\alpha_{n,n+1}^j=\cdots=\alpha_{n,M'}^j=0$, we may extend the sums in \eqref{coeff_gcb_munj} up to $m=M'$. Now, we sum \eqref{coeff_gcb_munj} multiplied by the coefficients from \eqref{coeff_gcb_mu}, and we find from \eqref{coeff_gcb_mu}, \eqref{apply_lem47} and \eqref{coeff_gcb_munj} 
\begin{align}
\sum_{m=0}^{M'}\sum_{n=0}^{M'}\alpha_n\alpha_{n,m}^1=1,\quad \sum_{m=1}^{M'} m\sum_{n=1}^{M'}\alpha_n\alpha_{n,m}^1\mu_{n,m}^1=\sum_{n=1}^{M'}\alpha_nM_n^1\mu_n^1=\bar M\mu|_{\Lambda_1}=M_1\mu_1\ ,\nonumber \\
\sum_{m=0}^{M'}\sum_{n=0}^{M'}\alpha_n\alpha_{n,m}^2=1,\quad \sum_{m=1}^{M'} m\sum_{n=1}^{M'}\alpha_n\alpha_{n,m}^2\mu_{n,m}^2=\sum_{n=1}^{M'}\alpha_nM_n^2\mu_n^2=\bar M\mu|_{\Lambda_2}=M_2\mu_2\ ,\nonumber \\
\mbox{and }\quad F_{\mathrm{GCB},N,\mathsf{c}_{11} +\mathsf{c}_{22}}^{\mathrm{OT}}(\mu,M') \ge \sum_{m=2}^{M'}\sum_{n=2}^{M'}\alpha_n\alpha_{n,m}^1 F_{m,\mathsf{c}}^\mathrm{OT}(\mu_{n,m}^1)+\sum_{m=2}^{M'}\sum_{n=2}^{M'}\alpha_n\alpha_{n,m}^2 F_{m,\mathsf{c}}^\mathrm{OT}(\mu_{n,m}^2)\ .\label{ineq_long}
\end{align}
By defining $\beta_m^1=\sum_{n=0}^{M'}\alpha_n\alpha_{n,m}^1$, $\beta_m^2=\sum_{n=0}^{M'}\alpha_n\alpha_{n,m}^2$, and by using the inequalities below and their analogues for $\mu_2$ (where the inequality below holds by the Monge-Kantorovich duality), we get for $\beta_m^1>0$
\begin{equation}\label{superadd}
\frac{\sum_{n=2}^{M'}\alpha_n\alpha_{n,m}^1F_{m,\mathsf{c}}^\mathrm{OT}(\mu_{n,m}^1)}{\sum_{n=0}^{M'}\alpha_n\alpha_{n,m}^1}=\frac{1}{\beta_m^1}\sum_{n=2}^{M'}\alpha_n\alpha_{n,m}^1F_{m,\mathsf{c}}^\mathrm{OT}(\mu_{n,m}^1)\ge F_{m,\mathsf{c}}^\mathrm{OT}\left(\frac{1}{\beta_m^1}\sum_{n=2}^{M'}\alpha_n\alpha_{n,m}^1\mu_{n,m}^1\right).
\end{equation}
Due to the first line in \eqref{ineq_long}, we find that the measures $(\beta_m^1)^{-1}\sum_{n=2}^{M'}\alpha_n\alpha_{n,m}^1\mu_{n,m}^1$ and coefficients $\beta_m^1$ form a competitor for the definition of $F_{\mathrm{GCB},M_1,\mathsf{c}}^\mathrm{OT}(\mu_1,M')$, and a similar construction and statement apply also for $\mu_2$. Thus, we obtain in the last line of \eqref{ineq_long} from (\ref{superadd}) 
\begin{eqnarray*}
 F_{\mathrm{GCB},N,\mathsf{c}_{11} +\mathsf{c}_{22}}^{\mathrm{OT}}(\mu,M')&\ge& \sum_{m=2}^{M'}\beta_m^1 F_{m,\mathsf{c}}^\mathrm{OT}\left(\frac{1}{\beta_m^1}\sum_{n=2}^{M'}\alpha_n\alpha_{n,m}^1\mu_{n,m}^1\right)+\sum_{m=2}^{M'}\beta_m^2 F_{m,\mathsf{c}}^\mathrm{OT}\left(\frac{1}{\beta_m^2}\sum_{n=2}^{M'}\alpha_n\alpha_{n,m}^2\mu_{n,m}^2\right)\\
 &\ge&F_{\mathrm{GCB},M_1,\mathsf{c}}^\mathrm{OT}(\mu_1,M')+F_{\mathrm{GCB},M_2,\mathsf{c}}^\mathrm{OT}(\mu_2,M').
 \end{eqnarray*}
Thereby \eqref{superadd} together with the above equation prove the first statement of the lemma.

\par The inequality $\ge$ in \eqref{supadd1gen0eqGC} follows by similar arguments, hence its proof will be omitted. The inequality $\le$ in \eqref{supadd1gen0eqGC} when $M_i>1, i=1,\ldots,k,$ follows from \eqref{subadgcinit}.
\end{proof}

As an immediate consequence of Lemma \ref{subadd12gena01} and Lemma \ref{subadd_gcb}, we have the next corollary.
\begin{corollary}\label{subadd12gen11}
Let $\mathsf{c}:\mathbb{R}^d\times\mathbb{R}^d\rightarrow\mathbb{R}_{>0}\cup \{+\infty\}$, and set $N\in\mathbb{N}_+, N\ge 2,$ and $\bar N\in \mathbb{R}_{>0}$. For $k\in\mathbb{N}_+$, consider the probability measures $\mu_1, \ldots, \mu_k,$ with densities $\rho_1,\ldots,\rho_k,$ such that the integrals required to define the quantities below are finite, and the measures $\mu_j$ have essentially disjoint supports, i.e that for all Borel sets $A$ and all $1\le i\neq j\le k$ there holds $\mu_i(A)\mu_j(A)=0$. Let $\alpha_j\in\mathbb{R}_{>0},j=1,\ldots,k,$ be such that $\sum_{j=1}^k\alpha_j=1$. Let $\mu\in \mathcal{P}(\mathbb{R}^d)$ have density $\rho=\sum_{j=1}^k\alpha_j\rho_j$. Denote by $\mathsf{c}_{ii}:=\mathsf{c} 1_{\Lambda_i\times\Lambda_i}, i=1,\ldots,k$. Then                                                
\begin{equation}\label{supadd1gen'}
E_{N, \sum_{i=1}^k \mathsf{c}_{ii}}^\mathrm{xc}\bigg(\sum_{j=1}^k\alpha_j\mu_j\bigg)\ \ge \sum_{j=1}^k {E_{\mathrm{GC},N\alpha_j,\mathsf{c}}^\mathrm{xc}}\left(\mu_j\right)~~\mbox{and}~~E_{\mathrm{GC},\bar N, \sum_{i=1}^k \mathsf{c}_{ii}}^\mathrm{xc}\bigg(\sum_{j=1}^k\alpha_j\mu_j\bigg)\ge \sum_{j=1}^k {E_{\mathrm{GC},\bar N\alpha_j,\mathsf{c}}^\mathrm{xc}}\left(\mu_j\right).
\end{equation}
\end{corollary}
\par We note here that in \cite[Thm. 3.1, Cor. 3.4]{LewLiebSeir17} also the result analogous to Proposition \ref{unif} for $E_{\mathrm{GC},N,\mathsf{c}}^\mathrm{xc}$ was proved. More precisely, as a consequence of \cite[Thm. 3.1]{LewLiebSeir17}, for any measure $\mu\in\mathcal P(\mathbb R^d)$ such that $d\mu(x)=\rho(x)dx$ with $\rho(x)=|\Lambda|^{-1}1_\Lambda(x)$ and $\Lambda$ a Borel set with $\phi$-regular boundary, we have
\begin{equation}\label{egclimitunifnat}
 \lim_{N\to\infty, N\in\mathbb{N}}N^{-1-s/d}E_{\mathrm{GC},N,s}^\mathrm{xc}(\mu)=  \lim_{N\to\infty}N^{-1-s/d}E_{N,s}^\mathrm{xc}(\mu)=-C(s,d)\int_{\mathbb R^d}\rho^{1+s/d}=-C(s,d)|\Lambda|^{-s/d}\ .
\end{equation}
In the above cited result from \cite{LewLiebSeir17}, \eqref{egclimitunifnat} is stated with $N\in\mathbb{N}$. In the next lemma we extend it to the case $N\in\mathbb{R}_{>0}$.
\begin{lemma}\label{eglimgen}
For any $\mu\in\mathcal P(\mathbb R^d)$ with density $\rho$ of form $\rho(x)=|\Lambda|^{-1}1_\Lambda(x)$, where $\Lambda$ is a Borel set with $\phi$-regular boundary, it holds that 
\begin{equation}\label{egclimitunif}
 \lim_{N\to\infty, N\in\mathbb{R}_{>0}}N^{-1-s/d}E_{\mathrm{GC},N,s}^\mathrm{xc}(\mu)= -C(s,d)\int_{\mathbb R^d}\rho^{1+s/d}=-C(s,d)|\Lambda|^{-s/d}\ .
\end{equation}
\end{lemma}
\begin{proof}
To show \eqref{egclimitunif}, fix a large non-integer $N\in\mathbb{R}_{>0}$. We have from Remark \ref{rmk_gcb}  2) and 4) that
\begin{eqnarray}
\label{someineqGC}
\lefteqn{E_{\mathrm{GC},[N]+2,s}^\mathrm{xc}(\mu)\le E_{\mathrm{GC},[N]+2,s}^\mathrm{xc}(\mu)-E_{\mathrm{GC},[N]+2-N,s}^\mathrm{xc}(\mu)}\nonumber\\
&\le& E_{\mathrm{GC},N,s}^\mathrm{xc}(\mu)\le E_{\mathrm{GC},[N]-1,s}^\mathrm{xc}(\mu)+E_{\mathrm{GC},N+1-[N],s}^\mathrm{xc}(\mu)\le  E_{\mathrm{GC},[N]-1,s}^\mathrm{xc}(\mu).
\end{eqnarray}
Combining now \eqref{egclimitunifnat} and \eqref{someineqGC} produces \eqref{egclimitunif}.
\end{proof}
\subsubsection{Proof of the lower bound}\label{ssseclowerboundpf}
Before we proceed, we remind the reader that the formula we will be working with below, as explained in \eqref{lbrough2}, \eqref{lbroughqi}, \eqref{lbrough3}, \eqref{notationerrorswh}, \eqref{notationerrorsw1} and \eqref{notationerrorsw2}, is
\begin{equation}\label{recallsplit}
\frac{1}{|x-y|^s} = \frac{M}{M+C}\bigg\{\int_{\Omega_l} \left(\sum_{A\in F^l_\omega}\frac{1_A(x)1_A(y)}{|x-y|^s}\right) d\,\mathbb{P}_l(\omega) + w(x-y)\bigg\}\ ,
\end{equation}
where $l>0$, and where $w$ satisfies the properties in Proposition \ref{prop3ws} above. 

\par Denote for each $\omega\in\Omega_l$ by 
\begin{subequations}
\begin{equation}\label{sigmaerr0}
\Lambda_\mathrm{nc}(l,\omega):=\Lambda\setminus\bigcup_{A\in F^l_\omega}\left(\Lambda\cap A\right)\ .
\end{equation}
This is the part left not covered in the packing of $\Lambda\subset\mathbb{R}^d$ by balls given by the Swiss cheese Lemma \ref{cheeselemma}. Moreover, let
\begin{equation}\label{sigmaerr1}
\Lambda_\mathrm{err}(l,\omega):= \bigcup\left\{ A\in F^l_\omega:\ A\cap\partial\Lambda\neq\emptyset\right\}\ .
\end{equation}
\end{subequations}
The \eqref{sigmaerr1} will be substituted in the proof of Proposition \ref{mixextlowb}, where $\Lambda=\cup_{i=1}^k\Lambda_i$, by
\begin{equation}\label{sigmaerr2}
\Lambda_\mathrm{err}(l,\omega):= \bigcup\left\{ A\in F^l_\omega:\ A\cap \bigcup_{i=1}^k\partial\Lambda_i\neq\emptyset\right\}\ ,
\end{equation}
for which the proof of Lemma \ref{subadd12} below also works. From \eqref{sigmaerr1}, it immediately follows that
\begin{equation}\label{remerr1}
\Lambda\setminus \left(\Lambda_\mathrm{err}(l,\omega)\cup \Lambda_\mathrm{nc}(l,\omega)\right)= \bigcup\left\{ A\in F^l_\omega:\ A\subset\Lambda\right\}\ ,
\end{equation}
and from \eqref{sigmaerr2} that
\begin{equation}\label{remerr2}
\Lambda\setminus  \left(\Lambda_\mathrm{err}(l,\omega)\cup \Lambda_\mathrm{nc}(l,\omega)\right)=\bigcup\left\{ A\in F^l_\omega:\ A\subset\Lambda_i,\quad\mbox{for some}\quad i=1,\ldots,k\right\}\ .
\end{equation}
We will next state a result involving the uniform decay in $\omega $ of $|\Lambda_\mathrm{nc}|, |\Lambda_\mathrm{err}|$ as $M\rightarrow\infty$.
\begin{lemma}\label{unifdecaySigma}
If $\Lambda$ has finite volume, then for $M,l$ as in Lemma \ref{cheeselemma} and for $\Omega_l$ as in \eqref{lbrough3} above, for each $\omega\in\Omega_l$, we have 
\begin{equation}\label{480appr}
 \sup_{\omega\in\Omega_l}|\Lambda_\mathrm{nc}(l,\omega)|\le \frac{C_d}{M}|(\Lambda)_{2l\sqrt{d}}|, \quad \mbox{where}\quad (\Lambda)_r:=\{x\in\mathbb R^d:\ \mathrm{dist}(x,\Lambda)\le r\},
 \end{equation}
and if $\Lambda$ has $\phi$-regular boundary in the sense of \eqref{phiregular} for a continuous $\phi:[0,t_0)\to \mathbb R^+$ with $\phi(0)=0$ then
\begin{equation}\label{481appr}
 \sup_{\omega\in\Omega_l}|\Lambda_\mathrm{err}(l,\omega)|\le  \phi(\ell |\Lambda|^{-1/d}) |\Lambda|.
\end{equation}
\end{lemma}
\begin{proof}
We observe first that by the property of having $\phi$-regular boundary from (\ref{phiregular}), the total volume of cubes $A\in F_\omega^l$ which touch $\partial \Lambda$ for the case \eqref{sigmaerr1} (respectively $\partial\Lambda_i$ for the case \eqref{sigmaerr2}), is smaller than
$$\ \left|\left\{x:\ d(x,\partial\Lambda)\le (1+\eta)\ell\right\}\right|\le\phi(\ell |\Lambda|^{-1/d}) |\Lambda|~~\mbox{and}~~\ \left|\left\{x:\ d(x,\partial\Lambda_i)\le (1+\eta)\ell\right\}\right|\le\phi(\ell (\min_{1\le i\le k}|\Lambda_i|)^{-1/d}) |\Lambda|,$$
as for all $\omega\in \Omega_l$ these sets are included in cubes of size at most $(1+\eta)l$, due to \eqref{lbrough3}. This proves (\ref{481appr}).

Applying Lemma \ref{cheeselemma}, we find that uniformly among $\omega\in\Omega_l$ of the form \eqref{lbrough3}, we can estimate for the case  \eqref{sigmaerr1} the contribution of $\Lambda_\mathrm{nc}(l,\omega)$ as being at most $C_d M^{-1}|K|$ restricted to each cube $K=K_\omega$ that meets $\Lambda$, and taking among the cubes from the cheese lemma coverings used for $\omega\in\Omega_l$. Each such cube $K_\omega$ remains within the neighborhood $(\Lambda)_{2l\sqrt{d}}$ of thickness $2l\sqrt{d}$ of $\Lambda$, as the diameter of $K$ is $tl\sqrt{d}\le 2l\sqrt{d}$, since $t\in[1-\eta,1+\eta]$. Summing all the contributions of all such cubes gives \eqref{480appr}.
\end{proof}
We introduce here the following normalization notation and convention for measures, which will be used at several instances below. If $\mu$ is a Borel measure, $A\subset \mathbb R^d$ is a Borel set and $\mu(A)>0$, then we denote
\begin{equation}\label{hm}
\hat\mu_{A}:=\frac{\mu|_{A}}{\mu(A)}\ ,
\end{equation}
and if $\mu(A)=0$ we set $\hat\mu_A=0$. Also note that automatically in this case there holds $\mu|_A=\mu(A)\hat\mu_A$, which will also be used in several instances below.

\medskip

Next, we will show
\begin{lemma}\label{subadd12}
Fix $0<\epsilon<d/2$, set $\epsilon\le s\le d-\epsilon$ and let $\mathsf{c}(x,y)=|x-y|^{-s}$. Let $\mu\in\calP(\mathbb{R}^d)$ be a probability measure with density $\rho\in L^{1+\frac{s}{d}}(\mathbb{R}^d)$ supported on a Borel set $\Lambda\subset\mathbb{R}^d$ with $\phi$-regular boundary. Set $N\in\mathbb{N}, N\ge 2$. Let $l>0$ and $M>0$ be as in Lemma \ref{cheeselemma}. Then there holds
\begin{eqnarray}\label{supadd1gen}
E_{N,s}^\mathrm{xc}\left(\mu\right)&\ge& E_{\mathrm{GC},N,s}^\mathrm{xc}\left(\mu\right)\nonumber\\
&\ge&  \frac{M}{M+C}\bigg\{\int_{\Omega_l} \bigg(\sum_{\substack{A\in F^l_\omega\\ A\subset\Lambda,\mu(A)>0}} {E_{\mathrm{GC},N\mu(A),s}^\mathrm{xc}}(\hat\mu_{ A})+E_{\mathrm{GC},N\mu(\Lambda_\mathrm{err}(l,\omega)),s}^\mathrm{xc}(\hat\mu_{ \Lambda_\mathrm{err}(l,\omega)})\bigg) d\,\mathbb{P}_l(\omega)\nonumber\\
&-&\frac{C(w,d, \epsilon)}{M}N^{1+s/d}\int_{\mathbb{R}^d}\rho^{1+s/d}(x)dx-\frac{C(w,d,\epsilon)}{M}R_1^{-s}(N-1)\bigg\}\ ,
\end{eqnarray}
for some $C,C(w,d, \epsilon),$ which depend only on $\epsilon$ and $d$. 
Furthermore, the second inequality in \eqref{supadd1gen} holds also for $N\in\mathbb{R}_+$. To state \eqref{supadd1gen} for the case \eqref{remerr2}, the summation under the first integral in \eqref{supadd1gen} is taken over $A\subset\Lambda_i$ and over $i=1,\ldots,k$.
\end{lemma}
\begin{proof}
\textbf{Step 1.}  We will show here the following inequality, independent of the specific properties of $\mu:$
\begin{equation}
 \label{simplifdist}
  E_{N,s}^\mathrm{xc}\left(\mu\right)\ge   E_{\mathrm{GC},N,s}^\mathrm{xc}\left(\mu\right)\ge E_{N, \mathsf{c}-\frac{M}{M+C}w}^\mathrm{xc}(\mu)-\frac{M}{M+C}\left(\frac{C(w,d,\epsilon)}{M}N^{1+s/d}\int_{\mathbb{R}^d}\rho^{1+s/d}(x)dx+\frac{C(w,d,\epsilon)}{M}R_1^{-s}(N-1)\right)\ ,
\end{equation}
where $C, C(w,d,\epsilon)>0$ are the constants from Proposition \ref{prop3ws}, and where from \eqref{recallsplit}
\begin{equation}\label{recallsplitagain0}
\mathsf{c}-\frac{M}{M+C}w=\frac{M}{M+C}\int_{\Omega_l} \left(\sum_{A\in F^l_\omega}\frac{1_A(x)1_A(y)}{|x-y|^s}\right) d\,\mathbb{P}_l(\omega).
\end{equation}

\par The proof of \eqref{simplifdist} follows immediately from the definitions of $E_{N,s}^\mathrm{xc}\left(\mu\right)$ and $E_{\mathrm{GC},N,s}^\mathrm{xc}\left(\mu\right)$, by making use of Proposition \ref{prop3ws} and of the re-expression $\mathsf{c}(x,y)=\mathsf{c}(x,y)-\frac{M}{M+C}w(x-y)+\frac{M}{M+C}w(x-y)$.

\par\textbf{Step 2.} We will show here that
\begin{multline}\label{nobdcond}
E_{\mathrm{GC},N,s}^\mathrm{xc}\left(\mu\right)\ge \frac{M}{M+C}\bigg\{ \int_{\Omega_l} \bigg(\sum_{\substack{A\in F^l_\omega\\ A\subset\Lambda,\mu(A)>0}} {E_{\mathrm{GC},N\mu(A),s}^\mathrm{xc}}(\hat\mu_{ A})+E_{\mathrm{GC},N\mu(\Lambda_\mathrm{err}(l,\omega)),s}^\mathrm{xc}(\hat\mu_{ \Lambda_\mathrm{err}(l,\omega)})\bigg) d\,\mathbb{P}_l(\omega)\\
-\frac{C(w,d,\epsilon)}{M}N^{1+s/d}\int_{\mathbb{R}^d}\rho^{1+s/d}(x)dx-\frac{C(w,d,\epsilon)}{M}R_1^{-s}(N-1)\bigg\}\ .
\end{multline}
To start the proof, we observe that $\mu(\Lambda\setminus \left(\Lambda_\mathrm{err}(l,\omega)\cup \Lambda_\mathrm{nc}(l,\omega)\right))>0$ holds whenever there exists at least one $A\in \mathcal F^l_\omega$ such that $A\subset \Lambda\setminus \left(\Lambda_\mathrm{err}(l,\omega)\cup \Lambda_\mathrm{nc}(l,\omega)\right)$ and $\mu(A)>0$.
Denote for all $A\in F^l_\omega$ by
\[
\mathsf{c}_{A}(x,y):=\frac{1_A(x)1_A(y)}{|x-y|^s}\ .
\]
With this notation, we have in view of \eqref{recallsplitagain0} 
\begin{eqnarray}\label{reduc1}
F_{N, \mathsf{c}-\frac{M}{M+C}w}^\mathrm{OT}(\mu)\ge\frac{M}{M+C}\int_{\Omega_l}F_{n, \sum_{A\in F^l_\omega}\mathsf{c}_{A}}^\mathrm{OT}(\mu)\,d\,\mathbb{P}_l(\omega)\ ,
\end{eqnarray}

where for the inequality we interchanged the integration order, took the minimum inside the integral, and used the definition of $F_{n, \sum_{A\in F^l_\omega}\mathsf{c}_{A}}^\mathrm{OT}(\mu)$.

\par Directing next our attention briefly to the mean field term, we get
\begin{multline}\label{prodmesreduc}
\int_{\mathbb R^{2d}} \left(\mathsf{c}(x,y)-\frac{M}{M+C}w(x,y)\right) \, d\mu(x)d\mu(y) \,d\,\mathbb{P}_l(\omega)\\
=\frac{M}{M+C}\int_{\Omega_l}\int_{\mathbb R^{2d}}\sum_{A\in F^l_\omega}\mathsf{c}_A(x,y)d\mu(x)d\mu(y) \,d\,\mathbb{P}_l(\omega)\ .
\end{multline}
From \eqref{reduc1} and \eqref{prodmesreduc}, we obtain in \eqref{simplifdist}
\begin{multline}\label{simplifdist1a}
  E_{\mathrm{GC}, N,s}^\mathrm{xc}\left(\mu\right)\\
  \ge  \frac{M}{M+C}\left(\int_{\Omega_l}E^\mathrm{xc}_{\mathrm{GC},N, \sum_{A\in F^l_\omega}\mathsf{c}_{A}}(\mu)\,d\,\mathbb{P}_l(\omega)-\frac{C(w,d, \epsilon)}{M}N^{1+s/d}\int_{\mathbb{R}^d}\rho^{1+s/d}(x)dx- \frac{C(w,d,\epsilon)}{M}R_1^{-s}(N-1)\right)\ .
\end{multline}
Furthermore
\[
\mu=\mu|_{\Lambda\setminus \left(\Lambda_\mathrm{err}(l,\omega)\cup \Lambda_\mathrm{nc}(l,\omega)\right)}+\mu|_{\Lambda_\mathrm{err}(l,\omega)\cup \Lambda_\mathrm{nc}(l,\omega)}\ ,
\]
With the notation \eqref{hm}, we now have on the right-hand side of \eqref{simplifdist1a} by Corollary \ref{subadd12gen11} (using the notations $E^\mathrm{xc}_\mathrm{GC}[N, \mathsf{c}](\mu):=E^\mathrm{xc}_{\mathrm{GC}, N,\mathsf{c}}(\mu)$, in order to make the formulas easier to read):
\begin{eqnarray}\label{simplifdist1ab}
E^\mathrm{xc}_{\mathrm{GC}}\bigg[N, \sum_{A\in F^l_\omega}\mathsf{c}_{A}\bigg]( \mu)&\ge& E_\mathrm{GC}^\mathrm{xc}\bigg[N\mu(\Lambda_\mathrm{err}(l,\omega)\cup \Lambda_\mathrm{nc}(l,\omega)), \sum_{\substack{A\in F^l_\omega\\ A\subset \Lambda_\mathrm{err}(l,\omega)}}\mathsf{c}_{A}\bigg](\hat\mu_{\Lambda_\mathrm{err}(l,\omega)\cup \Lambda_\mathrm{nc}(l,\omega)})\nonumber\\
&&+E_\mathrm{GC}^\mathrm{xc} \bigg[N\mu(\Lambda\setminus \left(\Lambda_\mathrm{err}(l,\omega)\cup \Lambda_\mathrm{nc}(l,\omega)\right)), \sum_{\substack{A\in F^l_\omega\\ A\subset\Lambda}}c_{A}\bigg]\left( \hat\mu_{\Lambda\setminus\left(\Lambda_\mathrm{err}(l,\omega)\cup \Lambda_\mathrm{nc}(l,\omega)\right)}\right)\ .\nonumber\\
\end{eqnarray}
 By means of \eqref{supadd1gen0eqGC} we get for the first term on the right hand side of \eqref{simplifdist1ab},
\[
E_\mathrm{GC}^\mathrm{xc}\bigg[N\mu(\Lambda_\mathrm{err}(l,\omega)\cup \Lambda_\mathrm{nc}(l,\omega)), \sum_{\substack{A\in F^l_\omega\\ A\subset\Lambda_\mathrm{err}(l,\omega)}}\mathsf{c}_{A}\bigg]\left(\hat\mu_{\Lambda_\mathrm{err}(l,\omega)\cup \Lambda_\mathrm{nc}(l,\omega)}\right)= E_\mathrm{GC}^\mathrm{xc}\bigg[N\mu(\Lambda_\mathrm{err}(l,\omega)),\mathsf{c}\ \bigg](\hat\mu_{ \Lambda_\mathrm{err}(l,\omega)})\ ,
\]
which integrated against $d\mathbb P_\ell(\omega)$ gives the second term in \eqref{nobdcond}. The simplification in the last equation appears because the cost $\sum_{A\in F^l_\omega}c_{A}$ is zero on $\Lambda_\mathrm{nc}(l,\omega)$ due to \eqref{sigmaerr0}.

By notation \eqref{hm} and via \eqref{remerr1} (respectively \eqref{remerr2}), we get
\begin{eqnarray*}
\hat\mu_{\Lambda\setminus \left(\Lambda_\mathrm{err}(l,\omega)\cup \Lambda_\mathrm{nc}(l,\omega)\right)}&=&\sum_{\substack{A\in F^l_\omega \\ A\subset\Lambda,\mu(A)>0}} \hat\mu_{\Lambda\setminus \left(\Lambda_\mathrm{err}(l,\omega)\cup \Lambda_\mathrm{nc}(l,\omega)\right)}( A)\frac{\hat\mu_{\Lambda\setminus  \left(\Lambda_\mathrm{err}(l,\omega)\cup \Lambda_\mathrm{nc}(l,\omega)\right)}|_A}{\hat\mu_{\Lambda\setminus \left(\Lambda_\mathrm{err}(l,\omega)\cup \Lambda_\mathrm{nc}(l,\omega)\right)}(A)}\\
&=&\sum_{\substack{A\in F^l_\omega\\ A\subset\Lambda,\mu(A)>0}}\hat\mu_{\Lambda\setminus \left(\Lambda_\mathrm{err}(l,\omega)\cup \Lambda_\mathrm{nc}(l,\omega)\right)}(A)\,\hat\mu_A.
\end{eqnarray*}
For the second term on the r.h.s. in \eqref{simplifdist1ab} we have
\begin{multline}\label{splithatmu}
E_\mathrm{GC}^\mathrm{xc}\bigg[N\mu(\Lambda\setminus  \left(\Lambda_\mathrm{err}(l,\omega)\cup \Lambda_\mathrm{nc}(l,\omega)\right)), \sum_{\substack{A\in F^l_\omega\\ A\subset\Lambda}}\mathsf{c}_{A}\,\bigg]\left(\hat\mu_{\Lambda\setminus  \left(\Lambda_\mathrm{err}(l,\omega)\cup \Lambda_\mathrm{nc}(l,\omega)\right)}\right)\\
=E_\mathrm{GC}^\mathrm{xc}\bigg[N\mu(\Lambda\setminus  \left(\Lambda_\mathrm{err}(l,\omega)\cup \Lambda_\mathrm{nc}(l,\omega)\right)), \sum_{\substack{A\in F^l_\omega\\ A\subset\Lambda}}\mathsf{c}_A\,\bigg]\bigg(\sum_{\substack{A\in F^l_\omega\\ A\subset\Lambda,\mu(A)>0}} \hat\mu_{\Lambda\setminus \left(\Lambda_\mathrm{err}(l,\omega)\cup \Lambda_\mathrm{nc}(l,\omega)\right)}(A){\hat\mu}_{A}\bigg)\ .
\end{multline}
We now obtain on the right-hand side of \eqref{splithatmu} via another application of Corollary \ref{subadd12gen11}
\begin{multline*}
E_\mathrm{GC}^\mathrm{xc}\bigg[N\mu(\Lambda\setminus  \left(\Lambda_\mathrm{err}(l,\omega)\cup \Lambda_\mathrm{nc}(l,\omega)\right)), \sum_{\substack{A\in F^l_\omega\\ A\subset\Lambda}}\mathsf{c}_A\,\bigg]\bigg(\sum_{\substack{A\in F^l_\omega\\ A\subset\Lambda,\mu(A)>0}}\hat\mu_{\Lambda\setminus  \left(\Lambda_\mathrm{err}(l,\omega)\cup \Lambda_\mathrm{nc}(l,\omega)\right)}(A){\hat\mu}_{A}\bigg)\\
\ge\sum_{\substack{A\in F^l_\omega\\ A\subset\Lambda,\mu(A)>0}}\,E_\mathrm{GC}^\mathrm{xc}\bigg[N\mu(A),\mathsf{c}\, \bigg]\left({\hat\mu}_{A}\right)\ .
\end{multline*}
Integrated against $d\mathbb P_\ell(\omega)$, this gives the first term in \eqref{nobdcond} and concludes the proof of Step 2, for our choice of $\mathsf{c}(x,y)=|x-y|^{-s}$.

\end{proof}
We are now ready to prove Proposition \ref{mixextlowb} concerning our sharp lower bound.
\begin{proof}[Proof of Proposition \ref{mixextlowb}:]
Assume the quantities below are as in Lemma  \ref{cheeselemma}. Assume also that $[-\ell/2,\ell/2]^d\subset\Lambda_i,\ell\ll |\Lambda_i|,i=1,\ldots,k$ (or else we can re-scale the $\mu_i$ in view of (\ref{scalingexc})). By Lemma \ref{subadd12} we have
\begin{multline}\label{applysubadd12}
E_{\mathrm{GC},N,s}^\mathrm{xc}\left(\mu\right)\ge\frac{M}{M+C}\bigg\{\sum_{i=1}^k\int_{\Omega_l} \bigg(\sum_{\substack{A\in F^l_\omega\\ A\subset\Lambda_i\\\mu(A)>0}} E_{\mathrm{GC},N\mu(A),s}^\mathrm{xc}(\hat\mu_{ A})+E_{\mathrm{GC},N\mu(\Lambda_\mathrm{err}(l,\omega)),s}^\mathrm{xc}(\hat\mu_{ \Lambda_\mathrm{err}(l,\omega)})\bigg) d\,\mathbb{P}_l(\omega)\\
-\frac{C(w,d,\epsilon)}{M}N^{1+s/d}\int_{\mathbb{R}^d}\rho^{1+s/d}(x)dx-\frac{C(w,d,\epsilon)}{M}R_1^{-s}(N-1)\bigg\}\ .
\end{multline}
In the next two steps we will calculate separately each of the terms $E_{\mathrm{GC},N\mu(\Lambda_\mathrm{err}(l,\omega)),s}^\mathrm{xc}(\hat\mu_{ \Lambda_\mathrm{err}(l,\omega)})$ and $E^\mathrm{xc}_{\mathrm{GC},N\mu(A),s}({\hat\mu}_{A})$ appearing in \eqref{applysubadd12}. In order to prove that the boundary terms have a small total volume, we will need to work in the regime $l\ll 1, M\gg 1$. Lemma \ref{cheeselemma} applies in particular if $R_M>(1+4\sqrt{d}|B^d_1|)^M R_1,$ and $l> 8\sqrt d |B_1|(M+C_d) R_M$. Thus, for Lemma \ref{cheeselemma} to apply, the extra constraint linking $M,l,R_1$, which can be formulated in two equivalent ways:
\begin{equation}\label{boundlmr}
l>C(M+C)R_M>C(M+C)C^M\ R_1 \quad\Leftrightarrow\quad \log\frac{l}{R_1}>\log C + \log(M+C) + M\log C\ ,
\end{equation}
 where $C:=\max\{1+4\sqrt{d}|B^d_1|,8\sqrt d |B_1|, C_d\}$ depends only on $d$.
If $R_1,l,M$ are such that
\begin{equation}\label{boundlmr2}
M< \frac{\log(l/R_1)}{3\log C} \quad\Leftrightarrow\quad R_1< C^{-3M}\ l\ ,
\end{equation}
then there exists $M_d>0$ depending only on $d$ such that \eqref{boundlmr} holds for all $M\ge M_d$. Indeed, note that if
\begin{equation}\label{boundM}
M\ge\max\left\{1,\tfrac{\log(M+C)}{\log C}\right\},
\end{equation}
then $3M\log C$ is larger than the right hand side of the second equation in \eqref{boundlmr}, and as a consequence \eqref{boundlmr2} implies \eqref{boundlmr} for such $M$. It suffices then to take $M_d$ to be the smallest value of $M\ge1$ such that \eqref{boundM} holds. It is easy to verify that for any $M\ge M_d$ \eqref{boundM} also holds, that the value $M_d$ depends only on $d$ because $C$ above depends only on $d$.

To find $l,M,R_1$ satisfying condition \eqref{boundlmr2} and also such that $l\ll1\ll M$, it suffices to fix $R_1>0$ separately for each choice of $l,M$ and small enough. Note that the above choices must be performed depending on $N$ as well, so as to optimize our asymptotic estimates. 

A suitable choice turns out to be as follows. We take for $N\ge C^{18d(M_d+1)}$ 
\begin{equation}\label{boundlmr3}
R_1:=N^{-\frac1{2d}},\quad l:=N^{-\frac1{3d}},\quad M:=\left[\frac{\log N}{18\ d\log C}\right]-1\ ,
\end{equation}
For clarity of exposition, we will substitute the choice \eqref{boundlmr3} for $l, R_1$ only at the very end in our estimates.

\par\textbf{Step 1.}  We will show here that
\begin{equation}\label{applysubadd12aa}
\liminf_{N\rightarrow\infty} N^{-1-s/d} E_{\mathrm{GC},N,s}^\mathrm{xc}\left(\mu\right)\ge\liminf_{N\rightarrow\infty}N^{-1-s/d} \sum_{i=1}^k\int_{\Omega_l}\bigg(\sum_{\substack{A\in F^l_\omega\\ A\subset \Lambda_i}} {E^\mathrm{xc}_{\mathrm{GC},N\alpha_i|A|/|\Lambda_i|,s}}({\hat\mu}_{A})\bigg)\,d\,\mathbb{P}_l(\omega)\ .
\end{equation}
Note that, since $l\ll 1$, and since all sets $A\in F_\omega^l$ are disjoint balls of scales $R_j<l$, for given $A\in F_\omega^l$, either the set $A$ intersects the boundaries $\partial\Lambda_i, i=1,\ldots k$, or it is included exactly in one of the $\Lambda_i, i=1,\ldots,k$. Thus, from \eqref{applysubadd12} and in view of the properties $M\gg 1$ and $R_1\gg N^{-\frac1d}$ following from \eqref{boundlmr3}, we find 
\begin{multline}\label{applysubadd12+}
\liminf_{N\rightarrow\infty}N^{-1-s/d}E_{\mathrm{GC},N,s}^\mathrm{xc}(\mu)\ge\liminf_{N\rightarrow\infty}N^{-1-s/d}\sum_{i=1}^k\int_{\Omega_l}\bigg(\sum_{\substack{A\in F^l_\omega\\ A\subset\Lambda_i}} E_{\mathrm{GC},N\mu(\Lambda\cap A),s}^\mathrm{xc}(\hat\mu_A)\bigg) d\,\mathbb{P}_l(\omega)\\
+\liminf_{N\rightarrow\infty}N^{-1-s/d}\int_{\Omega_l}E_{\mathrm{GC},N\mu(\Lambda_\mathrm{err}(l,\omega)),s}^\mathrm{xc}(\hat\mu_{ \Lambda_\mathrm{err}(l,\omega)})\, d\,\mathbb{P}_l(\omega)\ .
\end{multline}
We next need to consider separately the two terms on the r.h.s in (\ref{applysubadd12+}). We observe first that $\mu(A)=\alpha_i\frac{|A|}{|\Lambda_i|}$ if $A\subset \Lambda_i$ for some $i=1,\ldots, k$. Therefore, in this case
\[
E_{\mathrm{GC},N\mu(A),s}^\mathrm{xc}(\hat\mu_{ A})=E^\mathrm{xc}_{\mathrm{GC},N\alpha_i|A|/|\Lambda_i|,s}({\hat\mu}_{A})\ .
\]
We also have
\[
\mu( \Lambda_\mathrm{err}(l,\omega))=\sum_{i=1}^k\alpha_i\frac{|\Lambda_i\cap  \Lambda_\mathrm{err}(l,\omega)|}{|\Lambda_i|}\ ,
\]
and thus
\[
E_{\mathrm{GC},N\mu(\Lambda_\mathrm{err}(l,\omega)),s}^\mathrm{xc}(\hat\mu_{\Lambda_\mathrm{err}(l,\omega)}) 
= E^\mathrm{xc}_{\mathrm{GC},\sum_{i=1}^kN\alpha_i {\frac{|\Lambda_i\cap  \Lambda_\mathrm{err}(l,\omega)|}{|\Lambda_{i}|}},s}({\hat\mu}_{ \Lambda_\mathrm{err}(l,\omega)})\ .
\]
Therefore, \eqref{applysubadd12+} becomes
\begin{multline}\label{applysubadd12++}
\liminf_{N\rightarrow\infty} N^{-1-s/d}E_{\mathrm{GC},N,s}^\mathrm{xc}(\mu)\ge \liminf_{N\rightarrow\infty} N^{-1-s/d}\sum_{i=1}^k\int_{\Omega_l}\bigg(\sum_{\substack{A\in F^l_\omega\\ A\subset\Lambda_i}}  {E^\mathrm{xc}_{\mathrm{GC},N\alpha_i|A|/|\Lambda_i|,s}}({\hat\mu}_{A})\bigg)\,d\,\mathbb{P}_l(\omega)\\
+\liminf_{N\rightarrow\infty} N^{-1-s/d}\int_{\Omega_l} {E^\mathrm{xc}_{\mathrm{GC},\sum_{i=1}^kN\alpha_i {|\Lambda_i\cap  \Lambda_\mathrm{err}(l,\omega)|/|\Lambda_{i}|},s}}({\hat\mu}_{ \Lambda_\mathrm{err}(l,\omega)})\,d\,\mathbb{P}_l(\omega)\ .
\end{multline}
Consider now the second term on the right hand side of \eqref{applysubadd12++}. Using \eqref{hm} we obtain
\begin{eqnarray}
\label{roughbla}
\lefteqn{{E^\mathrm{xc}_{\mathrm{GC},\sum_{i=1}^kN\alpha_i {|\Lambda_i\cap \Lambda_\mathrm{err}(l,\omega)|/|\Lambda_{i}|},s}}({\hat\mu}_{\Lambda_\mathrm{err}(l,\omega)})}\nonumber\\
&\ge& -c_{\mathrm{LO}}(d,\epsilon)\int_{\Omega_l}\left(\sum_{i=1}^kN\alpha_i \frac{|\Lambda_i\cap \Lambda_\mathrm{err}(l,\omega)|}{|\Lambda_i|}\right)^{1+s/d}\frac{\sum_{i=1}^k\alpha_i^{1+s/d}\left(\frac{|\Lambda_i\cap \Lambda_\mathrm{err}(l,\omega)|}{|\Lambda_i|}\right)^{1+s/d}}{\left(\sum_{i=1}^k\alpha_i\frac{|\Lambda_i\cap\Lambda_\mathrm{err}(l,\omega)|}{|\Lambda_i|}\right)^{1+s/d}}\,d\,\mathbb{P}_l(\omega)\nonumber\\
&=&-c_{\mathrm{LO}}(d,\epsilon)N^{1+s/d}\int_{\Omega_l}\ \sum_{i=1}^k\alpha_i^{1+s/d}\left(\frac{|\Lambda_i\cap \Lambda_\mathrm{err}(l,\omega)|}{|\Lambda_i|}\right)^{1+s/d}\,d\,\mathbb{P}_l(\omega)\nonumber\\
&\ge& -c_{\mathrm{LO}}(d,\epsilon)N^{1+s/d}\int_{\Omega_l} \sum_{i=1}^k\phi(\ell|\Lambda_i|^{-1/d})\frac{\alpha_i^{1+s/d}}{|\Lambda_i|^{s/d}}\,d\,\mathbb{P}_l(\omega)\nonumber\\
&=& -c_{\mathrm{LO}}(d,\epsilon) C_\rho N^{1+s/d}\ell\int_{\mathbb{R}^d}\rho^{1+s/d}(x)dx,
\end{eqnarray}
where $C_\rho>0$. For the first inequality we applied the lower bound \eqref{loterm1gs}, which is analogous to \cite[(3.3)]{LewLiebSeir17} but uniform in $s\in (\epsilon, d-\epsilon)$. For the second inequality, we applied the domain boundary regularity estimates \eqref{481appr}, in which we observe that for a fixed hyperrectangle $\Lambda_i$, we can take $\phi(t)=C_{\Lambda_i} t$. 

From \eqref{roughbla}, it follows that
\begin{equation*}
\int_{\Omega_l}{E^\mathrm{xc}_{\mathrm{GC},\sum_{i=1}^kN\alpha_i |\Lambda_i\cap \Lambda_\mathrm{err}(l,\omega)|/|\Lambda_{i}|,s}}({\hat\mu}_{\Lambda_\mathrm{err}(l,\omega)})\,d\,\mathbb{P}_l(\omega)
\ge  -C_\rho C(d,\epsilon)\ell N^{1+s/d}\ .
 \end{equation*}
Together with \eqref{applysubadd12++}, the above proves \eqref{applysubadd12aa}.

\par\textbf{Step 2.} We will show here the statement of the theorem.

\par Firstly, since for $A\subset\Lambda_i, i=1,\ldots,k,$ we have
\[
\hat\mu_{A}=\frac{\alpha_i}{|\Lambda_i|}1_{A}\bigg/\frac{\alpha_i |A|}{|\Lambda_i|}=\frac{1_A}{|A|}\ ,
\]
we get from \eqref{applysubadd12aa}
\begin{multline}\label{applysubadd12aa1}
\liminf_{N\to\infty}N^{-1-s/d} E_{\mathrm{GC},N,s}^\mathrm{xc}(\mu)\ge \liminf_{N\to\infty} N^{-1-s/d}\ \sum_{i=1}^k \int_{\Omega_l}\bigg(\sum_{\substack{A\in F^l_\omega\\ A\subset\Lambda_i}}  {E^\mathrm{xc}_{\mathrm{GC},N\alpha_i|A|/|\Lambda_i|,s}}\left(\frac{1_A}{|A|}\right)\bigg) \,d\,\mathbb{P}_l(\omega)\ .
\end{multline}
Next we follow in some more detail the dependency of the sets $A\in F^l_\omega$. More precisely, fix $i=1,\dots,k.$ Any $A\in F_\omega^l$ is of the form $B_R(x)$ for some choice $R\in\{tR_1,\ldots, tR_M\}$, $t\in [1-\eta,1+\eta]$ and $x\in\mathbb R^d$. In this case we find, in view of \eqref{scalinggc}
\[
E^\mathrm{xc}_{\mathrm{GC},N\alpha_i|A|/|\Lambda_i|,s}\left(\frac{1_A}{|A|}\right)={E^\mathrm{xc}_{\mathrm{GC},N\alpha_i|B_{R}|/|\Lambda_i|,s}}\left(\frac{1_{B_R}}{|B_R|}\right)\ ,
\]
where we recall that $B_R$ is the ball of radius $R$ centred at $0$. The proof of this statement follows by a change of variables idea, done by means of Proposition \ref{depascdual}, and will be omitted.
Furthermore, by the same type of argument we can further reduce to the unit ball $B_1$:
\[
E^\mathrm{xc}_{\mathrm{GC},N\alpha_i|B_R|/|\Lambda_i|,s}\left(\frac{1_{B_R}}{|B_R|}\right)=R^{-s}{E^\mathrm{xc}_{\mathrm{GC},N\alpha_i|B_R|/|\Lambda_i|,s}}\left(\frac{1_{B_1}}{|B_1|}\right)\ .
\]
Plugging this in \eqref{applysubadd12aa1}, we get
\begin{multline}\label{furtherexplim}
\liminf_{N\to\infty}N^{-1-s/d} E_{\mathrm{GC},N,s}^\mathrm{xc}(\mu)
\\
\ge \liminf_{N\to\infty} N^{-1-s/d}\sum_{i=1}^k \int_{\Omega_l}\bigg(\sum_{\substack{A\in F_\omega^l:\ \exists x\in\mathbb{R}^d,\\ A=B_R(x)\subset\Lambda_i}} R^{-s}E^\mathrm{xc}_{\mathrm{GC},N\alpha_i|B_R|/|\Lambda_i|,s}\left(\frac{1_{B_1}}{|B_1|}\right)\bigg)\,d\,\mathbb{P}_l(\omega)\ .
\end{multline}
Next, we will use that by \eqref{egclimitunif} we have for every fixed $i=1,\ldots,k$ 
\begin{equation}\label{limitbla}
\lim_{N\to\infty}\left(N\alpha_iR^d/|\Lambda_i|\right)^{-1-s/d}{E^\mathrm{xc}_{\mathrm{GC},N\alpha_i|B_R|/|\Lambda_i|,s}}\left(\frac{1_{B_1}}{|B_1|}\right)=-C(s,d)|B_1|\ ,\quad\mbox{where}\quad C(s,d)>0\ ,
\end{equation}
which will allow us to take the limits in \eqref{furtherexplim}, for the sums under the integral, uniformly in $N$, $t$ and $A\in F^l_\omega$. Using \eqref{limitbla} in \eqref{furtherexplim}, we get for $i=1,\ldots,k,$ for arbitrary $\delta>0$ and large $N$
\begin{eqnarray}
\lefteqn{\int_{\Omega_l}\bigg(\sum_{\substack{A\in F_\omega^l:\ \exists x\in\mathbb{R}^d,\\ A=B_R(x)\subset\Lambda_i}} R^{-s}{E^\mathrm{xc}_{\mathrm{GC},N\alpha_i|B_R|/|\Lambda_i|,s}}\left(\frac{1_{B_1}}{|B_1|}\right)\bigg)\,d\,\mathbb{P}_l(\omega)}\nonumber\\
&\ge&\int_{\Omega_l}\bigg(\sum_{\substack{A\in F_\omega^l:\ \exists x\in\mathbb{R}^d,\\ A=B_R(x)\subset\Lambda_i}} R^{-s}|B_1|\left(N\alpha_i{R}^d/|\Lambda_i|\right)^{1+s/d}\left(-C(s,d)-\delta\right)\bigg)\,d\,\mathbb{P}_l(\omega)\nonumber\\
&\ge&\int_{\Omega_l}\bigg(\sum_{\substack{A\in F_\omega^l:\ \exists x\in\mathbb{R}^d,\\ A=B_R(x)\subset\Lambda_i}} N^{1+s/d}\frac{\alpha_i^{1+s/d}}{|\Lambda_i|^{1+s/d}}R^d|B_1|\bigg)\left(-C(s,d)-\delta\right)\,d\,\mathbb{P}_l(\omega)\nonumber\\
&=&\left(-C(s,d)-\delta\right)N^{1+s/d}\left(\frac{\alpha_i}{|\Lambda_i|}\right)^{1+s/d}\int_{\Omega_l}\bigg(\sum_{\substack{A\in F_\omega^l:\ \exists x\in\mathbb{R}^d,\\ A=B_R(x)\subset\Lambda_i}}R^d|B_1|\bigg)\,d\,\mathbb{P}_l(\omega)\ ,
\label{useuniffurther}
\end{eqnarray}
where for the first inequality we used \eqref{limitbla} and in order to remove the integer part and obtain the second inequality, we used the fact that in view of \eqref{boundlmr3} we have $R_1\gg N^{-1/d}$ and thus  $((1-\eta)R_1)^dN\gg1$ and thus automatically due to the choices of $t, R_j$, also $R^dN\gg 1$ for any $R\in\{tR_1,\ldots,tR_M\}$ and any $t\in[1-\eta,1+\eta]$. 

\par For bounding \eqref{useuniffurther} we use the fact that $C(s,d)>0$ and the following bound:
\begin{equation}\label{useuniffurther2}
\sup_{\omega\in \Omega_l}\sum_{\substack{A\in F_\omega^l:\ \exists x\in\mathbb{R}^d,\\ A=B_R(x)\subset\Lambda_i}} R^d|B_1|=\sup_{\omega\in \Omega_l}\sum_{\substack{A\in F^l_\omega\\ A\subset\Lambda_i}}|A|\le|\Lambda_i|\ .
\end{equation}
Making use of \eqref{useuniffurther} and \eqref{useuniffurther2} in \eqref{furtherexplim} we find that with the choices \eqref{boundlmr3} we have
\begin{multline*}
\liminf_{N\to\infty}N^{-1-s/d} E_{\mathrm{GC},N,s}^\mathrm{xc}(\mu)
\ge\liminf_{N\to\infty}\left(-C(s,d)-\delta\right)\sum_{i=1}^k\left(\frac{\alpha_i}{|\Lambda_i|}\right)^{1+s/d}\int_{\Omega_l}\bigg(\sum_{\substack{A\in F_\omega^l:\ \exists x\in\mathbb{R}^d,\\ A=B_R(x)\subset\Lambda_i}} R^d|B_1|\bigg)\,d\,\mathbb{P}_l(\omega)\\
\ge\left(-C(s,d)-\delta\right)\sum_{i=1}^k\left(\frac{\alpha_i}{|\Lambda_i|}\right)^{1+s/d}|\Lambda_i|=\left(-C(s,d)-\delta\right)\int_{\mathbb R^d}\left(\sum_{i=1}^k\alpha_i\rho_i(x)\right)^{1+s/d}dx\ .
\end{multline*}
Taking now $\delta\rightarrow 0$ in the above proves the statement of our Proposition.
\end{proof}
\subsection{Proof of Theorem \ref{upboundcont}}\label{ssecproofmainthm}
\begin{proof}[Proof of Theorem \ref{upboundcont}:]
	
\par \textbf{Step 1.} Let $\rho$ be continuous and with compact support $\mathrm{supp}(\mu)=:\Lambda\subset\mathbb{R}^d$. Fix $l>0$ such that $[-\tfrac{l}2, \tfrac{l}2]^d$ is much smaller than $\Lambda$. Take as in Lemma \ref{cheeselemma} a Swiss cheese packing of $[-\tfrac{l}2, \tfrac{l}2]^d$ by balls of radii ${R_1},\ldots, {R_M}$, extended by periodicity to the whole $\mathbb{R}^d$. We also assume below that $R_1,l,M,$ satisfy \eqref{boundlmr2}, but without the specific $N$-dependent choices \eqref{boundlmr3}. At the end of the proof we will further take $l\to 0, M\to\infty,$ and the discussion is going to be independent of the choice of $R_1>0$, as long as $R_1<C^{-3M}l$ as in \eqref{boundlmr2}.              

We write, recalling definition \eqref{sigmaerr0} for the second equality (now used for one single covering family, i.e. without the $\omega$-dependence in \eqref{sigmaerr0})
\begin{equation*}
\rho=\rho|_{\cup_{j=1}^M\cup_{A\in {\cal B}^l_{R_j}\atop A\cap\Lambda\neq\emptyset}A}+\rho|_{\Lambda\setminus \big(\cup_{j=1}^M\cup_{A\in {\cal B}^l_{R_j}\atop A\cap\Lambda\neq\emptyset}A\big)}=\rho|_{\Lambda\setminus\Lambda_\mathrm{nc}}+\rho|_{\Lambda_\mathrm{nc}},
\end{equation*}
where we denoted by ${\cal B}^l_{R_j}$ the set of balls of radius $R_j$ from the Swiss cheese packing. Define for all $x\in\Lambda$ the sequences
\begin{equation}
\label{4106}
\rho^l_{\mathrm{min}}(x):=\frac{1}{k^l_M}\sum_{j=1}^M\sum_{A\in {\cal B}^l_{R_j}\atop A\cap\Lambda\neq\emptyset}\left(\min_{x\in A\cap\Lambda}\rho(x)\right) 1_{A\cap\Lambda}(x)~~~\mbox{and}~~~\rho^l_{\mathrm{max}}(x):=\frac{1}{m^l_M}\sum_{j=1}^M\sum_{A\in {\cal B}^l_{R_j}\atop A\cap\Lambda\neq\emptyset}\left(\max_{x\in A\cap\Lambda}\rho(x)\right) 1_{A\cap\Lambda}(x),
\end{equation}
where $k^l_M$ and $m^l_M$ are normalization constants required to make $\rho^l_\mathrm{min},\rho^l_\mathrm{max}$ probability densities and where we set $\rho_\mathrm{min}^l=0$ if $\min_{x\in A\cap\Lambda}\rho(x)=0$ for all $A\in\bigcup_{j=1}^M\mathcal B_{R_j}^l$. Then
\begin{equation}
\label{4107}
\sum_{j=1}^M\sum_{A\in {\cal B}^l_{R_j}\atop A\cap\Lambda\neq\emptyset}\left(\min_{x\in A\cap\Lambda}\rho(x)\right) 1_{A\cap\Lambda}(x)\le\rho(x)|_{\cup_{j=1}^M\cup_{A\in {\cal B}^l_{R_j}\atop A\cap\Lambda\neq\emptyset}A}=\rho|_{\Lambda\setminus \Lambda_\mathrm{nc}}\le \sum_{j=1}^M\sum_{A\in {\cal B}^l_{R_j}\atop A\cap\Lambda\neq\emptyset}\left(\max_{x\in A\cap\Lambda}\rho(x)\right) 1_{A\cap\Lambda}(x),
\end{equation}
 and we claim that $\rho^l_{\mathrm{min}}, \rho^l_{\mathrm{max}}$ converge strongly in $L^{1+s/d}(\mathbb{R}^d)$ to $\rho$ as $l\to 0, M\to\infty$.  
	
We prove only the convergence $\lim_{l\to 0}\left\|\rho-\rho^l_\mathrm{min}\right\|_{L^{1+s/d}(\mathbb{R}^d)}=0$, as the convergence of $\rho^l_\mathrm{max}$ is proved similarly. Note first that as $\rho$ is continuous and compactly-supported, therefore it is uniformly continuous. Thus, for all $\delta>0$ there exists $l_\delta>0$ such that for all $l\le l_\delta$ and all $A$ in a given packing of $[-\tfrac{l}2, \tfrac{l}2]^d$, we have $\max_{x\in A\cap\Lambda}\rho(x)-\min_{x\in A\cap\Lambda}\rho(x)\le\delta$. Hence, for $l\le l_\delta$
$$m^l_M-k^l_M\le \delta \sum_{j=1}^M\sum_{\substack{A\in {\cal B}^l_{R_j}\\ A\cap\Lambda\neq\emptyset}}\int_{\mathbb{R}^d} 1_{A\cap\Lambda}(x)\le \delta |\Lambda|~~\to 0~~\mbox{as}~~l,\delta\to 0.$$
Furthermore, for $l\le\l_\delta$ we have
\begin{eqnarray}
\label{4108}
\lefteqn{\left(\int_{\mathbb{R}^d}\left(\rho(x)-k^l_M\rho^l_{\mathrm{min}}(x)\right)^{1+s/d}dx\right)^{\frac{1}{1+s/d}}}\nonumber\\
&\le& \bigg(\int_{\mathbb{R}^d}\bigg(\sum_{j=1}^M\sum_{A\in {\cal B}^l_{R_j}\atop A\cap\Lambda\neq\emptyset}\left(\max_{x\in A\cap\Lambda}\rho(x)-\min_{x\in A\cap\Lambda}\rho(x)\right) 1_{A\cap\Lambda}(x)\bigg)^{1+s/d}dx\bigg)^{\frac{1}{1+s/d}}+\bigg(\int_{\Lambda_{\mathrm{nc}}(l)}\big(\max_{y\in\Lambda}\rho(y)\big)^{1+s/d}dx\bigg)^{\frac{1}{1+s/d}}\nonumber\\
&\le&\delta |\Lambda|^{1+s/d}+\big(\max_{y\in\Lambda}\rho(y)\big)\,|\Lambda_{\mathrm{nc}}(l)|^{\frac{1}{1+s/d}},
\end{eqnarray}
which tends to $0$ as $l,\delta\to 0$ and $M\to\infty$, in view of (\ref{480appr}) and of the boundedness of $\Lambda$. As $\rho$ and $\rho^l_{\mathrm{min}}$ are probability densities, (\ref{4108}) gives in particular that $k^l_M\to 1$ as $l\to 0$. Then, by means of (\ref{4108}), we obtain
\begin{eqnarray*}
\left\|\rho-\rho^l_{\mathrm{min}}\right\|_{L^{1+s/d}}\le\left(\int_{\mathbb{R}^d}\left(\rho(x)-k^l_M\rho^l_{\mathrm{min}}(x)\right)^{1+s/d}dx\right)^{\frac{1}{1+s/d}}+(1-k^l_M)\left(\int_{\mathbb{R}^d} \left(\rho^l_{\mathrm{min}}(x)\right)^{1+s/d}dx\right)^{\frac{1}{1+s/d}},
\end{eqnarray*}
which tends to $0$ as $l\to 0$ and $M\to\infty$. This proves the claim.

\par \textbf{Step 2.} We assume here that $\rho$ is continuous and $\rho\in L^{1+s/d}(\mathbb{R}^d)$, and we prove the lower bound. The proof of the upper bound follows similarly from an argument similar to (\ref{subadd4}) and by means of Step 1.  Take $0<l\ll 1$. Then for a fixed compact Borel measurable set $\Lambda$ with nonempty interior, such that $\Lambda\subseteq\mathrm{supp}(\mu)$ and $\mu(\Lambda)>0$, we denote $\widetilde \Lambda:=\mathrm{supp}(\mu)\setminus \Lambda$ and, with the notations \eqref{sigmaerr0}, \eqref{sigmaerr1}, we use the splitting
\begin{equation}\label{4109}
\mu=\mu|_{\Lambda\setminus \Lambda_\mathrm{nc}(l,\omega)}+\mu|_{\Lambda_\mathrm{nc}(l,\omega)}+\mu|_{\widetilde\Lambda\setminus\widetilde\Lambda_{nc}(l,\omega)}+\mu|_{\widetilde\Lambda_{nc}(l,\omega)}.
\end{equation}
By applying the same arguments as in Lemma \ref{subadd12}, by using Corollary \ref{subadd12gen11} and with the notation \eqref{hm}, we have by means of (\ref{simplifdist1a}) and (\ref{4109}) (and with the convention that $E_{\mathrm{GC},k,\mathsf{c}}^\mathrm{xc}(0)=0$)
\begin{eqnarray}\label{4110}
E_{\mathrm{GC}, N,s}^\mathrm{xc}\left(\mu\right)
&\ge&  \frac{M}{M+C}\bigg(\int_{\Omega_l}E^\mathrm{xc}_{\mathrm{GC},N\mu(\Lambda\setminus \Lambda_\mathrm{nc}(l,\omega)), \sum_{A\in F^l_\omega\atop A\cap \Lambda\neq\emptyset}\mathsf{c}_{A}}({\hat\mu}_{\Lambda\setminus \Lambda_\mathrm{nc}(l,\omega)})\,d\,\mathbb{P}_l(\omega)\nonumber\\
&&+\int_{\Omega_l}E^\mathrm{xc}_{\mathrm{GC},N\mu(\widetilde\Lambda\setminus\widetilde\Lambda_{nc}(l,\omega)), s}({\hat\mu}_{\widetilde\Lambda\setminus\widetilde\Lambda_{nc}(l,\omega)})\,d\,\mathbb{P}_l(\omega)-\frac{C(w,d,\epsilon)}{M}R_1^{-s}\left(N-1\right)\nonumber\\
&&-\frac{C(w,d, \epsilon)}{M}N^{1+s/d}\int_{\mathbb{R}^d}\rho^{1+s/d}(x)dx\bigg)\ .
\end{eqnarray}
We now apply the construction \eqref{4107} from Step 1, to the first term in (\ref{4110}), with $\hat\mu_{\Lambda\setminus\Lambda_{\mathrm{nc}}(l,\omega)}$ (respectively ${\hat\mu}_\Lambda$) instead of $\mu$ and with respect to the ball covering $F_\omega^l$, and we denote by $(\hat\rho_{\Lambda\setminus\Lambda_\mathrm{nc}(l,\omega)})^{l,\omega}_{\mathrm{max}}$ (respectively $({\hat\rho}_\Lambda)^{l,\omega}_{\mathrm{max}}$) the so-obtained densities. Explicitly, we have by using the notations \eqref{hm} and (\ref{4106})
\begin{eqnarray}
\label{rhomax}
(\hat\rho_{\Lambda\setminus\Lambda_\mathrm{nc}(l,\omega)})^{l,\omega}_{\mathrm{max}}(x):&=&\frac{1}{{\hat m}^{l,\omega}_{\Lambda\setminus\Lambda_\mathrm{nc}(l,\omega)}}  \sum_{A\in F^l_\omega\atop A\cap \Lambda\neq\emptyset} \left(\max_{x\in A\cap\Lambda}{\hat\rho}_{\Lambda\setminus \Lambda_\mathrm{nc}(\omega,l)}(x)\right) 1_{A\cap\Lambda}(x)\nonumber\\
&=&\frac{1}{{\hat m}^{l,\omega}_\Lambda}  \sum_{A\in F^l_\omega\atop A\cap \Lambda\neq\emptyset} \left(\max_{x\in A\cap\Lambda}{\hat\rho}_\Lambda(x)\right) 1_{A\cap\Lambda}(x)=({\hat\rho}_\Lambda)^{l,\omega}_{\mathrm{max}}(x),
\end{eqnarray}
where ${\hat m}^{l,\omega}_{\Lambda\setminus\Lambda_\mathrm{nc}(l,\omega)}, {\hat m}^{l,\omega}_{\Lambda}$ are normalization factors. Explicitly, we have
\begin{eqnarray}\label{hatm}
{\hat m}^{l,\omega}_{\Lambda\setminus\Lambda_{nc}(l,\omega)}:&=&\sum_{\substack{A\in F_\omega^l\\A\cap\Lambda\neq \emptyset}}|A\cap\Lambda|\max_{y\in A\cap\Lambda}\hat\rho_{\Lambda\setminus\Lambda_\mathrm{nc}(l,\omega)}(y)\nonumber\\
&=&\frac{1}{\mu(\Lambda\setminus\Lambda_\mathrm{nc}(l,\omega))}\sum_{\substack{A\in F_\omega^l\\A\cap\Lambda\neq \emptyset}}|A\cap\Lambda|\max_{y\in A\cap\Lambda}\rho(y)
=\frac{\mu(\Lambda)}{\mu(\Lambda\setminus\Lambda_\mathrm{nc}(l,\omega))} {\hat m}^{l,\omega}_\Lambda\ge 1.
\end{eqnarray}
We find directly from the first equality in (\ref{rhomax}) that $(\hat\mu_{\Lambda\setminus\Lambda_\mathrm{nc}(l,\omega)})^{l,\omega}_{\mathrm{max}}\ge \frac{1}{{\hat m}^{l,\omega}_{\Lambda\setminus \Lambda_{nc}(l,\omega)}} {\hat\mu}_{\Lambda\setminus \Lambda_\mathrm{nc}(l,\omega)}.$ By \eqref{subadd_gcbinit} from item 2 of Remark \ref{rmk_gcb}, applied with $\mu''=\hat\mu_{\Lambda\setminus\Lambda_\mathrm{nc}}, N''=N\mu(\Lambda\setminus\Lambda_\mathrm{nc}(l,\omega))$ and $N'+N''=N\mu(\Lambda){\hat m}^{l,\omega}_\Lambda$, and by applying \eqref{EGC_negative} to the resulting $\mu'$ we get
\begin{equation}\label{4111}
E^\mathrm{xc}_{\mathrm{GC},N\mu(\Lambda\setminus \Lambda_\mathrm{nc}(l,\omega)), \sum_{A\in F^l_\omega\atop A\cap \Lambda\neq\emptyset}\mathsf{c}_{A}}({\hat\mu}_{\Lambda\setminus \Lambda_\mathrm{nc}(l,\omega)})\ge E^\mathrm{xc}_{\mathrm{GC},N \mu(\Lambda){{\hat m}^{l,\omega}_\Lambda}, \sum_{A\in F^l_\omega\atop A\cap\Lambda\neq\emptyset}\mathsf{c}_{A}}((\hat\mu_{\Lambda\setminus\Lambda_\mathrm{nc}(l,\omega)})^{l,\omega}_{\mathrm{max}}).
\end{equation}
By plugging \eqref{4111} in \eqref{4110}, while the other terms remain unchanged, we obtain
\begin{eqnarray}\label{4112}
E_{\mathrm{GC}, N,s}^\mathrm{xc}\left(\mu\right)
&\ge&  \frac{M}{M+C}\bigg(\int_{\Omega_l}E^\mathrm{xc}_{\mathrm{GC},N\mu(\Lambda) {\hat m}^{l,\omega}_\Lambda, \sum_{A\in F^l_\omega\atop A\cap \Lambda\neq\emptyset}\mathsf{c}_{A}}((\hat\mu_{\Lambda\setminus\Lambda_\mathrm{nc}(l,\omega)})^{l,\omega}_{\mathrm{max}})d\,\mathbb P_l(\omega)\nonumber\\
&&+\int_{\Omega_l}E^\mathrm{xc}_{\mathrm{GC},N\mu(\widetilde\Lambda\setminus\widetilde\Lambda_{nc}(l,\omega)), s}({\hat\mu}_{\widetilde\Lambda\setminus\widetilde\Lambda_{nc}(l,\omega)})\,d\,\mathbb{P}_l(\omega)-\frac{C(w,d,\epsilon)}{M}R_1^{-s}\left(N-1\right)\nonumber\\
&&-\frac{C(w,d, \epsilon)}{M}N^{1+s/d}\int_{\mathbb{R}^d}\rho^{1+s/d}(x)dx\bigg)\ .
\end{eqnarray}
By using \eqref{hatm} and applying now to the first term in (\ref{4112}) a similar argument as the one to get (\ref{furtherexplim}) from the proof of Proposition \ref{mixextlowb}, but this time without any need to consider separately the terms with $A\subset\Lambda$ and those with $A\cap\partial\Lambda\neq\emptyset$, as $l$ is now independent of $N$, we get 
\begin{eqnarray}\label{4113}
\lefteqn{E_{\mathrm{GC}, N,s}^\mathrm{xc}\left(\mu\right)\ge \frac{M}{M+C}\bigg\{\int_{\Omega_l}\bigg(\sum_{\substack{A\in F_\omega^l\\ A\cap \Lambda\neq\emptyset}} E^\mathrm{xc}_{\mathrm{GC}, N|A\cap\Lambda|  \left(\max_{x\in A\cap\Lambda}\rho(x)\right),s}\left(\frac{1_{A\cap\Lambda}}{|A\cap\Lambda|}\right)\bigg)\,d\,\mathbb{P}_l(\omega)}\\
&&-c_{\mathrm{LO}}(d,\epsilon)N^{1+s/d}\int_{\mathbb R^d\setminus\Lambda} \rho^{1+s/d}(x)dx-\frac{C(w,d,\epsilon)}{M}R_1^{-s}\left(N-1\right)-\frac{C(w,d, \epsilon)}{M}N^{1+s/d}\int_{\mathbb{R}^d}\rho^{1+s/d}(x)dx\bigg\}.\nonumber
\end{eqnarray}
For the second term in (\ref{4113}) we applied (\ref{loterm1gs}) and used the fact that $\widetilde\Lambda\subseteq\mathbb R^d\setminus\Lambda$. By (\ref{egclimitunif}), we have for the first term in (\ref{4113}), using also \eqref{rhomax} in order to transfer the estimate to $(\hat \rho_\Lambda)^{l,\omega}_{\mathrm{max}}$
\begin{eqnarray*}\label{firstterm4113}
\lefteqn{\liminf_{N\to\infty}\frac{1}{N^{1+s/d}}\int_{\Omega_l}\bigg(\sum_{\substack{A\in F_\omega^l\\A\cap \Lambda\neq\emptyset}} E^\mathrm{xc}_{\mathrm{GC}, N|A\cap\Lambda|  \left(\max_{x\in A\cap\Lambda}\rho(x)\right),s}\left(\frac{1_{A\cap\Lambda}}{|A\cap\Lambda|}\right)\bigg)\,d\,\mathbb{P}_l(\omega)}\nonumber\\
&\ge&-C(s,d)\int_{\Omega_l}\sum_{\substack{A\in F_\omega^l\\A\cap \Lambda\neq\emptyset}}|A\cap\Lambda|  \left(\max_{y\in A\cap\Lambda}\rho(y)\right)^{1+s/d}d\mathbb P_l(\omega)\nonumber\\
&=&-C(s,d)\int_{\Omega_l}(\mu(\Lambda))^{1+s/d}({\hat m}^{l,\omega}_\Lambda)^{1+s/d}\int_{\mathbb{R}^d}\left( (\hat\rho_{\Lambda})^{l,\omega}_{\mathrm{max}}(x)\right)^{1+s/d} dx\,d\,\mathbb{P}_l(\omega).
\end{eqnarray*}
Therefore, dividing in (\ref{4113}) by $N^{1+s/d}$ and taking limits directly gives the bound
\begin{eqnarray}\label{4114}
\liminf_{N\to\infty}\frac{E_{\mathrm{GC}, N,s}^\mathrm{xc}(\mu)}{N^{1+s/d}}&\ge& -\frac{M}{M+C}\bigg\{C(s,d)\int_{\Omega_l}(\mu(\Lambda))^{1+s/d}({\hat m}^{l,\omega}_\Lambda)^{1+s/d}\int_{\mathbb{R}^d}\left( (\hat\rho_{\Lambda})^{l,\omega}_{\mathrm{max}}(x)\right)^{1+s/d} dx\,d\,\mathbb{P}_l(\omega)\nonumber\\
&&+c_{\mathrm{LO}}(d,\epsilon)\int_{\mathbb R^d\setminus\Lambda} \rho^{1+s/d}(x)dx+\frac{C(w,d, \epsilon)}{M}\int_{\mathbb{R}^d}\rho^{1+s/d}(x)dx\bigg\}.
\end{eqnarray}

To take the limits $l\to 0, M\to\infty$ in the above, for the first integral we make use of Step 1, in particular that ${\hat m}^{l,\omega}_\Lambda\to 1$ and $(\hat\rho_\Lambda)^{l,\omega}_{\mathrm{max}}\to \hat\rho_\Lambda$ in $L^{1+s/d}(\mathbb{R}^d)$ as $l\to 0$, and we use the Dominated Convergence Theorem (whose conditions are satisfied since $\rho|_\Lambda$ is bounded and $\Lambda$ is bounded).
In conclusion, from \eqref{4114} we get
\begin{eqnarray}\label{4115}
\liminf_{N\to\infty}\frac{E_{\mathrm{GC},N,s}^\mathrm{xc}(\mu)}{N^{1+s/d}}&\ge& -C(s,d)\int_{\Lambda} \rho^{1+s/d}(x) dx-c_{\mathrm{LO}}(d,\epsilon)\int_{\mathbb R^d\setminus\Lambda} \rho^{1+s/d}(x)dx.
\end{eqnarray}
\par Taking now $\Lambda=[-R,R]^d\cap\,\mathrm{supp}(\mu)$ and by taking $R\to \infty$ we have $\rho|_\Lambda\le\rho$ and $\rho|_\Lambda(x)\to\rho(x)$ for all $x\in\mathbb R^d$, therefore again by dominated convergence the first integral in \eqref{4115} converges to $\int_{\mathbb R^d}\rho^{1+s/d}(x)dx$, and similarly the second integral in \eqref{4115} tends to zero. This allows to conclude Theorem \ref{upboundcont} in this case.

\textbf{Step 3.} The extension to general marginals $\rho\in L^{1+s/d}(\mathbb{R}^d)$ follows now from Lusin's Theorem due to the measurability of $\rho$ and via (\ref{4115}) from Step 2. More precisely, if $\nu(A):=\int_A\rho^{1+s/d}(x)dx$ then for every $n\in\mathbb N$ there exists a compact set $\Lambda_n\subset\mathbb{R}^d$ such that $\rho|_{\Lambda_n}$ is continuous and $\nu(\mathbb R^d\setminus\Lambda_n)<2^{-n}$. Then we have from (\ref{4115}), with $\Lambda_n$ instead of $\Lambda$
\begin{equation*}
\liminf_{N\to\infty}\frac{E_{\mathrm{GC},N,s}^\mathrm{xc}(\mu)}{N^{1+s/d}}\ge -C(s,d)\int_{\Lambda_n}\rho^{1+s/d}(x)dx-c_{\mathrm{LO}}(d,\epsilon)\int_{\mathbb R^d\setminus \Lambda_n} \rho^{1+s/d}(x)dx.
\end{equation*}
Then $\nu(\mathbb R^d\setminus\Lambda_n)=\int_{\mathbb R^d\setminus \Lambda_n}\rho^{1+s/d}(x)dx\to 0$ as $n\to \infty$, which produces the desired result.
\end{proof}

\section{Small oscillations property of $E^{\mathrm{xc}}_{\mathrm{GC}, N,s}(\mu)/N^{1+s/d}$}
\label{secsmalloscprop}
%
%
%
%

\par \textbf{Proof of Theorem \ref{monincrsub}}

In order to prove \eqref{smalloscgc}, we will use again the Fefferman-Gregg decomposition. Let $\mu$ be a piecewise constant density of form $\sum_{j=1}^k\alpha_j\mu_j$, where for all $j=1,\ldots, k,$, $\mu_j\in\calP(\mathbb{R}^d)$ is a uniform measure on $\Lambda_j\subset\mathbb{R}^d$, with $\Lambda_j, j=1,\ldots,k,$ hyper-rectangles with disjoint nonempty interiors, and $\alpha_j\in \mathbb{R}_{>0}, j=1,\ldots k,$ $\sum_{j=1}^k\alpha_j=1$. 

\medskip

We prove below that 
\begin{equation}\label{bound_toprove}\frac{E_{\mathrm{GC}, N,s}^{\mathrm{xc}}(\mu)}{N^{1+s/d}}\ -\ \frac{E_{\mathrm{GC},\widetilde N,s}^{\mathrm{xc}}(\mu)}{\widetilde N^{1+s/d}}\ \ge\ -\frac{C(\Lambda_1,\ldots,\Lambda_k, \alpha_1,\ldots,\alpha_k, d,\epsilon)}{\log \widetilde N}\ ,
\end{equation}
as the other direction can be easily argued similarly. 

\par \textbf{Step 1.} At first, from Lemma \ref{subadd12} and eqn. (\ref{applysubadd12++}) and (\ref{roughbla}) from Step 1 in the proof of Proposition \ref{mixextlowb}, we get, if $B_1$ is the unit ball in $\R^d$ centered at zero,
\begin{eqnarray}\label{fefgreg}
E_{\mathrm{GC},N,s}^\mathrm{xc}(\mu)&\ge&\frac{M}{M+C}\bigg\{\sum_{j=1}^k \int_{\Omega_l}\bigg(\sum_{\substack{A\in F_\omega^l,\exists x\in\mathbb{R}^d\\\ A=B_R(x)\subset\Lambda_j}} R^{-s}E^\mathrm{xc}_{\mathrm{GC},N\alpha_j|B_R|/|\Lambda_j|,s}\left(\frac{1_{B_1}}{|B_1|}\right)\bigg)\,d\,\mathbb{P}_l(\omega)\nonumber\\
&&-\frac{C(w,\epsilon, d)}{M}R_1^{-s}(N-1)-\frac{C(w, d,\epsilon)}{M}N^{1+s/d}\int_{\mathbb{R}^d}\rho^{1+s/d}(x)dx\nonumber\\
&&-c_{\mathrm{LO}}(d,\epsilon) N^{1+s/d} \int_{\Omega_l}\bigg(\int_{\Lambda_\mathrm{err}(l,\omega)}\rho^{1+s/d}(x)dx\bigg)\,d\,\mathbb{P}_l(\omega)\bigg\},
\end{eqnarray}
where any $A\in F_\omega^l$ is a ball $B_R(x)$ for some choice $R\in\{tR_1,\ldots, tR_M\}$, $t\in [1-\eta,1+\eta]$ and some $x\in\mathbb R^d$. Furthermore, for each $l,\omega$ the set $\Lambda_\mathrm{err}(l,\omega)$ is the union of the balls from $F_\omega^l$ which intersect $\bigcup_{j=1}^k\partial\Lambda_j$, which in turn are contained in cubes of sidelength $tl$ that intersect $\partial\Lambda_j$ for some $1\le j\le k$, which cubes are all included in the set $\bigcup_{j=1}^k(\partial \Lambda_j)_{2l\sqrt{d}}$, where we recall the notation for $(\partial \Lambda_j)_{2l\sqrt{d}}$ from (\ref{480appr}). We assume in all our calculations below that $l$ is such that $|\Lambda_j|\gg 2l\sqrt{d}$ for all $i=1,\ldots,d$.

\medskip

To proceed, we now re-write for each $j=1,\ldots, k,$ in the first term in \eqref{fefgreg} 
\begin{eqnarray}\label{applcor}
\lefteqn{\int_{\Omega_l}\bigg(\sum_{\substack{B_R(x)\in F_\omega^l,\\ B_R(x)\subset\Lambda_j}} R^{-s}E^\mathrm{xc}_{\mathrm{GC},N\alpha_j|B_R|/|\Lambda_j|,s}\left(\frac{1_{B_1}}{|B_1|}\right)\bigg)\,d\,\mathbb{P}_l(\omega)}\nonumber\\
&=& \int_{\Omega_l}\bigg(\sum_{\substack{B_R(x)\in F_\omega^l,\\ B_R(x)\subset\Lambda_j}} R^{-s}(N\alpha_jR^d|B_1|/|\Lambda_j|)^{1+s/d}\frac{E^\mathrm{xc}_{\mathrm{GC},N\alpha_j|B_R|/|\Lambda_j|,s}\left(\frac{1_{B_1}}{|B_1|}\right)}{(N\alpha_jR^d|B_1|/|\Lambda_j|)^{1+s/d}}\bigg)\,d\,\mathbb{P}_l(\omega)\nonumber\\
&=& \bigg(\frac{\alpha_j}{|\Lambda_j|}\bigg)^{1+s/d}\sum_{i=1}^M\int_{\Omega_l}\bigg(\sum_{\substack{B_{tR_i}(x)\in F_\omega^l,\\ B_{tR_i}(x)\subset\Lambda_j}} (tR_i)^{-s}(N(tR_i)^d|B_1|)^{1+s/d}\frac{E^\mathrm{xc}_{\mathrm{GC},N\alpha_j|B_{tR_i}|/|\Lambda_j|,s}\left(\frac{1_{B_1}}{|B_1|}\right)}{(N\alpha_j(tR_i)^d|B_1|/|\Lambda_j|)^{1+s/d}}\bigg)\,d\,\mathbb{P}_l(\omega)\nonumber\\
&=& \bigg(\frac{\alpha_j}{|\Lambda_j|}\bigg)^{1+s/d}\sum_{i=1}^M\int_{\Omega_l}\frac{E^\mathrm{xc}_{\mathrm{GC},N\alpha_j|B_{tR_i}|/|\Lambda_j|,s}\left(\frac{1_{B_1}}{|B_1|}\right)}{(N\alpha_j(tR_i)^d|B_1|/|\Lambda_j|)^{1+s/d}}\bigg(\sum_{\substack{B_{tR_i}(x)\in F_\omega^l,\\ B_{tR_i}(x)\subset\Lambda_j}} (tR_i)^{-s}(N(tR_i)^d|B_1|)^{1+s/d}\bigg)\,d\,\mathbb{P}_l(\omega).\nonumber\\
\end{eqnarray}
We now use the fact that by Lemma \ref{cheeselemma} there holds for each $1\le i\le M$, $1\le j\le k$ and $\omega\in \Omega_l$, by a reasoning similar to the one leading to \eqref{480appr} 
\begin{equation}\label{54upbd}
\sum_{\substack{B_{tR_i}(x)\in F_\omega^l,\\ B_{tR_i}(x)\subset\Lambda_j}} (tR_i)^d|B_1|\le \frac{|\{y: \mathrm{dist}(y,\Lambda_j)\le 2l\sqrt{d}\}|}{M+C_d}=:\frac{|(\Lambda_j)_{2l\sqrt{d}}|}{M+C_d}.
\end{equation}
To obtain the above bound, note that the balls from $F^l_\omega$ which have radius $tR_i$ cover at most $(M+C_d)^{-1}$ of each cube $K\in\mathcal F_\omega^l$, where $\mathcal F_\omega^l$ is a covering of $\mathbb R^d$ by cubes of sidelength $tl$ with $t<2$, and disjoint interiors. The cubes of such covering that have nonempty intersection with $\Lambda_j$ stay within distance $2l\sqrt{d}$ of $\Lambda_j$ and thus their total volume is at most $|(\Lambda_j)_{2l\sqrt{d}}|$. These considerations directly lead to \eqref{54upbd}.

Applying \eqref{54upbd}, we obtain 
\begin{eqnarray}
\label{54upbdjerm}
\sum_{\substack{B_{tR_i}(x)\in F_\omega^l,\\ B_{tR_i}(x)\subset\Lambda_j}} (tR_i)^{-s}(N(tR_i)^d|B_1|)^{1+s/d}&\le& N^{1+s/d}|B_1|^{s/d}\frac{|(\Lambda_j)_{2l\sqrt{d}}|}{M+C_d}.
\end{eqnarray}
In view of (\ref{EGC_negative}), from \eqref{applcor}, \eqref{54upbdjerm} we get for $1\le j\le k$
\begin{eqnarray}\label{55upbd}
\lefteqn{\int_{\Omega_l}\bigg(\sum_{\substack{B_R(x)\in F_\omega^l,\\ B_R(x)\subset\Lambda_j}} R^{-s}E^\mathrm{xc}_{\mathrm{GC},N\alpha_j|B_R|/|\Lambda_j|,s}\left(\frac{1_{B_1}}{|B_1|}\right)\bigg)\,d\,\mathbb{P}_l(\omega)}\nonumber\\
&\ge& N^{1+s/d}|B_1|^{s/d} \frac{|(\Lambda_j)_{2l\sqrt{d}}|}{M+C_d}\int_{\Omega_l}\sum_{i=1}^M\frac{E^\mathrm{xc}_{\mathrm{GC},N\alpha_j|B_{tR_i}|/|\Lambda_j|,s}\left(\frac{1_{B_1}}{|B_1|}\right)}{(N\alpha_j(tR_i)^d|B_1|/|\Lambda_j|)^{1+s/d}}d\,\mathbb{P}_l(\omega).
\end{eqnarray}

\medskip

We now move to finding an upper bound for the last term in \eqref{fefgreg}. Similarly to the reasoning used to justify \eqref{54upbd}, $\Lambda_\mathrm{err}(l,\omega)\cap \Lambda$ is contained in $\bigcup_{j=1}^k(\partial\Lambda_j)_{2l\sqrt{d}}$, and thus we have for the last term in \eqref{fefgreg} 
\begin{eqnarray}\label{56err}
\int_{\Lambda_\mathrm{err}(l,\omega)}\rho^{1+s/d}(x)dx&\le& \int_{\bigcup_{j=1}^k(\partial\Lambda_j)_{2l\sqrt{d}}}\rho^{1+s/d}(x)dx \nonumber\\
&\le&\int_{\bigcup_{j=1}^k(\partial\Lambda_j)_{2l\sqrt{d}}}\rho^{1+\epsilon/d}(x)dx+\int_{\bigcup_{j=1}^k(\partial\Lambda_j)_{2l\sqrt{d}}}\rho^{1+(d-\epsilon)/d}(x)dx.
\end{eqnarray}
Thus, by using the estimates from \eqref{54upbd}, we obtain in (\ref{fefgreg}) by means of (\ref{55upbd}) and of \eqref{56err} for some $c'(d,\epsilon)>0$
\begin{eqnarray}\label{57step1}
\lefteqn{\frac{E_{\mathrm{GC},N,s}^\mathrm{xc}(\mu)}{N^{1+s/d}}\ge |B_1|^{s/d}\frac{M}{M+C}\sum_{j=1}^k\bigg(\frac{\alpha_j}{|\Lambda_j|}\bigg)^{1+s/d}\frac{|(\Lambda_j)_{2l\sqrt{d}}|}{M+C_d}\sum_{i=1}^M\int_{\Omega_l}\frac{E^\mathrm{xc}_{\mathrm{GC},N\alpha_j|B_{tR_i}|/|\Lambda_j|,s}\left(\frac{1_{B_1}}{|B_1|}\right)}{(N\alpha_j(tR_i)^d|B_1|/|\Lambda_j|)^{1+s/d}}d\,\mathbb{P}_l(\omega)}\nonumber\\
&&-\frac{c'(d,\epsilon)}{M}\frac{N-1}{N^{1+\frac{s}{d}}}R_1^{-s}-c'(d,\epsilon)\sum_{s'\in\{\epsilon,d-\epsilon\}}\bigg(\int_{{\bigcup_{j=1}^k(\partial\Lambda_j)_{2l\sqrt{d}}}}\rho^{1+s'/d}(x)dx\,+\frac{1}{M}\int_{\mathbb{R}^d}\rho^{1+s'/d}(x)dx\bigg).
\end{eqnarray}
\par \textbf{Step 2.} We will next get a lower bound for $E_{\mathrm{GC},\tilde N,s}^\mathrm{xc}(\mu)$ in function of $E_{\mathrm{GC}, N|B_{t{\tilde R}_i}|,s}^\mathrm{xc}(\mu)$, where $\widetilde N>0$ and $({\tilde R}_i)$ is a new balls cover.

To begin with, using as in the proof of Lemma \ref{subadd12}, for another cover  and for any ${\widetilde\Lambda}_\mathrm{err}(l,\omega)$ as in (\ref{sigmaerr2}), we can write

\begin{equation}
	\label{help1}
\mu=\mu(\Lambda\setminus {\widetilde\Lambda}_\mathrm{err}(l,\omega))\hat\mu_{\Lambda\setminus {\widetilde\Lambda}_\mathrm{err}(l,\omega)}+\mu({\widetilde\Lambda}_\mathrm{err}(l,\omega))\hat\mu_{ {\widetilde\Lambda}_\mathrm{err}(l,\omega)}
\end{equation}
and 
\begin{equation}
\label{help2}
\hat\mu_{\Lambda\setminus {\widetilde\Lambda}_\mathrm{err}(l,\omega)}=\sum_{\substack{{\widetilde {A}}\in {\widetilde F}^l_\omega\\ \widetilde{A}\subset\Lambda,\mu(\widetilde A)>0}}\hat\mu_{\Lambda\setminus  {\widetilde\Lambda}_\mathrm{err}(l,\omega)}(\widetilde A)\,\hat\mu_{\widetilde A}+\sum_{\substack{\widetilde A\in {\widetilde F}^l_\omega\\ \tilde A\subset\Lambda,\mu(\widetilde A)=0}}\hat\mu_{\Lambda\setminus  {\widetilde\Lambda}_\mathrm{err}(l,\omega)}|_{\widetilde A}\ .
\end{equation}
 We assume below that $\widetilde N\alpha_j|B_{\widetilde R}|/|\Lambda_j|>1, j=1,\ldots,k$. Applying Remark \ref{rmk_gcb} (4) twice, firstly to (\ref{help1}) and then to (\ref{help2}), and then using (\ref{EGC_negative}), we get
\begin{eqnarray}\label{eachtproc}
E_{\mathrm{GC},\widetilde N,s}^\mathrm{xc}(\mu)&\le&E_{\mathrm{GC},\widetilde N\mu(\Lambda\setminus {\widetilde\Lambda}_\mathrm{err}(l,\omega)),s}^\mathrm{xc}\left(\hat\mu_{\Lambda\setminus {\widetilde\Lambda}_\mathrm{err}(l,\omega)}\right)\nonumber\\
&\le&\sum_{j=1}^k\sum_{\substack{{\widetilde {A}}\in {\widetilde F}^l_\omega\\ \widetilde{A}\subset\Lambda_j,\mu(\widetilde A)>0}}E_{\mathrm{GC},\widetilde N\mu(\widetilde A),s}^\mathrm{xc}({\hat\mu}_{\tilde A})=\sum_{j=1}^k\sum_{\substack{{\widetilde {A}}\in {\widetilde F}^l_\omega\\ \widetilde{A}\subset\Lambda_j,\mu(\widetilde A)>0}}E^\mathrm{xc}_{\mathrm{GC},\widetilde N\alpha_j|A|/|\Sigma_j|,s}\left(\frac{1_{\tilde A}}{|\tilde A|}\right)\nonumber\\
&=&\sum_{j=1}^k\sum_{\tilde A=\substack{B_{\widetilde R}(x)\in \widetilde F_\omega^l,\\ B_{\widetilde R}(x)\subset\Lambda_j}} E^\mathrm{xc}_{\mathrm{GC},\widetilde N\alpha_j|B_{\widetilde R}|/|\Lambda_j|,s}\left(\frac{1_{B_{\widetilde R}}}{|B_{\widetilde R}|}\right)=\sum_{j=1}^k\sum_{\substack{B_{\widetilde R}(x)\in \widetilde F_\omega^l,\\ B_{\widetilde R}(x)\subset\Lambda_j}}\widetilde R^{-s}E^\mathrm{xc}_{\mathrm{GC},\widetilde N\alpha_j|B_{\widetilde R}|/|\Lambda_j|,s}\left(\frac{1_{B_1}}{|B_1|}\right)\ .\nonumber\\
\end{eqnarray}
Performing this procedure for every $\omega$, and then integrating, we get
\begin{eqnarray}\label{511ubd}
E_{\mathrm{GC},\widetilde N,s}^\mathrm{xc}(\mu)&\le&\sum_{j=1}^k\int_{\Omega_l}\sum_{\substack{B_{\widetilde R}(x)\in \widetilde F_\omega^l,\\ B_{\widetilde R}(x)\subset\Lambda_j}} {\widetilde R}^{-s}E^\mathrm{xc}_{\mathrm{GC},\widetilde N\alpha_j|B_{\widetilde R}|/|\Lambda_j|,s}\left(\frac{1_{B_1}}{|B_1|}\right)\,d\,\mathbb{P}_l(\omega).
\end{eqnarray}

By Lemma \ref{cheeselemma} and similarly to \eqref{54upbd}, for each $1\le i\le \tilde M$, each $1\le j\le k$ and each $\omega\in \Omega_l$ we have, by a reasoning similar to the one used to justify \eqref{54upbd},

\begin{equation}\label{512lobd}
\sum_{\substack{{\tilde B}_{t{\tilde R}_i}(x)\in {\tilde F}_\omega^l,\\ B_{t{\tilde R}_i}(x)\subset\Lambda_j}} (t{\tilde R}_i)^d|B_1|\ge\frac{|\{y:\ \mathrm{dist}(y,\mathbb R^d\setminus\Lambda_j)>2l\sqrt{d}\}|}{\widetilde{M}+C_d+1}:=\frac{|(\Lambda_j)_{2l\sqrt{d}}^*|}{\widetilde{M}+C_d+1}.
\end{equation}
To justify \eqref{512lobd}, we use the same notation as in the paragraph following \eqref{54upbd}. As ${\tilde F}_\omega^l$ forms a covering of $\mathbb R^d$, any point $y\in(\Lambda_j)^*_{2l\sqrt{d}}$ is in some cube $K_y\in {\tilde F}_\omega^l$. As the sidelength of $K_y$ is smaller than $2l$ and $y\in(\Lambda_j)^*_{2l\sqrt{d}}$, we have $K_y\subset B(y,2l\sqrt{d})\subset\Lambda_j$; this proves that $(\Lambda_j)^*_{2l\sqrt{d}}$ is covered by cubes from ${\tilde F}^l_\omega$ that are completely contained in $\Lambda_j$. The balls appearing in the sum on the left in \eqref{512lobd} cover at least $(M +C_d+1)^{-1}$ of $K$ for each $K\in {\tilde F}_\omega^l$ completely contained in $\Lambda_j$, and as the $K
\in {\tilde F}_\omega^l$ completely contained in $\Lambda_j$ cover $(\Lambda_j)^*_{2l\sqrt{d}}$, \eqref{512lobd} follows.

Now \eqref{511ubd} and \eqref{512lobd} give
\begin{eqnarray}\label{513step2}
\frac{E_{\mathrm{GC},\widetilde N,s}^\mathrm{xc}(\mu)}{{\widetilde N}^{1+s/d}}&\le&\sum_{j=1}^k\int_{\Omega_l}\sum_{\substack{B_{\widetilde R}(x)\in \widetilde F_\omega^l,\\ B_{\widetilde R}(x)\subset\Lambda_j}} {\widetilde R}^{-s}\frac{E^\mathrm{xc}_{\mathrm{GC},\widetilde N\alpha_j|B_{\widetilde R}|/|\Lambda_j|,s}\left(\frac{1_{B_1}}{|B_1|}\right)}{{\widetilde N}^{1+s/d}}\,d\,\mathbb{P}_l(\omega)\nonumber\\
&=&\sum_{j=1}^k\sum_{i=1}^{\widetilde M}\int_{\Omega_l}\sum_{\substack{B_{t\widetilde R_i}(x)\in \widetilde F_\omega^l,\\ B_{t\widetilde R_i}(x)\subset\Lambda_j}} (t{\widetilde R}_i)^{-s}((t\alpha_j{\widetilde R}_i)^d|B_1|/|\Lambda_j|)^{1+s/d}\frac{E^\mathrm{xc}_{\mathrm{GC},\widetilde N\alpha_j|B_{t{\widetilde R}_i}|/|\Lambda_j|,s}\left(\frac{1_{B_1}}{|B_1|}\right)}{(\widetilde N \alpha_j(t{\widetilde R_i})^d|B_1|/|\Lambda_j|)^{1+s/d}}\,d\,\mathbb{P}_l(\omega)\nonumber\\
&=&\sum_{j=1}^k\left(\frac{\alpha_j}{|\Lambda_j|}\right)^{1+s/d}\sum_{i=1}^{\widetilde M}\int_{\Omega_l}\frac{E^\mathrm{xc}_{\mathrm{GC},\widetilde N\alpha_j|B_{t{\widetilde R}_i}|/|\Lambda_j|,s}\left(\frac{1_{B_1}}{|B_1|}\right)}{(\widetilde N \alpha_j(t{\widetilde R_i})^d|B_1|/|\Lambda_j|)^{1+s/d}}\sum_{\substack{B_{t\widetilde R_i}(x)\in \widetilde F_\omega^l,\\ B_{t\widetilde R_i}(x)\subset\Lambda_j}} (t{\widetilde R}_i)^d|B_1|^{1+s/d}\,d\,\mathbb{P}_l(\omega)\nonumber\\
&\le&{|B_1|^{s/d}}\sum_{j=1}^k\left(\frac{\alpha_j}{|\Lambda_j|}\right)^{1+s/d}\frac{|(\Lambda_j)_{2l\sqrt{d}}^*|}{\widetilde{M}+C_d+1}\sum_{i=1}^{\widetilde M}\int_{\Omega_l}\frac{E^\mathrm{xc}_{\mathrm{GC},\widetilde N\alpha_j|B_{t{\widetilde R}_i}|/|\Lambda_j|,s}\left(\frac{1_{B_1}}{|B_1|}\right)}{(\widetilde N \alpha_j(t{\widetilde R_i})^d|B_1|/|\Lambda_j|)^{1+s/d}}\,d\,\mathbb{P}_l(\omega).
\end{eqnarray}
\par \textbf{Step 3.} We start by subtracting \eqref{513step2} from \eqref{57step1}. Therefore we have
\begin{eqnarray}\label{5144}
 \lefteqn{\frac{E_{\mathrm{GC},N,s}^\mathrm{xc}(\mu)}{N^{1+s/d}}-\frac{E_{\mathrm{GC},\widetilde N,s}^\mathrm{xc}(\mu)}{{\widetilde N}^{1+s/d}}}\nonumber\\
 &\ge&-\frac{c'(d,\epsilon)}{M}\frac{N-1}{N^{1+\frac{s}{d}}}R_1^{-s}-c'(d,\epsilon)\sum_{s'\in\{\epsilon,d-\epsilon\}}\bigg(\int_{\bigcup_{j=1}^k(\partial\Lambda_j)_{2l\sqrt{d}}}\rho^{1+s'/d}(x)dx+\frac{1}{M}\int_{\mathbb{R}^d}\rho^{1+s'/d}(x)dx\bigg)\nonumber\\
 &&+|B_1|^{s/d}\sum_{j=1}^k\bigg(\frac{\alpha_j}{|\Lambda_j|}\bigg)^{1+s/d}\bigg[|(\Lambda_j)_{2l\sqrt{d}}|\frac{M}{(M+C)(M+C_d)}\sum_{i=1}^M\int_{\Omega_l}\frac{E^\mathrm{xc}_{\mathrm{GC},N\alpha_j|B_{tR_i}|/|\Lambda_j|,s}\left(\frac{1_{B_1}}{|B_1|}\right)}{(N\alpha_j(tR_i)^d|B_1|/|\Lambda_j|)^{1+s/d}}d\,\mathbb{P}_l(\omega)\nonumber\\
 &&\phantom{|B_1|^{s/d}\sum_{j=1}^k\bigg(\frac{\alpha_j}{|\Lambda_j|}\bigg)^{1+s/d}}-\frac{|(\Lambda_j)_{2l\sqrt{d}}^*|}{\widetilde{M}+C_d+1}\sum_{i=1}^{\widetilde M}\int_{\Omega_l}\frac{E^\mathrm{xc}_{\mathrm{GC},\widetilde N\alpha_j|B_{t{\widetilde R}_i}|/|\Lambda_j|,s}\left(\frac{1_{B_1}}{|B_1|}\right)}{(\widetilde N \alpha_j(t{\widetilde R_i})^d|B_1|/|\Lambda_j|)^{1+s/d}}\,d\,\mathbb{P}_l(\omega)\bigg].\nonumber\\
 &:=&-\frac{c'(d,\epsilon)}{M}\frac{N-1}{N^{1+\frac{s}{d}}}R_1^{-s}-c'(d,\epsilon)\sum_{s'\in\{\epsilon,d-\epsilon\}}I(s',\rho, l,M) + II(s,\rho,l,M,\widetilde M,N, \widetilde N)
\end{eqnarray}
In order to treat the more complicated term $II(s,\rho,l,M,\widetilde M,N, \widetilde N)$ in \eqref{5144}, in the next paragraphs we fix $l,M,R_1,N, \widetilde M, {\widetilde R}_1, \widetilde N,$ satisfying \eqref{boundlmr} and \eqref{boundlmr2}, (where we define then $R_i=(c(d))^{i-1}R_1$ and $\widetilde R_j=(c(d))^{j-1}\widetilde R_1$ for $2\le i\le M$ and $2\le j\le \widetilde M$, where $c(d)$ is the constant coming from the Lemma \ref{cheeselemma}, depending only on $d$). 

\medskip

Firstly, we may take $M=\widetilde M$, and choose $N,R_1,\widetilde N,\widetilde R_1$ such that $N\ge\widetilde N$, and
\begin{equation}\label{514NRchoice}
\widetilde N  \widetilde R_i^d=N R_i^d, \quad \mbox{for all}\quad i=1,\ldots,M.
\end{equation}
Furthermore, we also need $\widetilde N\widetilde R_i^d\ll \widetilde N, i=1,\ldots, M$. This can be ensured by imposing the precise relation $\widetilde R_1=R_1(N/\widetilde N)^{\frac1d}\ge R_1$ and $NR_M^d\ll \widetilde N$. The required conditions from Lemma \ref{cheeselemma} for $N, M,R_1,l$ and for $\widetilde N, M, \widetilde R_1,l$ can be written as
\begin{equation*}
l\ll 1\ll M, \quad (M+C)C^M\max\{R_1,\widetilde R_1\}=(M+C)C^M\widetilde R_1<l,
\end{equation*}
where we recall that $C$ is a dimensional constant, independent of the choice of the function $\mathsf{c}$.

As we assumed $N\ge \widetilde N$ and $C>1$, we find that, due to \eqref{514NRchoice}, the above condition is thus equivalent to the Lemma \ref{cheeselemma} condition on the parameters $\widetilde N, \widetilde R_1, l, M$ and is therefore achievable.  The assumptions in \eqref{boundlmr3} combined with \eqref{514NRchoice} give
\begin{equation}\label{515lmr}
{\widetilde R}_1:={\widetilde N}^{-\frac1{2d}},\quad R_1={\widetilde N}^{\frac{1}{2d}} N^{-\frac{1}{d}},\quad l:={\widetilde N}^{-\frac1{3d}},\quad M=\widetilde M:=\left[\frac{\log\widetilde N}{18d\log C}\right]-1\ .
\end{equation}
The requirements \eqref{515lmr} together with the requirement $M\ge M_d$ as in the discussion preceding \eqref{boundlmr3} turn out to be sufficient for our purposes. 

The choices \eqref{514NRchoice} and \eqref{515lmr} imply that for all $i=1,\ldots,M,$ and for all $j=1,\ldots,k,$ we have
\begin{equation}
\label{516}
\int_{\Omega_l}\frac{E^\mathrm{xc}_{\mathrm{GC},N\alpha_j|B_{t R_i}|/|\Lambda_j|,s}\left(\frac{1_{B_1}}{|B_1|}\right)}{(N\alpha_j(t R_i)^d|B_1|/|\Lambda_j|)^{1+s/d}}\,d\,\mathbb{P}_l(\omega)=\int_{\Omega_l}\frac{E^\mathrm{xc}_{\mathrm{GC},\widetilde N\alpha_j|B_{t\widetilde R_i}|/|\Lambda_j|,s}\left(\frac{1_{B_1}}{|B_1|}\right)}{(\widetilde N\alpha_j(t \widetilde R_i)^d|B_1|/|\Lambda_j|)^{1+s/d}}\,d\,\mathbb{P}_l(\omega).
\end{equation}
Also note that the $2l\sqrt{d}$-outer layer of $\Lambda_j$ differs from its $2l\sqrt{d}$-inner layer precisely by the $2l\sqrt{d}$-neighborhood of $\partial\Lambda_j$, thus with the notations introduced in \eqref{54upbd} and \eqref{512lobd} we have
\begin{equation}\label{5188}
 \left|(\Lambda_j)_{2l\sqrt{d}}\right| = \left|(\Lambda_j)_{2l\sqrt{d}}^*\right|+  \left|(\partial\Lambda_j)_{2l\sqrt{d}}\right|,
\end{equation}
where 
$$(\partial\Lambda_j)_{2l\sqrt{d}}:=\{y: \mathrm{dist}(y,\partial\Lambda_j)\le 2l\sqrt{d}\}.$$
Using \eqref{5188} and using the fact that we chose $M=\widetilde M$ in \eqref{515lmr}, we can write, for $C>1$ and $\widetilde M\ge M_d$ large enough so that the leading-order terms in $\widetilde M$ below dominate the other terms
\begin{eqnarray}\label{5199}
 \lefteqn{|(\Lambda_j)_{2l\sqrt{d}}| \frac{M}{(M+C)(M+C_d)}- \frac{|(\Lambda_j)_{2l\sqrt{d}}^*|}{\widetilde{M}+C_d+1}}\nonumber\\
 &=&|(\Lambda_j)_{2l\sqrt{d}}^*|\left(\frac{\widetilde M}{(\widetilde M+C)(\widetilde M+C_d)} - \frac{1}{\widetilde M + C_d+1}\right)+|(\partial \Lambda_j)_{2l\sqrt{d}}|\frac{\widetilde M}{(\widetilde M+C)(\widetilde M+C_d)}\nonumber\\
 &\le&-(C-1)\frac{|(\Lambda_j)_{2l\sqrt{d}}^*|}{\widetilde M^2} +\frac{|(\partial \Lambda_j)_{2l\sqrt{d}}|}{\widetilde M} \le \frac{|(\partial \Lambda_j)_{2l\sqrt{d}}|}{\widetilde M}.
\end{eqnarray}
Inserting \eqref{5199} in the $II$-term from \eqref{5144} and keeping in mind the non-positivity \eqref{EGC_negative}, the bound \eqref{loterm1gs} and the fact that $\mathbb P_l$ is a probability measure, we get
\begin{eqnarray}\label{5200}
 \lefteqn{II(s,\rho,l,M,\widetilde M,N, \widetilde N)}\nonumber\\
 &\ge&\frac{\max_{1\le j\le k}|(\partial \Lambda_j)_{2l\sqrt{d}}|}{\widetilde M}\sum_{j=1}^k|B_1|^{s/d}\bigg(\frac{\alpha_j}{|\Lambda_j|}\bigg)^{1+s/d}\sum_{i=1}^{\widetilde M}\int_{\Omega_l}\frac{E^\mathrm{xc}_{\mathrm{GC},\widetilde N\alpha_j|B_{t\widetilde R_i}|/|\Lambda_j|,s}\left(\frac{1_{B_1}}{|B_1|}\right)}{(\widetilde N\alpha_j(t\widetilde R_i)^d|B_1|/|\Lambda_j|)^{1+s/d}}d\,\mathbb{P}_l(\omega)\nonumber\\
 &\ge&-c_\mathrm{LO}(d,\epsilon)\left(\max_{1\le j\le k}|(\partial \Lambda_j)_{2l\sqrt{d}}|\right)\sum_{j=1}^k|B_1|^{s/d}\bigg(\frac{\alpha_j}{|\Lambda_j|}\bigg)^{1+s/d}\ge-l C_1(\Lambda_1,\ldots,\Lambda_k, \alpha_1,\ldots,\alpha_k,d,\epsilon),\quad\quad
\end{eqnarray}
where the crucial ingredient used is the fact that $\Lambda_j$ are $\phi$-regular and in particular $|(\partial \Lambda_j)_r|\le C_{\Lambda_j}r$ for $r\in(0,1)$, and the constant $C_1$ above depends on these bounds $C_{\Lambda_j}$ as well.

\medskip

In order to bound the first term of the last line in \eqref{5144}, we directly use \eqref{515lmr} and write
\begin{equation}\label{5211}
R_1^{-s}N^{-\frac{s}{d}}=\left(\frac{\widetilde N}{N}\right)^{-\frac{s}{d}}{\widetilde N}_1^{\frac{s}{2d}}N^{-\frac{s}{d}}= {\widetilde N}^{-\frac{s}{2d}}.
\end{equation}
In order to estimate the $I$-term in \eqref{5144} we treat separately the two integrals appearing in it. For the integration on the neighborhood of the boundaries, we use again the bound $|(\partial \Lambda_j)_r|\le C_{\Lambda_j}r$ and the explicit form of $\rho$. For the integral of $\rho^{1+s'/d}$ with $s'\in\{\epsilon, d-\epsilon\}$, again we can use the explicit form of $\rho$, this time in an even more direct way. Explicitly, we obtain
\begin{eqnarray}\label{5222}
\lefteqn{c'(d,\epsilon)\sum_{s'\in\{\epsilon, d-\epsilon\}}I(s',\rho,l,M)}\nonumber\\
&\le& c'(d,\epsilon)\left(2l\sqrt{d}\cdot \left(\max_{1\le j\le k}C_{\Lambda_j}\right)\sum_{s'\in\{\epsilon, d-\epsilon\}}\max_{1\le j\le k}\frac{\alpha_i^{1+s'/d}}{|\Lambda_j|^{1+s'/d}} + \frac{1}{M}\sum_{s'\in\{\epsilon, d-\epsilon\}}\sum_{j=1}^k\frac{\alpha_i^{1+s'/d}}{|\Lambda_j|^{s'/d}}\right)\nonumber\\
&:=&\left(l+\frac{1}{M}\right) C_2(\Lambda_1,\ldots,\Lambda_k, \alpha_1,\ldots,\alpha_k,d,\epsilon).
\end{eqnarray}
Summing now the estimates \eqref{5200}, \eqref{5211} and \eqref{5222} and taking $C'(\Lambda_1,\ldots,\Lambda_k, \alpha_1,\ldots,\alpha_k,d,\epsilon)$ to be the sum of the analogous constants from those equations, we obtain from \eqref{5144} that (recalling again that we chose $\widetilde M=M$)
\begin{equation}\label{5233}
 \frac{E_{\mathrm{GC},N,s}^\mathrm{xc}(\mu)}{N^{1+s/d}}-\frac{E_{\mathrm{GC},\widetilde N,s}^\mathrm{xc}(\mu)}{{\widetilde N}^{1+s/d}}\ge \frac{C'(w,\epsilon,d)}{\widetilde M}\widetilde N^{-\frac{s}{2d}} - \left(l+\frac{1}{\widetilde M}\right)  C'(\Lambda_1,\ldots,\Lambda_k, \alpha_1,\ldots,\alpha_k,d,\epsilon).
\end{equation}
By the choices \eqref{515lmr} we now see that $\widetilde M^{-1}>l$ and thus its contribution dominates the bound in \eqref{5233}, allowing to conclude.
\qed

An immediate consequence of Theorem \ref{monincrsub}, and of the Cauchy criterion of uniform convergence is: 
\begin{corollary}[Uniform convergence of $E_{\mathrm{GC}, N,s}^{\mathrm{xc}}(\mu)$ with respect to $s$]\label{unifconv}
Fix $0<\epsilon<d/2$ and let $\epsilon\le s\le d-\epsilon$. Let $\mu\in \mathcal {P}(\mathbb{R}^d)$ be a probability measure with density of the form $\rho(x)=\sum_{i=1}^k\alpha_i1_{\Lambda_i}(x)$ where $\Lambda_1,\ldots,\Lambda_k$ are Borel sets with $\phi$-regular boundary and disjoint interiors. Then the sequence of functions
$$
f_s(N):=\frac{E_{\mathrm{GC}, N,s}^\mathrm{xc}(\mu)}{N^{1+s/d}}
$$
converges as $N\to\infty$ uniformly with respect to the parameter $s\in[\epsilon, d-\epsilon]$.
\end{corollary}
\begin{remark}
As detailed also in the introduction, both Theorem \ref{monincrsub} and Corollary \ref{unifconv} can be shown for more general densities, such as any positive Riemann-integrable density, including in particular continuous densities with compact support and $\phi$-regular boundary. The proof follows by a similar argument as the proof of Theorem \ref{monincrsub}, combined with an adaptation of the small oscillations proof for piece-wise constant densities and using also the small oscillations inequality proved above for the uniform density on $1_{B_1}/|B_1|$.
\end{remark}

\appendix
\section{Fefferman decomposition and positive definiteness}\label{appendixfefferman}
\subsection{Derivative bounds and positive definite control of the error}\label{ssecfeffermandecomposition}
We present here the main ingredients for the generalization of Fefferman's and Gregg's approach to effectively localizing our interaction kernels. We restrict to the case $c(x-y)=|x-y|^{-s}, 0<s<d$, though the methods below can be applied to more general costs of the form $c(x-y)$, under suitable regularity assumptions.
\subsubsection{Fefferman-Gregg positive definiteness criterion}\label{sssecgregg}
In this subsection we present a generalization to arbitrary $0<s<d$ of the main lemma and of the kernel decomposition from \cite{Gregg89}, given there for $d=3$ and $0<s<3$.

\par We recall the multi-index notation. We will denote by $\beta=(\beta_1,\ldots,\beta_d)\in \mathbb N^d$ a multiindex and its length will be defined as $|\beta|:=\beta_1+\cdots+\beta_d$. For a function $f:\mathbb R^d\to \mathbb R$, for $x=(x_1, \ldots, x_d)\in\mathbb R^d$ and $\xi=(\xi_1, \ldots, \xi_d)\in\mathbb R^d$ we then use the partial differentiation and monomial notations
\begin{equation}\label{derivnotation}
\partial_x^\beta f(x) := \partial_{x_1}^{\beta_1}\partial_{x_1}^{\beta_2}\cdots\partial_{x_d}^{\beta_d}f(x), \quad x^\beta:= x_1^{\beta_1}x_2^{\beta_1}\cdots x_d^{\beta_d},\quad \xi^\beta:= \xi_1^{\beta_1}\xi_2^{\beta_1}\cdots \xi_d^{\beta_d}\ .
\end{equation}
In particular $\partial_x^0f(x)=f(x)$.
\begin{lemma}\label{fourierbdlemma}
Let $0<s<d$ and consider a kernel $g:\mathbb R^d\to \mathbb R$ which is $d$ times differentiable on $\mathbb R^d\setminus\{0\}$.  Assume that there exists $\widetilde{C}(d,\epsilon)>0$, depending only on $d,\epsilon,$ if $0<\epsilon<d/2$ and $\epsilon\le s\le d-\epsilon$, such that for all $x\neq 0$, for $|\beta|=0$ and for all $\beta$ with $|\beta|=d$, there holds
\begin{equation}\label{boundgregg}
|\partial_x^\beta g(x)|\le \widetilde{C} |x|^{-s-|\beta|}\ .
\end{equation}
Then there exists a constant $c(d,\epsilon)>0$, also depending only on $d,\epsilon$ for our choices of $s,d,\epsilon$, and such that for all $\xi\neq 0$ there holds
\begin{equation}\label{fourierboundgregg}
|\hat g(\xi)|\le c(d,\epsilon)|\xi|^{s-d}\ .
\end{equation}
\end{lemma}
The proof of the lemma follows the same steps as in the beginning of the proof of the main lemma in \cite[p. 257-258, down to the line after eq. (3)]{Gregg89}, with the extra verification that the bounds in that paper are explicit and depend on $s$ in a polynomial way, without exploding in the interval $s\in(\epsilon,d-\epsilon)$. The details will be omitted. As an immediate consequence of the above lemma, we have
\begin{corollary}[positive definiteness criterion]\label{pdcriterion}
Under the hypotheses of Lemma \ref{fourierbdlemma}, there exists a constant $C>0$ depending only on $d,\epsilon$, such that $|x|^{-s} \pm C g(x)$ is positive definite. In particular, if $c(d,\epsilon)$ is the constant appearing in Lemma \ref{fourierbdlemma} and the Fourier transform of $|x|^{-s}$ is $c_{s,d}|\xi|^{d-s}$ then we can choose $C=\frac{\inf_{s\in(\epsilon,d-\epsilon)}c_{s,d}}{c(d,\epsilon)}$.
\end{corollary}
We will now prove that estimates of the type \eqref{boundgregg} hold for two functions built from $Q_{i,\eta}$ and which will be relevant in our constructions. This is reminiscent of, and can be seen as a generalization of the bounds \cite[p. 272, for the functions called $k_\dag^i$]{Gregg89}, which are there stated in $d=3$ and not proved explicitly. Due to this, and in order to provide an explicit proof in general dimension, we include here the full proof. In preparation for the more difficult result of Lemma \ref{derivativebounds} below, which really uses the geometric growth of the $R_i$ from Lemma \ref{cheeselemma}, we prove an easier bound:
\begin{lemma}\label{derivativeboundseasy}
Let $0<\epsilon<d/2$ and $\epsilon\le s\le d-\epsilon$, and let $Q_{i,\eta}$ be defined as in \eqref{qieta}, for $R_i$ defined as in Lemma \ref{cheeselemma}. Then there exists $C(\rho_\eta, d, \epsilon)>0$ depending only on the choice of $\rho_\eta$ and of $d,\epsilon$, such that the bounds \eqref{boundgregg} hold for $|\beta|\le d+1$ for
\begin{equation}\label{badbound}
\frac{1}{C(\rho_\eta,d,\epsilon)}\frac{Q_{i,\eta}(x)}{|x|^s}\ .
\end{equation}
\end{lemma}
\begin{proof}
We first note that 
\begin{equation}\label{reduceq00}
\partial_x^\beta\left(\frac{Q_{i,\eta}(x)}{|x|^s}\right) = \frac{1}{R_i^{s+|\beta|}}\left.\partial_x^\beta\left(\frac{Q_{0,\eta}(x')}{|x'|^s}\right)\right|_{x'=\frac{x}{R_i}}\ .
\end{equation}
therefore, looking at \eqref{boundgregg}, we see that it suffices to prove the bound for $R_i=1$ (i.e. for the case where we replace $Q_{i,\eta}$ by $Q_{0,\eta}$), and the one for general $R_i$ will directly follow from it via \eqref{reduceq00}. We note that due to the mollification \eqref{mollifballs}, partial derivatives of $Q_{0,\eta}$ up to order $|\beta|=d$ satisfy supremum bounds depending only on $\rho_\eta, d$, and the partial derivatives up to order $|\beta|=d$ of $C^{-1}|x|^{-s}$ satisfy \eqref{boundgregg} whenever $C>c(d,\epsilon)$ for some constant depending only on $d,\epsilon$. The claim follows by triangular inequality, if we distribute the $|\beta|\le d$ partial derivatives on the two terms of the product $|x|^{-s}Q_{0,\eta}(x)$.
\end{proof}
The above estimate \eqref{badbound} will not be directly helpful for us, because as we sum it over the $M$ terms corresponding to $i=1,\ldots,M,$ we will obtain a bound dependent on $M$. The following subtler result uses crucially the geometric growth of the $R_i$ from Lemma \ref{cheeselemma}, and requires that for each $i$ we subtract a tamer kernel which effectively cancels the tail behaviour of $|x|^{-s}(1- Q_{i,\eta}(x))$ at infinity.
\begin{lemma}\label{derivativebounds}
Let $0<\epsilon<d/2$ and $\epsilon\le s\le d-\epsilon$, and let $Q_{i,\eta}$ be defined as in \eqref{qieta}, for $R_i$ defined as in Lemma \ref{cheeselemma}. Then there exists $C(\rho_\eta, d,\epsilon)>0$ depending only on the choice of $\rho_\eta$ and of $d,\epsilon$, such that the bounds \eqref{boundgregg} hold for $|\beta|\le d+1$ for
\begin{equation}\label{bounderrerror}
\frac{1}{C(\rho_\eta,d,\epsilon)}\sum_{i=1}^M\left(\frac{1-Q_{i,\eta}(x)}{|x|^s} - \left(\int_{\mathbb R^d}Q_{i,\eta}(y)dy\right)^{-1}\int_{\mathbb R^d}\frac{Q_{i,\eta}(y)}{|x-y|^s}dy\right)\ .
\end{equation}
\end{lemma}
We then prove the following:
\begin{proposition}\label{errorestimatecor}
With the above notations, the kernels
\begin{equation}\label{errorestimate1}
\frac{C(\rho_\eta,d,\epsilon)}{M}|x|^{-s} \pm \frac{1}{M}\sum_{i=1}^M\left(\frac{1-Q_{i,\eta}(x)}{|x|^s} - \left(\int_{\mathbb R^d}Q_{i,\eta}(y)dy\right)^{-1}\int_{\mathbb R^d}\frac{Q_{i,\eta}(y)}{|x-y|^s}dy\right)
\end{equation}
and
\begin{equation}\label{errorestimate2}
\frac{1}{M}\sum_{i=1}^M\left[\left(\int_{\mathbb R^d}Q_{i,\eta}(y)dy\right)^{-1}\int_{\mathbb R^d}\frac{Q_{i,\eta}(y)}{|x-y|^s}dy\right]\ ,
\end{equation}
as well as 
\begin{equation}\label{errorestimate3}
\frac{C(\rho_\eta,d,\epsilon)}{|x|^s} \pm \frac{Q_{i,\eta}(x)}{|x|^s},\qquad\text{ for every }1\le i\le M\ ,
\end{equation}
are positive definite.
\end{proposition}
\begin{proof}[Proof of Proposition \ref{errorestimatecor}:]
The first claim follows from lemmas \ref{derivativebounds} and \ref{fourierbdlemma}. As the bounds required in these lemmas do not depend on sign, both signs are allowed. The second claim follows by using the fact that the Fourier transform transforms convolutions into products, that $|x-y|^{-s}$ is positive definite, and the structure \eqref{qieta} of $Q_{i,\eta}$ where the integral over $r$ can be commuted with the Fourier transform, by linearity. The last claim follows directly from Lemma \ref{derivativeboundseasy} via lemma \ref{fourierbdlemma}, exactly like the first claim. 
\end{proof}
\begin{proof}[Proof of Lemma \ref{derivativebounds}:]
We will estimate first separately the terms in the large parenthesis from \eqref{bounderrerror} at fixed $i=1,\ldots, M$, then we will combine the estimates in the end. We will distinguish below between terms in $B_{3R_i}$ and terms outside of $B_{3R_i}$, because the convolution $1_{B_{3R_i/2}}*1_{B_{3R_i/2}}$ has support in $B_{3R_i}$.

\par\textbf{Step 1:}\emph{ The first term on $B_{3R_i}$.}

\par We will prove a separate bound for each of the two terms in the sum \eqref{bounderrerror} and then use the triangle inequality to sum all the terms. By keeping in mind the definition \eqref{qieta}, we find that
\begin{equation}\label{reduceq0}
\partial_x^\beta\left(\frac{1- Q_{i,\eta}(x)}{|x|^s}\right) = \frac{1}{R_i^{s+|\beta|}}\left.\partial_x^\beta\left(\frac{1- Q_{0,\eta}(x')}{|x'|^s}\right)\right|_{x'=\frac{x}{R_i}}\ .
\end{equation}
We note that $f_\eta(|x|):=1-Q_{0,\eta}(x)$ is a radial increasing function which vanishes at zero. Due to the definition \eqref{mollifballs} and to the explicit form of the function \eqref{convolballs} of which $Q_{0,\eta}$ is a superposition, $f_\eta$ has a finite right derivative at zero: Indeed, the function appearing in \eqref{convolballs} can be directly differentiated in $|x|$ and the averaging in \eqref{mollifballs} preserves these bounds. Due to the definition \eqref{mollifballs} and to the explicit form of the function \eqref{convolballs} of which $Q_{0,\eta}$ is a superposition, $f_\eta$ has a finite right derivative at zero: Indeed, the function appearing in \eqref{convolballs} can be directly differentiated in $|x|$ and the averaging in \eqref{mollifballs} preserves these bounds. Therefore by Taylor expansion, and since $f_\eta$ only depends on $\rho_\eta$ and $d$ and we assumed that $\epsilon\le s\le d-\epsilon$, we have the following control for $|\beta|=0$ and $x\neq 0$:
\begin{equation}\label{0dernearzero}
\frac{1- Q_{0,\eta}(x)}{|x|^s} \le |x|^{1-s}\lim_{r\downarrow 0}\frac{f_\eta(r)}{r} + C_{\rho_\eta,d}|x|^{1-s}= C(\rho_\eta,d,\epsilon)|x|^{1-s}\ .
\end{equation}
Similarly, based on the formulas \eqref{convolballs} and \eqref{mollifballs} and on the boundedness properties of $\rho_\eta$, we see that the norms of all the derivatives $\partial_x^\beta Q_{0,\eta}$ have a finite bound near zero, and thus all the partial derivatives are \emph{uniformly} bounded on $\mathbb R^d\setminus\{0\}$. This means that for $0<|\beta|\le {d+1}$, again using the bound \eqref{0dernearzero} for the term where all the partial derivatives fall on $|x|^{-s}$ as we apply the Leibnitz rule to distribute derivatives over the product $Q_{0,\eta}(x)|x|^{-s}$, we find the following bounds valid for $x\neq 0$:
\begin{eqnarray}
\left|\partial_x^\beta\left(|x|^{-s}(1-Q_{0,\eta}(x))\right)\right|&\le&\sum_{\beta'+\beta''=\beta}\left|\partial_x^{\beta'}(|x|^{-s})\right|\left|\partial_x^{\beta''}(1-Q_{0,\eta}(x))\right| \nonumber\\
&\le& C_1(\beta)(|x|^{-s-|\beta|}\left|(1-Q_{0,\eta}(x))\right| \nonumber\\[3mm]
&&+C_2(\beta)|x|^{1-s-|\beta|}\sum_{|\beta'|=1}\left|\partial_x^{\beta'}Q_{0,\eta}(x)\right|+O(|x|^{2-s-|\beta|}) \quad\text{as}\quad|x|\to 0\nonumber\\
&\le&C(\rho_\eta,d,\epsilon)|x|^{1-s-|\beta|} + O(|x|^{2-s-|\beta|})\quad\text{as}\quad|x|\to 0\ ,\label{betadernearzero}
\end{eqnarray}
where the constants $C_1(\beta), C_2(\beta)$ are the necessary bounds valid for all $\epsilon\le s\le d-\epsilon$, which we can verify to be depending only on $d,\epsilon$ for $0<\epsilon<d/2$. The bounds \eqref{betadernearzero} and \eqref{0dernearzero} can be applied to the r.h.s. of \eqref{reduceq0} to give
\begin{equation}\label{qinearzero}
\left|\partial_x^\beta\left(\frac{1- Q_{i,\eta}(x)}{|x|^s}\right)\right|\le C(\rho_\eta,d,\epsilon)\frac{1}{R_i}|x|^{1-s-|\beta|} + o(|x|^{1-s-|\beta|})\quad\text{as}\quad|x|\to 0\ .
\end{equation}
Also note that the expansion in powers of $|x|$ from the bounds \eqref{betadernearzero} and \eqref{0dernearzero} still holds on the whole of $B_3$ (i.e. on the whole support of $Q_{0,\eta}$), as all the growth behavior of different powers of $|x|$ does not change for $|x|\in(0,1)$, whereas for $|x|\in[1,3]$ the powers of $|x|$ contributing to leading order to the estimates are equivalent up to a constant factor depending only on $\epsilon,d$. Similarly, for each fixed $\beta\ge 0$ the derivatives $\partial_x^\beta Q_{0,\eta}$ are thus all uniformly bounded in $x$, with a constant depending only on $d,|\beta|$. Then we can appeal to the same scaling reasoning as in passing from \eqref{betadernearzero} to \eqref{qinearzero}, which leads to the fact that also on the whole of $B_{3R_i}$ the bounds \eqref{qinearzero} extend, with a finite constant depending on $d,\epsilon,|\beta|$. As we are concerned only with the finitely many (and whose number depends only on $d$) choices $0\le |\beta|\le d+1$, the constant we obtain effectively depends only on $s,d$. Summarizing this reasoning, at the cost of increasing the value of the constant $C(\rho_\eta,d,\epsilon)$, we still have for $0\le |\beta|\le {d+1}$
\begin{equation}\label{qitillri}
\left|\partial_x^\beta\left(\frac{1- Q_{i,\eta}(x)}{|x|^s}\right)\right|\le C(\rho_\eta,d,\epsilon)\frac{1}{R_i}|x|^{1-s-|\beta|}\quad\text{ for }\quad|x|\le 3R_i\ .
\end{equation}
For the remaining term in \eqref{bounderrerror} we have a similar procedure, but it is better controlled near $0$, as we now show. First note that $\int_{\mathbb R^d}Q_{i,\eta}(x)dx=C(\rho_\eta,d)R_i^d$ and $\int_{B_R(x)}|y|^{-s}dy \le \int_{B_R}|y|^{-s}dy\le c(d,\epsilon)R^{d-s}$. Using these, for $0\le|\beta|\le d+1$ we have:
\begin{multline}\label{reduceq0av}
\left|\partial_x^\beta\left(\left(\int_{\mathbb R^d}Q_{i,\eta}(y)dy\right)^{-1}\int_{\mathbb R^d}\frac{Q_{i,\eta}(y)}{|x-y|^s}dy\right) \right|\ \le\ \left(\int_{\mathbb R^d}Q_{i,\eta}(y)dy\right)^{-1}\int_{\mathbb R^d}\frac{\left|\partial_x^\beta Q_{i,\eta}(x-y')\right|}{|y'|^s}dy'\\
\le C(\rho_\eta,d,\epsilon)\frac{\sup_y|\partial_x^{\beta}Q_{0,\eta}(y)|}{R_i^{|\beta|}}\frac{1}{R_i^d}\int_{B_{2(1+\eta)R_i}(x)}\frac{dy}{|y|^s}\ \le\ \frac{C(\rho_\eta,d,\epsilon)}{R_i^{s+|\beta|}}\quad\text{for}\quad|x|<3R_i\ .
\end{multline}
Note that in the range $|x|<3R_i$ the bound \eqref{reduceq0av} is stronger than \eqref{qitillri}.

\par\textbf{Step 3:} \textit{Bound outside $B_{3R_i}$, and final estimates at fixed $i$.}

\par For $|x|>3R_i$ we note that $Q_{i,\eta}(x)=0$ and the first term from \eqref{bounderrerror} simplifies as $(1-
Q_{i,\eta}(x))|x|^{-s}=|x|^{-s}$ in this range. Thus, using the fact that the integral below is with respect to a probability measure in order to pass the term $|x|^{-s}$ inside the integral, and noting that the $x$-derivatives can be equally passed inside the integral, we are required to bound the following:
\begin{multline}\label{todolarger3ri}
\left|\int_{B_{2(1+\eta)R_i}}\frac{Q_{i,\eta}(y)}{\int Q_{i,\eta}(z)dz}\partial_x^\beta\left(\frac{1}{|x|^{s}} - \frac{1}{|x-y|^{s}}\right)dy\right|\\
\le\ C(\rho_\eta,d,\epsilon)\ C(\beta)\int_{B_{3R_i}}\frac{Q_{i,\eta}(y)}{\int Q_{i,\eta}(z)dz}\left|\frac{1}{|x|^{s+|\beta|}} - \frac{1}{|x-y|^{s+|\beta|}}\right|dy\ .
\end{multline}
By using the fact that $|y|\le 3R_i$ and the Taylor expansion to first order of $|x|^{-s-|\beta|}$ near $x$, we have the bound by $R_i|x|^{-1-s-|\beta|}$, for $0\le|\beta|\le {d+1}$ in \eqref{todolarger3ri} for $|x|>3R_i$. Adding this bound to the bounds \eqref{qitillri} and \eqref{reduceq0av}, we find
\begin{multline}\label{boundfixedi}
\left|\partial_x^\beta\left(\frac{1-Q_{i,\eta}(x)}{|x|^s} - \left(\int_{\mathbb R^d}Q_{i,\eta}(y)dy\right)^{-1}\int_{\mathbb R^d}\frac{Q_{i,\eta}(y)}{|x-y|^s}dy\right)\right|\\
\le\ \frac{1}{|x|^{s+|\beta|}}\cdot\left\{
\begin{array}{ll}
C(\rho_\eta,d,\epsilon)\frac{|x|}{R_i}&\text{ for }|x|\le 3 R_i\ ,\\
C(\rho_\eta,d,\epsilon){\frac{R_i}{|x|}}&\text{ for }|x|\ge 3R_i\ .
\end{array}\right.
\end{multline}
\textbf{Step 4:} \textit{Summing over $i$.}

\par For $r>0$ let $i_r\in\{0,\ldots, M\}$ be defined as follows:
\begin{equation}\label{defix}
i_r:=\left\{\begin{array}{l}0\text{ if }r<3R_1\ ,\\
\text{such that }r\in(3R_{i_r}, 3 R_{i_r+1}]\text{ if }R_1<r\le R_M\ ,\\
M\text{ if }r> R_M\ .\end{array}\right.
\end{equation}
Then by using \eqref{boundfixedi} we find that
\begin{multline}\label{firstbd}
\left|\partial_x^\beta\left(\frac{1-Q_{i,\eta}(x)}{|x|^s} - \left(\int_{\mathbb R^d}Q_{i,\eta}(y)dy\right)^{-1}\int_{\mathbb R^d}\frac{Q_{i,\eta}(y)}{|x-y|^s}dy\right)\right|\ 
\le\ \frac{C(\rho_\eta,d,\epsilon)}{|x|^{s+|\beta|}}\ \left(\ \sum_{i=1}^{i_{|x|}} {\frac{R_i}{|x|}}\ +\ \sum_{i=i_{|x|}}^M\frac{|x|}{R_i}\ \right)\ .
\end{multline}
Denote by $\bar C:=1+4\sqrt d |B_1|>1$ the constant appearing in Lemma \ref{cheeselemma} and such that $R_{i+1}>\bar C R_i$ for all $1\le i\le M$. Then 
\begin{subequations}\label{boundssums}
\begin{eqnarray}
\text{for }i\le i_{|x|}\text{ there holds}&&|x|\ge 3 R_{i_{|x|}}\ge 3 \bar C^{i_{|x|}-i}R_i,\ \text{ therefore }{\frac{R_i}{|x|}}\le {\frac{1}{3}\frac{1}{\bar C^{i_{|x|}-i}}}\ ,\\ 
\text{for }i> i_{|x|}\text{ there holds}&&|x|\le 3R_{i_{|x|}{+1}}\le 3 \bar C^{{i_{|x|}-i+1}}R_i,\ \text{ therefore }\frac{|x|}{R_i}\le 3 {\bar C}\frac{1}{\bar C^{i-i_{|x|}}}\ .
\end{eqnarray}
\end{subequations}
Since $\sum_{i\ge 0}\bar C^{-i}\le \frac{2\bar C-1}{\bar C-1}$ (which up to change of index can be used in both above cases), and $\bar C$ depends only on $d$,  we find from \eqref{firstbd} and \eqref{boundssums} that the sum over $i=1,\ldots,M$ of the terms appearing in \eqref{boundfixedi} is bounded by $C(\rho_\eta,d,\epsilon)|x|^{-s-|\beta|}$, and this holds for all $0\le|\beta|\le {d+1}$, which proves the claim.
\end{proof}
\subsubsection{Properties of $w$}\label{sssecw}
We recall now that from \eqref{errorw} we have
\begin{eqnarray*}
\lefteqn{w(x_1-x_2):=\left(1+\frac{C}{M}\right)|x_1-x_2|^{-s}- \int_{\Omega_l}\left(\sum_{A\in F_\omega^l}\frac{1_A(x_1)1_A(x_2)}{|x_1-x_2|^s}\right)d\mathbb P_l(\omega)}\\
&\stackrel{\eqref{lbroughqi}}{=}&\left(1+\frac{C}{M}\right)|x_1-x_2|^{-s}-\sum_{i=1}^Mc_i\frac{Q_{i,\eta}(x_1-x_2)}{|x_1-x_2|^s}\nonumber\\
&=&\frac{C}{M}|x_1-x_2|^{-s}+\frac{1}{M}\sum_{i=1}^M\frac{1-Q_{i,\eta}(x_1-x_2)}{|x_1-x_2|^s} + \sum_{i=1}^M\left(\frac{1}{M}-c_i\right)\frac{Q_{i,\eta}(x_1-x_2)}{|x_1-x_2|^s}.\nonumber\\
&=&\frac{C}{2M}|x_1-x_2|^{-s}+\frac{1}{M}\sum_{i=1}^M\left(\frac{1-Q_{i,\eta}(x_1-x_2)}{|x_1-x_2|^s} - \left(\int_{\mathbb R^d}Q_{i,\eta}(y)dy\right)^{-1}\int_{\mathbb R^d}\frac{Q_{i,\eta}(y)}{|x_1-x_2-y|^s}dy\right)\nonumber\\ 
&&+ \frac{1}{M}\sum_{i=1}^M\left(\int_{\mathbb R^d}Q_{i,\eta}(y)dy\right)^{-1}\int_{\mathbb R^d}\frac{Q_{i,\eta}(y)}{|x_1-x_2-y|^s}dy+\frac{C}{2M}|x_1-x_2|^{-s}\nonumber\\
&&+ \sum_{i=1}^M\left(\frac{1}{M}-c_i\right)\frac{Q_{i,\eta}(x_1-x_2)}{|x_1-x_2|^s}\nonumber
\end{eqnarray*}
where for $0<\epsilon<d/2$ and $\epsilon\le s\le d-\epsilon$ we can fix $C=C(\rho_\eta,d,\epsilon)>0$ depending only on the choice of $\rho_\eta$, on $d$ and on $\epsilon$, the form of which will be used in Proposition \ref{prop3ws} above and is made explicit in Lemma \ref{derivativebounds}.

\par In order to facilitate the discussions that follow, and to be able to analyse the properties of $w$, we will find it useful to introduce the following notations:
\begin{subequations}
\begin{equation}\label{notationerrorswh}
w_\#(x):=\frac{C}{2M}|x|^{-s}+\frac{1}{M} \sum_{i=1}^M\left(\frac{1-Q_{i,\eta}(x)}{|x|^s} - \left(\int_{\mathbb R^d}Q_{i,\eta}(y)dy\right)^{-1}\int_{\mathbb R^d}\frac{Q_{i,\eta}(y)}{|x-y|^s}dy\right)\ ,
\end{equation}
\begin{equation}\label{notationerrorsw1}
w_1(x)=\frac{1}{M}\sum_{i=1}^M\left(\int_{\mathbb R^d}Q_{i,\eta}(y)dy\right)^{-1}\int_{\mathbb R^d}\frac{Q_{i,\eta}(y)}{|x-y|^s}dy\ ,
\end{equation}
and
\begin{equation}\label{notationerrorsw2}
w_2(x):=\frac{C}{2M}|x|^{-s}+\sum_{i=1}^M\left(\frac{1}{M}-c_i\right)\frac{Q_{i,\eta}(x)}{|x|^s}\ .
\end{equation}
\end{subequations}
Then 
\[
w(x)=w_\#(x)+w_1(x)+w_2(x)\ .
\]
We observe that, as explained in Lemma \ref{prop3wsb}, the form chosen for $w_\#$ and $w_1$ is used as much for positive definiteness as to derive the rough next order lower bound. More precisely, the $w_1$ term is key to us being able to control the tail of $w_\#$ in such a way that the error from $w_\#$ is of order $1/M$. Moreover, we have
\begin{lemma}\label{prop3wsa}
Let $0<\epsilon<d/2$ and $\epsilon\le s\le d-\epsilon$. Then there exists a constant $C=C(\rho_\eta,d,\epsilon)$ depending only on the choice of $\rho_\eta$, on $d$ and on $\epsilon$ such that for $w_\#, w_1$ and $w_2$ defined as in \eqref{notationerrorswh},  \eqref{notationerrorsw1} and  \eqref{notationerrorsw2}, 
the following properties hold:
\begin{enumerate}
\item $w_\#(x)$ is positive definite.
\item $w_1$ is continuous, positive definite and $w_1(x)\le \frac{C}{M} R_1^{-s}$.
\item $w_2$ is positive definite. 
\end{enumerate}
\end{lemma}
\begin{proof}
The first point and the positive definiteness part of the second point are precisely the first two statements of Proposition \ref{errorestimatecor}.

\par The boundedness of $w_1(x)$ follows by noting that due to the scaling \eqref{qieta}, there holds $\int_{\mathbb R^d}Q_{i,\eta}(y)dy= C(\rho_\eta,d) R_i^d$ and due to the upper bound \eqref{sptqieta} and to the fact that $\eta\le1/2$ we have for each fixed $i$ that
\begin{multline*}
\frac{1}{\int Q_{i,\eta}(z)dz}\int\frac{Q_{i,\eta}(x-y)}{|y|^s}dy\ \le\ \frac{1}{C(\rho_\eta,d)R_i^d}\int_{B_{3R_i}(x)}|y|^{-s}dy\\
\le\ \frac{1}{C(\rho_\eta,d)R_i^d}\int_{B_{3R_i}}|y|^{-s}dy\ =\ \frac{C(\rho_\eta,d,\epsilon)}{R_i^s}\ .
\end{multline*}
As $R_i$ forms a geometric series, summing this bound over $i$ and dividing by $M$ gives the desired bound on $w_1$, up to increasing $C(\rho_\eta,d,\epsilon)$ to a value which for $0<\epsilon<d/2$ and $\epsilon\le s\le d-\epsilon$ still depends only on $\rho_\eta,d,\epsilon$.

\par For proving the positive definiteness in the last point we recall the bounds from the Swiss cheese lemma, \eqref{cheesebound}, from which it follows that for each $i$ we have 
\begin{equation}
\frac{1}{M}-c_i \in\left[\frac{1}{M}-\frac{1}{M+C_d}, \frac{1}{M} - \frac{1}{M+C_d+1}\right]\subset\frac{1}{M^2}\left[\frac{C_d}{1+C_d}, C_d+1\right]\ .
\end{equation}
Together with \eqref{errorestimate3} from Proposition \eqref{errorestimatecor} this proves that for every $i=1,\ldots,M$ and up to increasing the value of $C=C(\rho_\eta,d,\epsilon)$ to a value still depending only on $\rho_\eta,d,\epsilon$,  we have that $M^{-2}C|x|^{-s} \pm (M^{-1} - c_i)|x|^{-s}Q_{i,\eta}(x)$ is positive definite. The constants are increased only a finite number of times (with the upper bound on that number depending only on $d$). Up to replacing $g$ from Lemma \ref{fourierbdlemma} by $g/C$, we obtain that a bound on the derivatives of $g$ by a $C|x|^{-s-|\beta|}$ implies that $C'|x|^{-s}-g(x)$ is positive definite, where $C'=c(d,\epsilon)C$ and $c(d,\epsilon)$ is bounded and depending only on $d,\epsilon$ if $0<\epsilon<d$ and if $s$ is in the range $(\epsilon, d-\epsilon)$ as before. In other words, use of the Lemma \ref{cheeselemma} just introduces yet another constant, depending only on $d$ and on $\epsilon$.

\par Summing the terms over $i=1,\ldots, M,$ we find the last claim of our lemma, and this concludes the proof.
\end{proof}
\subsection{Rough next-order lower bound for $w_\#$ and $w_2$}\label{ssecroughbounds}
The present section produces a lower bound for the $E_{N,\mathsf{c}}^\mathrm{xc}(\mu)$ energy in the case of costs $d$ of a special form, inspired by the above constructions. This is based on the proof of a rough fractional Lieb-Oxford inequality from \cite[Appendix]{LundNamPort}, which itself is inspired by \cite[Lem. 5.3]{liebsolovejyngvason} (see also \cite{Lieb79}, \cite{Li83}, \cite{LO81}. Translated to our notation, the bound proved in \cite[Lem. 16]{LundNamPort} states that for cost $\mathsf{c}(x,y)=|x-y|^{-s}$ for $0<s<d$, and for any transport plan $\gamma_N\in\mathcal P_{sym}((\mathbb R^d)^N), \gamma_N\mapsto\mu$, for $d\mu(x)=\rho(x)dx$ with $\rho\in L^{1+s/d}(\mathbb R^d)$, there holds
\begin{multline}\label{boundlnp}
\int \sum_{\substack{1\le i,j\le N\\i\neq j}}\mathsf{c}(x_i,x_j)d\gamma_N(x_1,\ldots,x_N) - N^2\int_{\mathbb R^d}\int_{\mathbb R^d}\mathsf{c}(x,y)\rho(x)\rho(y)dx\, dy\\
\ge\ - C N^{1+s/d}\int_{\mathbb R^d}\rho^{1+s/d}(x) dx\ .
\end{multline}
The proof of \eqref{boundlnp} done in \cite{LundNamPort} is based on the following radial decomposition formula, that seems to first have been used by Fefferman-De la Llave \cite{feffermanllave} in a statistical mechanics context:
\begin{equation}\label{feffll}
\frac{1}{|x-y|^s}={c(s,d)}\int_0^\infty\int_{\mathbb  R^d} 1_{B_r}(x-u)1_{B_r}(y-u)du \frac{dr}{r^{d+1+s}}\quad\text{for}\quad x\neq y\ ,
\end{equation}
which is well-defined for $s>0$ (the reason why \eqref{boundlnp} only holds for $0<s<d$, is that for $s\ge d$ the energy in \eqref{boundlnp} is $-\infty$ as $|x-y|^{-s}$ stops being integrable near $x=y$).

\par Even though the lemma below can be shown for much more general costs, we restrict ourselves to showing it only for $w_\#, w_1$ and $w_2$. We are now ready to show:
\begin{lemma}\label{prop3wsb}
Let $0<\epsilon<d/2$ and $\epsilon\le s\le d-\epsilon$, and let $w_\#,w_1$ and $w_2$ be defined as in \eqref{notationerrorswh}, \eqref{notationerrorsw1}  and \eqref{notationerrorsw2}. Then for all $\mu\in\calP(\mathbb{R}^d)$ with density $\rho\in L^{1+\frac{s}{d}}(\mathbb R^d)$ we have
\begin{eqnarray}
E_{N, w_\#}^\mathrm{xc}(\mu)&\ge&-\frac{C(w_\#,d,\epsilon)}{M}N^{1+s/d}\int_{\mathbb{R}^d}\rho^{1+s/d}(x)dx\ ,\label{gsineqlnp}\\
E_{N, w_2}^\mathrm{xc}(\mu)&\ge&-\frac{C(w_2,d,\epsilon)}{M}N^{1+s/d}\int_{\mathbb{R}^d}\rho^{1+s/d}(x)dx\label{gsineqlnp1}
\end{eqnarray}
and
\begin{equation}
\label{gsineqlnp2}
E_{N, w_1}^\mathrm{xc}(\mu)\ge -\frac{C(\rho_\eta,d,\epsilon)}{M}R_1^{-s}(N-1)\ ,
\end{equation}
for some constants $C(w_\#,d,\epsilon), C(w_2,d,\epsilon)>0$, which are independent of the choices of $\rho, N$ and $s\in[\epsilon,d-\epsilon]$, and $C(\rho_\eta,d,\epsilon)$ is the one of Lemma \ref{prop3wsa}.

\par Similar statements as above, and with the same lower bounds constants, hold also for the $E^{\mathrm{OT}}_{\mathrm{GC},N}$ versions. 
\end{lemma}
\begin{proof}
The proof of the inequality involving $w_1$ follows immediately from Lemma 3.6 in \cite{FoLeSol15} due to its boundedness, and will be omitted. We focus next on the remaining two inequalities, whose proof relies partly on an adaptation of the proof of \cite[Lem. 16]{LundNamPort}, the main ideas of which we briefly sketch it below.  A key factor in our reasoning is that, since $w_1$ cancels the tail behaviour of $w_\#$, we preserve the factor $1/M$.

\par To begin with, we will write $w_\#$ and $w_2$ in a form similar to \eqref{feffll}. Since the reasoning is the same for both costs, we will only detail below the argument for $w_\#$. To this purpose, we will first apply \cite[Thm. 1 and eqn. (11)]{HS} to 
$w_\#(x)$. This is possible due to the fact that, by the arguments in Lemma \ref{errorestimatecor}, $w_\#$ satisfies all the conditions of \cite[eqn. (11)]{HS}. 
\par More precisely, we are going to use the following general representation as proved in \cite[Thm. 1]{HS}: Let $V:\mathbb{R}^d\to\mathbb{R}$ be a radial function which is $d+1$ times differentiable away from $x=0$. Assume also  that $\lim_{|x|\to\infty}|x|^m \partial_{|x|}^m V(x)=0$ for all $0\le m\le [d/2]+{1}$. Then
\begin{equation}\label{HSrep}
V(x)=\int_0^\infty \int_{\mathbb{R}^d} 1_{B_{r/2}}(u)1_{B_{r/2}}(x-u) f(r)du\, dr\ ,
\end{equation}
where
\begin{equation}\label{HSf}
f(r)=\frac{(-1)^{d+1}}{\Gamma([d/2]+2)}\frac{2}{(\pi r^2)^{(d-1)/2}}\int_r^\infty V^{(d+1)}(v)v (v^2-r^2)^{(d-3)/2}dv\ ,
\end{equation}
and where by abuse of notation $V(v)=V(|x|),$ with $|x|=v$. Let
\begin{eqnarray*}
w_*(|x|)&:=&\sum_{i=1}^M\left(\frac{1-Q_{i,\eta}(x)}{|x|^s} - \left(\int_{\mathbb R^d}Q_{i,\eta}(y)dy\right)^{-1}\int_{\mathbb R^d}\frac{Q_{i,\eta}(y)}{|x-y|^s}dy\right)\ ,
\end{eqnarray*}
which is a radial function since $Q_{i,\eta}(x)$ is a radial function by \eqref{qieta}, and $\int_{\mathbb R^d}\frac{Q_{i,\eta}(y)}{|x-y|^s}dy$ is radial since it is a convolution of radial functions. In view of Lemma \ref{derivativebounds} and of \cite[Thm. 1]{HS}, we have for $w_*$ the representation
\[
w_*(|x|)=\int_0^\infty \int_{\mathbb{R}^d}1_{B_{r/2}}(u)1_{B_{r/2}}(x-u) f_2(r)du\, dr\,
\]
for which we use the expression of the type \eqref{HSf} for $f_2(r)$, depending on the $d+1$-th derivative of $w_*$. More precisely we have for a constant $c(d)>0$ depending only on the dimension,
\begin{equation}\label{boundf2}
f_2(r) = (-1)^{d+1}c(d)\int_r^\infty w_*^{(d+1)}(v)v(v^2 - r^2)^{\frac{d-3}{2}}dv\ .
\end{equation}
If now $f_1(r)$ is the weight $f$ corresponding to $V(x)=C|x|^{-s}$ as obtained from \eqref{HSf}, we also have
\[
f_1(r)=c(d)C\int_r^\infty\frac{1}{v^{s+d+1}}v(v^2 - r^2)^{\frac{d-3}{2}}dv\ ,
\]
because of the fact that the $d+1$-th derivative of $|x|^{-s}$ has sign $(-1)^{d+1}$. Therefore
\begin{multline*}
w_\#(x,y)=\frac{C}{M}|x-y|^{-s}+\frac{1}{M} w_*(x-y)
=\frac{1}{M}\int_0^\infty 1_{B_{r/2}}(x-u)1_{B_{r/2}}(y-u) (f_1(r)+f_2(r))du dr\ ,
\end{multline*}
where
\begin{equation}\label{exprf1f2}
f_1(r)+f_2(r)=c(d)\int_r^\infty\left(Cv^{-s-d-1} +(-1)^{d+1}w_*(v)^{(d+1)}\right)v(v^2- r^2)^{\frac{d-3}{2}}\, dv\ .
\end{equation}
Note here that due to the bound \eqref{bounderrerror} for $|\beta|\le d+1$, we find that up to enlarging the above constant $C$ by a factor depending only on $d,\epsilon$ for our choice of $s$, there holds for a constant $\widetilde C>0$ depending only on $d,\epsilon,C$,
\begin{eqnarray}\label{intervalderivatives}
\widetilde{C} v^{-s-d-1}\ge Cv^{-s-d-1} +(-1)^{d+1}w_*(v)^{(d+1)}\ge C v^{-s-d-1} -|w_*(v)^{(d+1)}|\ge 0\ ,
\end{eqnarray}
and thus, due to the fact that the weight $ v(v^2-r^2)^{(d-3)/2}$ appearing in \eqref{exprf1f2} is positive and to the fact that $f_1(r)$ can be equivalently re-expressed also via \eqref{feffll}, we find from \eqref{intervalderivatives} that for some constant $\bar C(d,\epsilon)$ depending only on $d, \epsilon, w_\#, \rho_\eta$ there holds
\begin{eqnarray}\label{helpfulrbd}
0\ \le\ f_1(r) + f_2(r)\ \le\ c(d)\widetilde C\int_r^\infty\frac{1}{v^{s+d+1}}v(v^2 - r^2)^{\frac{d-3}{2}}dv\ \le\ \frac{\bar C(d,\epsilon)}{r^{d+s+1}}\ .
\end{eqnarray}
Fix now $\mu\in\calP(\mathbb{R}^d)$ with density $\rho\in L^{\frac{d}{d-s},1}(\mathbb R^d)$ as well. Since $f_1+f_2\ge 0$, we can proceed next as in the proof of \cite[Lem. 16]{LundNamPort}. As in \cite[(79) and (80)]{LundNamPort}, we have 
\[
\int_{\mathbb{R^{2d}}} w_\#(x,y)d\rho(x) d\rho(y)=\frac{1}{M}\int_0^\infty\int_{\mathbb{R}^d} h_{r,\rho}^2(u) (f_1(r) +f_2(r))\, du\, dr\ ,
\]
where
\[
h_{r,\rho}:=\rho*1_{B_{r/2}}\ .
\]
Moreover, for any $\gamma_N\in \mathcal{P}(\R^{Nd}), \gamma_N\mapsto\mu$, we have
\[
\int_{\mathbb{R}^{Nd}}\sum_{1\le i\neq j\le N}w_\#(x_i,x_j)\,d\gamma_N(x_1,\ldots,x_N)=\frac{1}{M}\int_0^\infty\int_{\mathbb{R}^d}k_{r,\rho}(u) (f_1(r)+f_2(r))\, du\, dr\ ,
\]
where
\[
k_{r,\rho}(u):=\int_{\mathbb{R}^{Nd}}\sum_{1\le i\neq j\le N}1_{B_{r/2}}(x_i-u)1_{B_{r/2}}(x_j-u)\,d\gamma_N(x_1,\ldots,x_N)\ .
\]
By the same Cauchy-Schwarz reasoning used in \cite{LundNamPort} to obtain \cite[eqn. (81)]{LundNamPort}, we get
\begin{eqnarray}
\label{minexclo}
k_{r,\rho}(u)&=&\int_{\mathbb{R}^{Nd}}\bigg(\sum_{i=1}^N1_{B_{r/2}}(x_i-u)\bigg)^2\,d\gamma_N(x_1,\ldots,x_N)-\int_{\mathbb{R}^{Nd}}\bigg(\sum_{i=1}^N1_{B_{r/2}}(x_i-u)\bigg)\,d\gamma_N(x_1,\ldots,x_N)\nonumber\\
&\ge&\left[\int_{\mathbb{R}^{Nd}}\bigg(\sum_{i=1}^N1_{B_{r/2}}(x_i-u)\bigg)\,d\gamma_N(x_1,\ldots,x_N)\right]^2-\int_{\mathbb{R}^{Nd}}\bigg(\sum_{i=1}^N1_{B_{r/2}}(x_i-u)\bigg)\,d\gamma_N(x_1,\ldots,x_N)\nonumber\\
&=&N^2\left[\int_{\mathbb{R}^{d}}1_{B_{r/2}}(x-u)\rho(x)\,dx\right]^2-N\int_{\mathbb{R}^{d}}1_{B_{r/2}}(x-u)\rho(x)dx=N^2h^2_{r,\rho}(u)-Nh_{r,\rho}(u)\nonumber\\
&\ge&N^2 h^2_{r,\rho}(u)-\min(Nh_{r,\rho}(u), N^2h_{r,\rho}^2(u),
\end{eqnarray}
where for the last inequality we used that $k_{r,\rho}(u)\ge 0$. We therefore have
\begin{eqnarray*}
\lefteqn{\int \sum_{1\le i,j\le N}w_\#(x_i,x_j)d\gamma_N(x_1,\ldots,x_N) - N^2\int_{\mathbb R^d}\int_{\mathbb R^d}w_\#(x,y)\rho(x)\rho(y)dx}\\
&\ge&-\frac{1}{M}\int_0^\infty\int_{\mathbb{R}^d} \min(N h_{r,\rho}(u), N^2h^2_{r,\rho}(u)) (f_1(r)+f_2(r))\, du\, dr\\
&\ge&-\frac{\bar{C}(d,\epsilon)}{M}\int_0^\infty\int_{\mathbb{R}^d} \min(N h_{r,\rho}(u), N^2 h^2_{r,\rho}(u)) \frac{1}{r^{s+d+1}}\, du\, dr\ ,
\end{eqnarray*}
where for the last inequality we applied \eqref{helpfulrbd}.

\par The proof now follows exactly by the same Hardy-Littlewood arguments as in \cite{LundNamPort}, and will be omitted.

\par To show the equivalent inequalities to (\ref{gsineqlnp}),   (\ref{gsineqlnp1}) and  (\ref{gsineqlnp2}) for the $E^{\mathrm{xc}}_{\mathrm{GC},N}$ versions, we restrict again for simplicity to $E^{\mathrm{xc}}_{\mathrm{GC},N, w_\#}(\mu)$. We will make use here of the definition of $F^{\mathrm{OT}}_{\mathrm{GC}}(\mu)$ as detailed in \eqref{OTGC} above. More precisely, let $(\lambda_0, \lambda_1,\ldots,\lambda_n,\ldots)$, $(\mu_1,\mu_2,\ldots,\mu_n,\ldots)$, be an optimizer for $F^{\mathrm{OT}}_{\mathrm{GC}}(\mu)$. Then by \eqref{minexclo} and using the same notations as before, we have 
\begin{eqnarray*}
k_{r,\rho}(u)&\ge&\sum_{n=2}^\infty\alpha_n n^2\left[\int_{\mathbb{R}^{d}}1_{B_{r/2}}(x-u)\,d\mu_n(x)\right]^2-\sum_{n=2}^\infty \alpha_n n\int_{\mathbb{R}^{d}}1_{B_{r/2}}(x-u)\,d\mu_n(x)\\
&\ge&\sum_{n=1}^\infty\alpha_n n^2\left[\int_{\mathbb{R}^{d}}1_{B_{r/2}}(x-u)\,d\mu_n(x)\right]^2-\sum_{n=1}^\infty \alpha_n n\int_{\mathbb{R}^{d}}1_{B_{r/2}}(x-u)\,d\mu_n(x)\\
&\ge&\bigg(\sum_{n=1}^\infty \alpha_n n\int_{\mathbb{R}^{d}}1_{B_{r/2}}(x-u)\,d\mu_n(x)\bigg)^2-\sum_{n=1}^\infty \alpha_n n\int_{\mathbb{R}^{d}}1_{B_{r/2}}(x-u)\,d\mu_n(x)\\
&=&N^2\left[\int_{\mathbb{R}^{d}}1_{B_{r/2}}(x-u)\rho(x)\,dx\right]^2-N\int_{\mathbb{R}^{d}}1_{B_{r/2}}(x-u)\rho(x)dx=N^2h^2_{r,\rho}(u)-Nh_{r,\rho}(u).\nonumber\\
\end{eqnarray*}
For the second inequality in the above we used that $\int_{\mathbb{R}^{d}}1_{B_{r/2}}(x-u)\,d\mu_1(x)\le 1$, for the third inequality we used that $\sum_{n=0}^\infty\alpha_n=1$, and for the first equality we applied $\sum_{n=1}^\infty n \alpha_n\mu_n = N\mu$. The argument proceeds now the same as for the $E^{\mathrm{xc}}_{N, w_\#}(\mu)$ term above.
\end{proof}
\section{Optimal function spaces estimates}\label{seccomput}
The kernel $g(x)=|x|^{-s}$ on $\mathbb R^d$ is in the Lorentz space $L^{\frac{d}{s}, \infty}(\mathbb R^d)$ (for the case where the second exponent is $\infty$ the space $L^{p,\infty}(\mathbb{R}^d)$ is also called weak-$L^p(\mathbb{R}^d)$, or Marcinkiewicz space). Recall that the space $L^{p,p}(\mathbb{R}^d))$ equals $L^p(\mathbb{R}^d)$ while for $q>p$ the space $L^{p,q}(\mathbb{R}^d)$ is slightly smaller than $L^p$ but includes all $L^{p+\epsilon}(\mathbb{R}^d), \epsilon>0$, while for $q<p$ it is slightly smaller than $L^p(\mathbb{R}^d)$ but includes all $L^{p-\epsilon}(\mathbb{R}^d),\epsilon>0$.

\par We now provide references to \eqref{finmu}:
\begin{itemize}
\item \textit{Translation of the condition that $f*g$ belongs to $L^\infty(\mathbb{R}^d)$}.
By a refined H\"older-Young inequality the first one gives $f\in L^{\frac{d}{d-s}, 1}(\mathbb{R}^d)$ (which is nothing but the dual space of $L^{\frac{d}{s},\infty}(\mathbb{R}^d)$, cf. \cite{grafa} thm. 1.4.17 (v)).
\item \textit{Translation of the condition that $f\cdot(f*g)$ belongs to $L^1(\mathbb{R}^d)$}.
Clearly then the requirement on $f*g$ is less restrictive than to be in $L^\infty(\mathbb{R}^d)$. By the refined multilinear estimate \cite{grafa} ex. 1.4.18 (which is proved by interpolation methods starting from the analogue more classical $L^p(\mathbb{R}^d)$-space version of the inequality) one gets this time, using again the hypothesis $f\in L^{\frac{d}{s},\infty}(\mathbb{R}^d)$, the requirement $f\in L^{\frac{2d}{2d-s},2}(\mathbb{R}^d)$, under which condition we have $\|f\cdot(g*f)\|_{L^1(\mathbb{R}^d)}\le \|f\|_{L^{\frac{2d}{2d-s},2}(\mathbb{R}^d)}\|g*f\|_{L^{\frac{2d}{s},2}(\mathbb{R}^d)}\le\|f\|_{L^{\frac{2d}{2d-s},2}(\mathbb{R}^d)}^2\|g\|_{L^{\frac{d}{s},\infty}(\mathbb{R}^d)}$.
\end{itemize}
\section{Lieb-Oxford bound uniform in $s$}\label{secunifboundLO}
\begin{lemma}\label{unifboundLO}
Let $N\ge 2$. Fix $0<\epsilon<d/2$, and let $\epsilon\le s\le d-\epsilon$. Then for all $\mu\in\mathcal {P}(\mathbb R^d)$ with $\rho\in L^{1+\frac{s}{d}}(\mathbb{R}^d)$, we have for some $-\infty<c_{\mathrm{LO}}(d,\epsilon)<0$ which does not depend on $N$ and $\mu$ 
\begin{equation}
\label{loterm1gs}
c_{\mathrm{LO}}(d,\epsilon)\int_{\mathbb{R}^d}\rho^{1+\frac{s}{d}}(x)dx \le N^{-1-\frac{s}{d}}E^{\mathrm{xc}}_{N,s}(\mu)\le 0\quad \mbox{and}\quad c_{\mathrm{LO}}(d,\epsilon)\int_{\mathbb{R}^d}\rho^{1+\frac{s}{d}}(x)dx \le N^{-1-\frac{s}{d}}E^{\mathrm{xc}}_{\mathrm{GC},N,s}(\mu)\le 0.
\end{equation}
\end{lemma}
\begin{proof}
We begin with the proof of the first inequality in \eqref{loterm1gs}. To start the proof, a careful analysis of the argument in \cite[Lem. 16]{LundNamPort} reveals that there exists $c_{\mathrm{LO}}(d,\epsilon)<0$ such that for all $\epsilon\le s\le d-\epsilon$.
 \begin{equation}
 \label{interlim}
 c_{\mathrm{LO}}(s,d)>c_{\mathrm{LO}}(d,\epsilon).
 \end{equation}
To see why \eqref{interlim} is true, we note that by the last inequality in the proof of \cite[Lem. 16]{LundNamPort}
\[c_{\mathrm{LO}}(s,d)=-\frac{d \bar{c}_{s,d}M_{s,d}}{2s(d-s)}|B_1|^{1+s/d},
\]
for some $\bar{c}_{s,d},M_{s,d}>0$. Let us discuss the different constants appearing in the above formula:
\begin{itemize}\item The constant ${\bar{c}}_{s,d}$ in \cite[Lem. 16]{LundNamPort} comes from the Fefferman-de Llave representation and the allowed range of the exponent $s$ is $0<s<d$. Moreover for $\epsilon\le s\le d-\epsilon$ it holds that $0<\sup_{s\in [\epsilon,d-\epsilon]}\bar{c}_{s,d}<\infty$  (see \cite{feffermanllave} for the Fefferman-de Llave representation for the Coulomb potential, \cite[Thm. 9.8]{LiLo} for homogenous potentials, and \cite[Thm. 1]{HS} for more general potentials). 
\item The $M_{s,d}$ comes from the application in the proof of \cite[Lem. 16]{LundNamPort} of the Hardy-Littlewood maximal inequality. As explained for example in \cite{Tao}, the standard Vitali-type covering argument used to establish the inequality gives $M_{s,d}=d C^d s^{-1}$ for some $C>1$ independent of $s$ and $d$.
\end{itemize}

The proof of the second inequality in \eqref{loterm1gs} uses the first bound from \eqref{loterm1gs} by applying the same arguments used in Lemma \ref{prop3wsb} to lower-bound the $E^{\mathrm{xc}}_{\mathrm{GC}}$ terms therein.
\end{proof}

\section{Proof of Remark \ref{rmk_gcb}}
\label{Remark4.9}
\textbf{Proof of 2).}

Towards proving \eqref{subadgcinit}, we use choices $\{\alpha_n',\mu_n',\gamma_n'\}_{n\ge 0}, \{\alpha_n'',\mu_n'',\gamma_n''\}_{n\ge 0}$, that realize the minimum for $F_{\mathrm{GC},N',\mathsf{c}}(\mu')$ and $F_{\mathrm{GC},N'',\mathsf{c}}(\mu'')$, respectively.  Since the corresponding optimal transport problems do not exist for $n=0,1$, we take by convention $\gamma_0'=\gamma_0''=1, \gamma_1'=\mu_1',\gamma_1''=\mu_1''$, in order to construct a competitor for the left-hand side in \eqref{subadgcinit}. Let then $\gamma_{i, n-i}:=\gamma_i'\otimes\gamma_{n-i}''\in \mathcal P^n({\mathbb{R}^d})$ for all $n\ge 2,$ and $0\le i\le n$. We symmetrise $\gamma_{i,n-i}$ by the formula
\begin{equation}\label{symmetrizationconvex}
{\widetilde{\gamma}}_{i,n-i}(B)=\frac{1}{n!}\sum_{\sigma\in S_n}\gamma_{i,n-i}(\{(x_{\sigma(1)},\ldots,x_{\sigma(n)}): (x_1,\ldots, x_n)\in B\}),
\end{equation}
for all Borel sets $B\subseteq(\mathbb{R}^d)^n$, where $S_n$ is the permutation group on $n$ symbols. Note that ${\widetilde{\gamma}}_{i,n-i}\mapsto \frac{i\mu_i'+(n-i)\mu_{n-i}''}n$. Define now for all $n\ge 0$
\begin{subequations}\label{newcompet}
\begin{equation}\label{newcompetalpha}
\alpha_n:=\sum_{i=0}^n\alpha_i'\alpha_{n-i}'',
\end{equation}
and if $\alpha_n\neq 0$, let for all $n\ge 1$
\begin{equation}\gamma_n:=\frac{1}{\alpha_n}\left(\sum_{i=0}^n\alpha_i'\alpha_{n-i}''{\widetilde{\gamma}}_{i,n-i}\right)\in \mathcal P^n_{sym}({\mathbb{R}^d}), ~\mu_n:=\frac{1}{\alpha_n}\left(\sum_{i=0}^n\alpha_i'\alpha_{n-i}''\frac{i\mu_i'+(n-i)\mu_{n-i}''}n\right).
\end{equation}
\end{subequations}
Then $\gamma_n\in\mathcal P^n(\mathbb R^d)$ and has marginal $\mu_n$, therefore by re-arranging the terms in the summations we have
\begin{eqnarray}\label{marginal}
\sum_{n=1}^\infty n\alpha_n\mu_n&=&\sum_{n=1}^\infty\sum_{i=0}^n\alpha_i'\alpha_{n-i}''\left(i\mu_i'+(n-i)\mu_{n-i}''\right)=\sum_{k=0}^\infty\sum_{i,n=0\atop n-i=k}^\infty\alpha_i'\alpha_{n-i}''\left(i\mu_i'+(n-i)\mu_{n-i}''\right)\nonumber\\
&=&\sum_{k=0}^\infty\left(\alpha''_k\sum_{n=1}^\infty n\alpha_n'\mu_n'+k\alpha_k''\mu_k''\sum_{n=0}^\infty\alpha_n'\right)=N'\mu'+N''\mu'',
\end{eqnarray}
and similarly $\sum_{n=0}^\infty\alpha_n=1$. From the definition of $\gamma_{i,n-i}$ and \eqref{symmetrizationconvex}, we have for all $n\ge 2$
\begin{multline}\label{firstpart}
\int_{({\mathbb{R}^d})^n}\sum_{k\neq  l=1}^nc(x_k,x_l)\,d {\widetilde{\gamma}}_{i,n-i}(x_1,\ldots,x_n)\\=F_i(\mu_i')+F_{n-i}(\mu_{n-i}'')+2{i(n-i)}\int_{{\mathbb{R}^d}\times {\mathbb{R}^d}}c(x,y){d\mu_i'(x)d\mu_{n-i}''(y)}.
\end{multline}
Due to \eqref{marginal} use the choice \eqref{newcompet} as a competitor for $F_{\mathrm{GC},N'+N'',\mathsf{c}}(\mu)$. Combining \eqref{firstpart} and using linearity, we get
\begin{eqnarray*}
F_{\mathrm{GC},N'+N'',\mathsf{c}}^\mathrm{OT}(\mu)&\le&\sum_{n=2}^\infty\sum_{i=0}^n\alpha'_i\alpha_{n-i}''\int_{({\mathbb{R}^d})^n}\sum_{1\le k\neq  l\le n}c(x_k,x_l)\,d {\widetilde{\gamma}}_{i,n-i}(x_1,\ldots,x_n)\\
&=&\sum_{n=2}^\infty \alpha_n'F_n(\mu_n')+ \sum_{n=2}^\infty \alpha_n''F_n(\mu_n'')+2\sum_{n=2}^\infty\sum_{i=0}^n\alpha'_i\alpha_{n-i}''i(n-i)\int_{\mathbb{R}^d\times \mathbb{R}^d}c(x,y)d\mu_i'(x)d\mu_{n-i}''(y)\\
&=&F_{\mathrm{GC},N',\mathsf{c}}^\mathrm{OT}(\mu')+F_{\mathrm{GC},N'',\mathsf{c}}^\mathrm{OT}(\mu'')+2 N'N''\int_{\mathbb{R^d}\times\mathbb{R}^d}\frac{1}{|x-y|^s}d\mu'(x)d\mu''(y).
\end{eqnarray*}
Then \eqref{subadgcinit} follows by applying the formula for the mean field term 
\begin{eqnarray*}(N'+N'')^2\int_{\mathbb{R^d}\times\mathbb{R}^d}\frac{1}{|x-y|^s}d\mu(x)d\mu(y)&=&(N')^2\int_{\mathbb{R^d}\times\mathbb{R}^d}\frac{1}{|x-y|^s}d\mu'(x)d\mu'(y)+(N'')^2\int_{\mathbb{R^d}\times\mathbb{R}^d}\frac{1}{|x-y|^s}d\mu''(x)d\mu''(y)\\
&&+2 N'N''\int_{\mathbb{R^d}\times\mathbb{R}^d}\frac{1}{|x-y|^s}d\mu'(x)d\mu''(y).
\end{eqnarray*}
As already noted in the introduction, for the $N'\le 1$ case, it suffices to take instead for $F_{\mathrm{GC},N',\mathsf{c}}(\mu')$ a competitor of the form $\alpha'_n=0,n\ge 2, \alpha'_1=N',\alpha'_0=1-N', \mu'_n=\mu', n\ge 1$, and then proceed as above.

\textbf{Proof of 3).}

To show (\ref{subadd_gcbinit}), we can apply the same reasoning as in the previous point with minor modifications, as follows. Take $\{\alpha_n',\mu_n',\gamma_n'\}_{n\ge 0}, \{\alpha_n'',\mu_n'',\gamma_n''\}_{n\ge 0}$, which realize the minima of $F_{\mathrm{GCB},N',\mathsf{c}}(\mu',\bar N'), F_{\mathrm{GCB},N'',\mathsf{c}}(\mu'',\bar N'')$, respectively. Here for $n>\bar N'$ we may arbitrarily assign $\mu_n',\gamma_n'$ because the coefficients $\alpha_n'$ are zero due to \eqref{OTGCBbis}, and similarly for $n>\bar N''$. Then we use the definitions \eqref{newcompet} verbatim, and in \eqref{newcompetalpha} one checks that, since $\alpha_n$ is the sum of terms of the form $\alpha_{n'}'\alpha_{n''}'',$ with $n'+n''=n, n'\ge 0, n''\ge 0$, this term can only be nonzero if both $\alpha_{n'}'$ and $\alpha_{n''}''$ are nonzero, i.e. if $n'\le \bar N', n''\le \bar N'',$ for some of the terms. From this it follows that $n\le \bar N'+\bar N''$ is a necessary condition for $\alpha_n$ to be nonzero. We claim that the choices \eqref{newcompet} give a competitor for $E_{\mathrm{GCB},N'+N'',\mathsf{c}}(\mu,\bar N'+\bar N'')$. Indeed, the vanishing conditions on the $\alpha_n$ are what we just verified, and the remaining conditions from \eqref{OTGCBbis} are checked exactly like for the $\mathrm{GC}$-problem, treated in the previous point. By the same passages as for the $\mathrm{GC}$-problem, we then obtain (\ref{subadd_gcbinit}).

\section{Some helpful optimal transport results}\label{ssecoptimaltransport}
Here we assume that $\mathsf{c}$ satisfies \eqref{ass_c_main_intro}, which we recall here:
\begin{subequations}\label{ass_c_gcb}
	\begin{equation}\label{ass_c_main}
	\mathsf{c}(x,y)=g(x-y)=l(|x-y|)\quad\text{where}\quad l:[0,\infty)\to[0,\infty)\quad\text{is}\quad\left\{\begin{array}{l}\mbox{continuous on $(0,\infty)$}\,\\\mbox{strictly decreasing}\ ,\\\mbox{such that }\lim_{t\to 0^+}l(t)=+\infty\ .\end{array}\right.
	\end{equation}
 Hypotheses \eqref{ass_c_gcb} are satisfied by $g(x)=|x|^{-s}$ for $s>0$, but we note here that \eqref{ass_c_main} is not satisfied by $g(x) = -\log|x|$. 
\end{subequations}
\begin{lemma}[Lower semi-continuity of $F^{\mathrm{OT}}_{\mathrm{GCB},N,\mathsf{c}}$ and $F^{\mathrm{OT}}_{\mathrm{GC},N,\mathsf{c}}$]\label{gcotlim}
	Let $\mathsf{c}$ be as in \eqref{ass_c_main} and let $(\mu_k)_{k\ge 1},\mu\in {\mathcal P}({\mathbb{R}}^d),$ such that the integrals required to define the quantities below are finite. Further suppose that $\mu_k$ converges weakly to $\mu$ and that $\sup_{k\ge 1} F^{\mathrm{OT}}_{\mathrm{GC},N,\mathsf{c}}(\mu_k)<\infty$ (it suffices that $\sup_{k\ge 1}\int_{\mathbb{R}^d}\int_{\R^d}\mathsf{c}(x,y)\ \mathrm{d}\mu_k(x)\mathrm{d}\mu_k(y)<\infty$ for this purpose). Fix $\bar N\ge N\in\mathbb{R}_+,N\ge 2$. Then there holds 
	\begin{equation}\label{grail2bis}
	\liminf_{k\to \infty}F^{\mathrm{OT}}_{\mathrm{GC},N,\mathsf{c}}(\mu_k)\ge F^{\mathrm{OT}}_{\mathrm{GC},N,\mathsf{c}}(\mu)\quad\mbox{and}\quad\liminf_{k\to \infty}F^{\mathrm{OT}}_{\mathrm{GCB},N,\mathsf{c}}(\mu_k, \bar N)\ge F^{\mathrm{OT}}_{\mathrm{GCB},N,\mathsf{c}}(\mu,\bar N).
	\end{equation}
\end{lemma}
\begin{proof}
	We do the proof only for the $\mathrm{GC}$-problem, as the $\mathrm{GCB}$-problem works similarly, and requires the only extra fact that the condition $\lambda_{n,k}=0$ for $n>\bar N$ passes to weak limits.
	
	\textbf{Step 0.} \textit{Notation and setting.} 
	
	Let $((\lambda_{n,k})_{n\ge 0}, (\gamma_{n,k})_{n\ge 2})$ be a minimizer for the $F^{\mathrm{OT}}_{\mathrm{GC},N,\mathsf{c}}(\mu_k)$ problem. Set:
	\begin{equation}\label{lambdas}
	\lambda_k:=\sum_{n\ge 0}\lambda_{n,k}\delta_n\in\mathcal P(\N_{\ge 0}),\quad\mbox{where}\quad\N_{\ge 0}:=\left\{0,1,2,\ldots\right\}.
	\end{equation}
	We define next the disjoint union space 
	\begin{equation}\label{defx}
	X_{\mathrm{GC}}:=\coprod_{n\ge 2} X_n\quad\mbox{where we denote}\quad X_n:=(\R^d)^n, n\ge 2.
	\end{equation}
	Noting that the Borel sets of a disjoint union $\mathcal B(\coprod_{n\ge 2} X_n)$ are generated by $\cup_{n\ge 2}\mathcal B(X_n)$, we may build measures $\gamma_{\mathrm{GC},N,k}\in \mathcal M_+(\coprod_{n\ge2}A_n)$ by defining 
	\begin{equation}\label{gammagcs}
	\gamma_{\mathrm{GC},N,k}\left(A_n\right):=\lambda_{n,k}\gamma_{n,k}(A_n)\quad\mbox{for all choices of Borel sets}\quad A_n\subseteq X_n,\ n\ge 2. 
	\end{equation}
	The two measures $\gamma_{\mathrm{GC},N,k}, \lambda_k$, defined above are related by the disintegration relation
	\begin{subequations}
		\begin{equation}\label{gammalambda}
		\forall n\ge 2,\quad \gamma_{\mathrm{GC},N,k}(X_n)=\lambda_k(\{n\}),
		\end{equation}
		and if $\gamma^{(m)}_{n,k}, m\ge1,$ is the $m$-marginal measure of $\gamma_{n,k}$, the constraint from the definition of $F^{\mathrm{OT}}_{\mathrm{GC}, N,\mathsf{c}}(\mu_k)$ gives the following relation on $\mu_k$:
		\begin{equation}\label{gammalambdamu}
		\sum_{n\ge 1}n\lambda_{n,k}\gamma_{n,k}^{(1)}=N\mu_k.
		\end{equation}
	\end{subequations}
	Given two measures $\lambda_k^*\in\mathcal P(\N_{\ge 2})$ and $\gamma_k^*\in\mathcal M_+(X_{\mathrm{GC}})$, the requirement that they satisfy \eqref{gammagcs} is  equivalent to saying that $(\lambda_k^*,\gamma_k^*)$, encode a competitor for $F^{\mathrm{OT}}_{\mathrm{GC}, N,\mathsf{c}}(\mu_k)$, given by $(\vec\lambda_k^*,\vec\gamma_k^*)$ with
	\begin{subequations}\label{encode_competitor}
		\begin{equation}
		\lambda_{n,k}^*:=\lambda_k^*(\{n\})\quad\mbox{for all}\quad n\ge 2,
		\end{equation}
		and for $\lambda_{n,k}^*>0$,
		\begin{equation}
		\gamma_{n,k}^*(A_n):=\left(\lambda_{n,k}^*\right)^{-1}\gamma_k^*( A_n), \quad A_n\subseteq X_n, n\ge 2.
		\end{equation}
	\end{subequations}
	Provided (\ref{encode_competitor}) holds, we further define
	\begin{subequations}\label{encode_competitorn01}
		\begin{equation}
		\lambda_k^*(\{0\}):=\lambda_{0,k}^*,\,\lambda_k^*(\{1\}):=\lambda_{1,k}^*,\quad \mbox{where}~~ \lambda_{1,k}^*:=N-\sum_{n=2}^\infty n\lambda_{n,k}^*,\quad \lambda_{0,k}^*:=1-\sum_{n=1}^\infty\lambda_{n,k}^*,
		\end{equation}
		and if $\lambda_{1,k}^*>0$, then we also define
		\begin{equation}
		\left(\gamma_{1,k}^*\right)^{(1)}(A_1):=N\mu(A_1)-\sum_{n\ge 2}n\lambda_{n,k}\left(\gamma_{n,k}^*\right)^{(1)}(A_1)\quad\mbox{for}\quad A_1\in\mathcal B(\mathbb R^d).
		\end{equation}
	\end{subequations}
	Assume $\mu$ has support $\Lambda$ and fix $\delta>0$. Let $\Lambda_R:=\mathrm{supp}(\mu)\cap B_R$, with $R=R(\delta)>0$ chosen such that
	\begin{equation}\label{spt_mu}
	\mu(\mathbb R^d\setminus \Lambda_R)<\delta,\quad\mbox{and}\quad\forall k\ge 1,\quad\mu_k(\mathbb R^d\setminus \Lambda_R)<\delta.
	\end{equation}
	This choice is possible because $\mu, \mu_k,$ are finite measures. Denote now for all $n\ge 2$ by
	\begin{equation}\label{ban}
	\Lambda^n_{R,\alpha}:=\bigg\{(x_1,\ldots,x_n)\in \Lambda_R^n: \min_{1\le i\neq j\le n}|x_i-x_j|\ge\alpha\bigg\},
	\end{equation}
	and let
	\begin{equation}\label{ba}
	\Lambda_{R,\alpha}:=\left(\coprod_{2\le n<n_\alpha} \Lambda^n_R\right)\coprod\left(\coprod_{n\ge n_\alpha} \Lambda^n_{R,\alpha}\right)\subset X_{\mathrm{GC}}.
	\end{equation}
	\medskip
	
	\textbf{Step 1.} \textit{Mass bound uniform in $k$.}
	Denoting the $2$-marginals of the $\gamma_{n,k}$ by $\gamma^{(2)}_{n,k}$, we obtain
	\begin{equation}
	F^{\mathrm{OT}}_{\mathrm{GC},N,\mathsf{c}}(\mu_k)=\sum_{n=2}^\infty n(n-1)\lambda_{n,k}\int_{{\mathbb{R}}^{2d}}\mathsf{c}(x,y)\ d\gamma^{(2)}_{n,k}(x,y).
	\end{equation}
	Due to our hypothesis assumptions, and in view of (\ref{EGC_negative}), we have
	\begin{eqnarray}\label{massbound}
	\sup_{k\ge 1} F^{\mathrm{OT}}_{\mathrm{GC},N,\mathsf{c}}(\mu_k)&=&\sup_{k\ge 1}\sum_{n=2}^\infty n(n-1)\lambda_{n,k}\int_{\mathbb{R}^{2d}}\mathsf{c}(x,y)\ \mathrm{d}\gamma_{n,k}^{(2)}(x,y)\nonumber\\
	&\le&N^2\,\sup_{k\ge 1}\int_{\mathbb{R}^d}\int_{\R^d}\mathsf{c}(x,y)\ \mathrm{d}\mu_k(x)\mathrm{d}\mu_k(y):= N^2\,C.
	\end{eqnarray}
	\textbf{Step 2.} \textit{The $\Lambda_{R,\alpha}^n$ contain the mass up to small error.}
	
	Let $\alpha\in (0,1[$. Then we have
	\begin{equation}\label{gammanssym0}
	\gamma_{n,k}(X_n\setminus \Lambda_{R,\alpha}^n)=\gamma_{n,k}(X_n\setminus \Lambda^n_R)+\gamma_{n,k}(\Lambda^n_R\setminus \Lambda_{R,\alpha}^n).
	\end{equation}
	For the first term in \eqref{gammanssym0} we note that $X_n\setminus \Lambda^n_R$ is included in the union of the sets $(\R^d)^i\times(\R^d\setminus \Lambda_R)\times(\R^d)^{n-1-i}$ for $i=0,\ldots,n-1$, each of which projects to $\R^d\setminus \Lambda_R$, and thus has $\gamma_{n,k}$-measure at most $\epsilon$ due to \eqref{spt_mu}. This means that
	\begin{equation}\label{gammanssym1}
	\gamma_{n,k}(X_n\setminus \Lambda^n_R)\le \sum_{i=0}^{n-1} \gamma_{n,k}\left((\R^d)^i\times(\R^d\setminus \Lambda_R)\times(\R^d)^{n-1-i}\right)\le n\delta.
	\end{equation}
	For the second term in \eqref{gammanssym0} we estimate
	\begin{eqnarray}\label{gammanssym2}
	\gamma_{n,k}\left(X_n\setminus \Lambda_{R,\alpha}^n\right)&=&\gamma_{n,k}\left(\left\{(x_1,\ldots,x_n)\in \Lambda^n_R: \min_{1\le i\neq j\le n}|x_i-x_j|<\alpha\right\}\right)\nonumber\\
	&=&n(n-1)\gamma_{n,k}^{(2)}\left(\left\{(x,y)\in \Lambda_R\times\Lambda_R: |x-y|<\alpha\right\}\right)\nonumber\\
	&=&n(n-1)\gamma_{n,k}^{(2)}\left(\left\{(x,y)\in \Lambda_R\times\Lambda_R: \mathsf{c}(x,y)>l(\alpha)\right\}\right)\nonumber\\
        &\le&n(n-1)(l(\alpha))^{-1}\int_{\Lambda_R\times\Lambda_R}\mathsf{c}(x,y)\ d\gamma^{(2)}_{n,k}(x,y).
	\end{eqnarray}
	For the second equality in the above, we used that $\gamma_{n,k}$ is a symmetric measure and the definition of $\gamma_{n,k}^{(2)}$, for the third equality we applied the properties of $l$ from (\ref{ass_c_main}), and for the inequality we used Markov's inequality. Then \eqref{gammanssym0}, \eqref{gammanssym1} and \eqref{gammanssym2} give 
	\begin{equation}\label{gammanssym}
	\gamma_{n,k}\left(X_n\setminus \Lambda^n_{R,\alpha}\right)\le n\delta + n(n-1)\ (l(\alpha))^{-1}\int_{\Lambda_R\times\Lambda_R}\mathsf{c}(x,y)\ d\gamma^{(2)}_{n,k}(x,y)
	\end{equation}
	Due to \eqref{massbound}, we now have from \eqref{gammanssym} that
	\begin{eqnarray}\label{bound1}
	\sup_{k\ge 1}\bigg[\sum_{n\ge 2}\lambda_{n,k}\gamma_{n,k}\left(X_n\setminus \Lambda_{R,\alpha}^n\right)\bigg]&\le&\sup_{k\ge 1}\bigg[\sum_{n\ge 2}n(n-1)\lambda_{n,k}l(\alpha)\int_{\Lambda_R\times\Lambda_R}\mathsf{c}(x,y)\ d\gamma^{(2)}_{n, k}(x,y)\bigg]+\sup_{k\ge 1}\bigg[\sum_{n\ge 2}n\lambda_{n,k}\delta\bigg]\nonumber\\
	&\le& (l(\alpha))^{-1}\ N^2 C+N\delta.
	\end{eqnarray}
	\textbf{Step 3.} \textit{$\gamma_{n,k}$ give no mass to the $\Lambda_{R,\alpha}^n$ for large $n$.} 
	
	We note that, by definition \eqref{ban} of the $\Lambda_{R,\alpha}^n$, if 
	\begin{equation}\label{condition}
	(x_1,\ldots,x_n)\in \Lambda^n_{R,\alpha},
	\end{equation}
	then the balls $B_{\alpha/2}(x_i)$, with $i=1,\ldots,n$, are disjoint and contained in the enlargement of $B_R$ given by $B_{R+\alpha/2}$, and so by volume comparison we have
	\begin{equation}
	n\ 2^{-d}\alpha^d=\sum_{i=1}^n\frac{|B_{\alpha/2}(x_i)|}{|B_1|}\le \frac{|B_{R+\alpha/2}|}{|B_1|}=\left(R+\frac{\alpha}{2}\right)^d,
	\end{equation}
	and thus 
	\begin{equation}\label{bound2}
	n\le n_\alpha:=\left(\frac{R}{\alpha/2}+1\right)^d.
	\end{equation}
	Therefore for $n>n_\alpha$ we have that condition \eqref{condition} cannot hold. Therefore 
	\begin{equation}\label{factii}
	\mbox{For all}\ n> n_\alpha\ \mbox{there holds}\quad\op{supp}(\gamma_{n,k})\subset X_n\setminus \Lambda_{R,\alpha}^n.
	\end{equation}
	\textbf{Step 4.} \textit{Tightness of the $(\lambda_k)_{ k\ge 1}$.} 
	
	Due to the result \eqref{factii} of Step 3, together with the result \eqref{bound1} of Step 2, 
	we find that
	\begin{eqnarray}\label{bound_high_n}
	\sup_{k\ge 1}\lambda_k\bigg(\N_{\ge 0}\setminus\left\{0,1,2,3,\ldots,n_\alpha\right\}\bigg)&\stackrel{\text{\eqref{lambdas},\eqref{gammalambda}}}{=}&\sup_{k\ge 1}\bigg[\sum_{n\ge n_\alpha}\lambda_{n,k}\bigg]\stackrel{\text{\eqref{factii}}}=\sup_{k\ge 1}\bigg[\sum_{n\ge n_\alpha}\lambda_{n,k}\gamma_{n,k}\left(X_n\setminus \Lambda_{R,\alpha}^n\right)\bigg]\nonumber\\
	&{\le}&\sup_{k\ge 1}\bigg[\sum_{n\ge 2}\lambda_{n,k}\gamma_{n,k}\left(X_n\setminus \Lambda_{R,\alpha}^n\right)\bigg]\nonumber\\
	&\stackrel{\text{\eqref{bound1}}}{\le}&
	(l(\alpha))^{-1}\cdot N^2 C+N\delta.
	\end{eqnarray}
	As the set $\{0,1,2,3,\ldots,n_\alpha\}\subset\N_{n\ge 0}$ is finite, and thus compact, this shows in particular that the family of probability measures 
	$$\lambda_k=\sum_{n\ge 0}\lambda_{n,k}\delta_n,\quad k\ge 1,$$
	is tight, and thus a fixed subsequence of $k\uparrow \infty$ which realizes the $\liminf$ from \eqref{grail2bis} has a further subsequence which we denote $(k_j)_j$, converging to $\infty$ along which the $\liminf$ in \eqref{grail2bis} is realized, in the sense that $\liminf_{k\to \infty}F_{GC,N,s}(\mu_k)=\lim_{j\to \infty}F_{GC,N,s}(\mu_{k_j})$, and such that furthermore the measures $\lambda_{k_j}$ converge weakly, i.e.
	\begin{equation}\label{convlambdas}
	\mathcal P(\N_{\ge 0})\ni \lambda_{k_j}\stackrel{*}{\rightharpoonup}\lambda^*\in\mathcal P(\N_{\ge 0}).
	\end{equation}
	\textbf{Step 5.} \textit{Tightness of $(\gamma_{\mathrm{GC},N,k})_{k\ge 1}$.} 
	
	Due to \eqref{gammanssym1} coupled with the result \eqref{bound_high_n} of Step 4, coupled with the links \eqref{gammagcs} and \eqref{gammalambda} between $\gamma_{\mathrm{GC},N,k}$ and $\gamma_{n,k},\lambda_{n,k}$, and by using \eqref{gammalambdamu}, we find that
	\begin{eqnarray}\label{bound_high_n_gamma}
	\sup_{k\ge 1}\gamma_{\mathrm{GC},N,k}\bigg(X_{\mathrm{GC}}\setminus \coprod_{2\le n< n_\alpha}\Lambda^n_R\bigg)
	&\le&\sup_{k\ge 1}\sum_{2\le n<n_\alpha}\gamma_{\mathrm{GC},N,k}\left(X_n\setminus \Lambda^n_R\right)\ +\sup_{k\ge 1}\gamma_{\mathrm{GC},N,k}\bigg(\coprod_{n\ge n_\alpha}X_n\bigg)\nonumber\\
	&\le&\sup_{k\ge 1}\sum_{2\le n<n_\alpha}n\lambda_{n,k}\epsilon+ l(\alpha)\cdot N^2 C+N\delta\nonumber\\
	&\le& (l(\alpha))^{-1}\cdot N^2 C +\ 2N\delta.
	\end{eqnarray}
	Due to the fact that for any choice of $\delta>0$ the set $B_R$ is compact and that the product of finitely many compact sets is compact, we find that the set $\coprod_{2\le n<n_\alpha}\Lambda^n_R$ is compact, and thus by the arbitrarity of $\alpha, \delta>0$, the bound \eqref{bound_high_n_gamma} shows that the sequence obtained from the previous step, denoted $\gamma_{\mathrm{GC},N,k_j}$, is tight. Thus, up to extracting a further subsequence, we find a measure $\gamma^*\in\mathcal M_+(X_{\mathrm{GC}})$ to which the $\gamma_{\mathrm{GC},N,k_j}$ converge weak-$*$, namely we have 
	\begin{equation}\label{tightconv_gamma}
	\gamma_{\mathrm{GC},N,k_j}\stackrel{*}{\rightharpoonup} \gamma^*\in\mathcal M_+(X_{\mathrm{GC}}).
	\end{equation}

	\textbf{Step 6.} \textit{$(\lambda^*, \gamma^*_{\mathrm{GC},N})$ is a competitor for $F^{\mathrm{OT}}_{\mathrm{GC}, N,\mathsf{c}}(\mu)$.}  
	
	Due to the facts \eqref{convlambdas} and \eqref{tightconv_gamma}, we obtain that the constraint \eqref{gammalambda} holds also for $(\gamma^*,\lambda^*)$. By testing the weak convergence from \eqref{convlambdas} against the function $f\in C_b(\N_{\ge 0})$, defined to be equal to one at $n$ and zero elsewhere (which is a continuous function on the discrete space $\N_{\ge 0}$), we find that for each $n\in\N_{\ge 0}$ there holds along our subsequence
	\begin{equation}\label{lambdan}
	\lambda_{n,k_j}\stackrel{\text{\eqref{lambdas},\eqref{encode_competitor}}}{=}\lambda_{k_j}(\{n\})\to\lambda^*(\{n\})=:\lambda^*_{n}\in[0,1],
	\end{equation}
	Condition \eqref{gammalambda}, together with the fact that by (\ref{tightconv_gamma}) we have for all $n\ge 2$ that $\gamma_{\mathrm{GC},N,k_j}|_{X_n}\stackrel{*}{\rightharpoonup} \gamma^*|_{X_n},$ means that for each $n\in \N_{\ge 2}$ such that $\lambda^*_{n}>0$, there holds
	\begin{equation}\label{gammana}
	\gamma_{\mathrm{GC},N,k_j}|_{X_n}\stackrel{\text{\eqref{gammagcs},\eqref{encode_competitor}}}{=}\lambda_{n,k_j}\gamma_{n,k_j}\stackrel{*}{\rightharpoonup}\gamma^*|_{X_n},
	\end{equation}
	which implies in turn via \eqref{lambdan} that whenever $\lambda_n^*\neq 0$ we have
	\begin{equation}\label{gamman}
	\gamma_{n,k_j}\stackrel{*}{\rightharpoonup}(\lambda^*_{n})^{-1}\gamma^*|_{X_n}:=\gamma^*_{n}.
	\end{equation}
	It follows from \eqref{lambdan} and \eqref{gamman} that $\gamma_{n}^*\in\mathcal P(X_n)$. It remains to show that  \eqref{gammalambdamu} holds with  $(\lambda_{n}^*,\gamma_{n}^*)$.
	
	\medskip 
	
	For the case of $\mathsf{c}>0$ and $f\in C^0_c(\mathbb R^d)$ with support contained in the ball $B(0,R)$, we have $\mathsf{c}(x,y)>\min_{x',y'\in B(0,R)}\mathsf{c}(x',y')=m_R>0$, valid for all $x,y\in \mathrm{supp}(f)$, and thus we have the pointwise bound
	\begin{equation}\label{bound_f_c}
	\mathsf{c}(x',y')\ge \frac{m_R}{\sup_{y\in\mathrm{supp}(f)}|f(y)|}|f(x')|:=C_f |f(x')|\quad\mbox{for all}\quad x',y'\in\mathbb R^d.
	\end{equation}
	By integrating in \eqref{bound_f_c} and summing over $i,j,$ and by the marginal property of $\gamma_{n,k}$, there holds for $k\in\mathbb N$
	\begin{eqnarray}\label{fgcbound}
	\lefteqn{\sum_{n=2}^\infty\lambda_{n, k}\int_{{\mathbb{R}}^{Nd}}\,\sum_{i,j=1,i\neq j}^n \mathsf{c}(x_i,x_j) d \gamma_{n,k}(x_1,\ldots,x_n)}\nonumber\\
	&\ge& C_f\sum_{n=2}^\infty\lambda_{n, k} n(n-1)\int_{{\mathbb{R}}^{Nd}}\,|f(x_1)| d \gamma_{n,k}(x_1,\ldots,x_n)=C_f\sum_{n=2}^\infty n(n-1)\lambda_{n,k}\int |f(x)|d\mu_{n,k}(x).
	\end{eqnarray}
	In view of \eqref{fgcbound} we have that
	\[
	C_f\sum_{n=2}^\infty n(n-1)\lambda_{n,k}\int |f|\ d\mu_{n,k}\le F^{\mathrm{OT}}_{\mathrm{GC},N,\mathsf{c}}(\mu_k)\le \sup_{k} F^{\mathrm{OT}}_{N,\mathsf{c}}(\mu_k)<N^2\sup_{k\ge 1}\int_{\mathbb{R}^{2d}}\mathsf{c}(x,y)\mathrm{d}\mu_k(x)\mathrm{d}\mu_k(y)<\infty,
	\]
	thus we obtain that at fixed $N, f$ for $(\vec\lambda_k,\vec \gamma_k)$ optimizing $F^{\mathrm{OT}}_{\mathrm{GC},N,\mathsf{c}}(\mu_k)$ 
	\begin{equation}\label{unifboundgammainfty}
	\sup_{k\ge 1}\sum_{n=2}^\infty n(n-1)\lambda_{n,k}\int |f(x)|\ d\mu_{n,k}(x)<\infty.
	\end{equation}
	Therefore, as $\mu_{n,k}$ are positive measures, the integrals
	\[
	\sum_{n=2}^\infty n\lambda_{n,k}\ \int f \mu_{n,k}
	\]
	are uniformly summable as $k\to\infty$ and thus equation \eqref{gammalambdamu} also passes to the weak limit in duality with $f\in C^0_c(\mathbb R^d)$, and holds with $\lambda_{n}^*,\gamma_{n}^*$ as in \eqref{lambdan}, \eqref{gamman}. Moreover, since the property of being a symmetric probability measure passes to weak limits, we find from \eqref{gamman} also that $\gamma_{n}^*\in\mathcal P_\mathrm{sym}\left((\R^d)^n\right)$. Thus $\gamma_{\mathrm{GC},N}^*:=\sum_{n\ge 2}\lambda_{n}^*\gamma_{n}^*\in\mathcal P(X_{\mathrm{GC}})$ is a competitor to $F^{\mathrm{OT}}_{\mathrm{GC},N,\mathsf{c}}(\mu)$. 
	
	\textbf{Step 7.} \textit{Conclusion of the proof.} 
	
	We start by noting that for all $n\ge 2$, we have 
	\begin{equation}\label{limitcompetn}
	\int_{\mathbb{R}^d\times \mathbb{R}^d}\mathsf{c}(x,y)\ \lambda^*_{n}d(\gamma^*_{n})^{(2)}(x,y)\le \liminf_{j\rightarrow \infty}\int_{\mathbb{R}^d\times \mathbb{R}^d}\mathsf{c}(x,y)\ \lambda_{n,k_j}\ d\gamma_{n, k_j}^{(2)}(x,y)=:\liminf_{j\to\infty}f_j(n).
	\end{equation}
	The above follows by applying Lemma 6.1 from \cite{ss2d} (or \cite[Thm. 3.9]{ButChampdePas16}). The assumptions of the lemma are then verified on the space $X=\mathbb{R}^d\times \mathbb{R}^d$ by $\gamma^{(2)}_{n,k_j}$ and $\mathsf{c}$.
	
	For all $k\ge 1$, we apply Fatou's theorem for the atomic measure $\nu$ on $\mathbb N$ with $\nu(\{n\})=n(n-1)$ for all $n\in \mathbb N$ and for $f_j(n)$ defined in \eqref{limitcompetn}, so by summing first over $n$ the positive terms coming from \eqref{limitcompetn} 
	\begin{eqnarray} 
	\lefteqn{\sum_{n=2}^\infty\int_{\mathbb{R}^d\times \mathbb{R}^d}n(n-1)\lambda^*_{n}\mathsf{c}(x,y)\ d(\gamma^*_{n})^{(2)}(x,y)\le\sum_{n=2}^\infty\liminf_{j\to\infty}\left(n(n-1)\int_{\mathbb{R}^d\times \mathbb{R}^d}\lambda_{n,k_j}\mathsf{c}(x,y)\ d\gamma_{n, k_j}^{(2)}(x,y)\right)}\nonumber\\
	&=&\int \liminf_{j\to\infty}f_j(n) d\nu(n)\le\liminf_{j\to\infty}\int f_j(n)d\nu(n)=\liminf_{j\to\infty}\left(\sum_{n=2}^\infty n(n-1)\int_{\mathbb{R}^d\times \mathbb{R}^d}\lambda_{n,k_j}\mathsf{c}(x,y)\ d\gamma_{n, k_j}^{(2)}(x,y)\right)\nonumber\\
	&=&\liminf_{j\to\infty}F_{\mathrm{GC},N,\mathsf{c}}^\mathrm{OT}(\mu_{k_j}).\label{limitn1}
	\end{eqnarray}
	Then, as  by Step 6 we have that $\gamma_{\mathrm{GC},N}^*$ is a competitor for $F^{\mathrm{OT}}_{\mathrm{GC}, N}(\mu)$, we get
	\begin{equation}\label{competineq}
	F_{\mathrm{GC},N,\mathsf{c}}^\mathrm{OT}(\mu)\le \sum_{n=2}^\infty\int_{\mathbb{R}^d\times \mathbb{R}^d}n(n-1)\ \lambda^*_{n}\ \mathsf{c}(x,y)\ d(\gamma^*_{n})^{(2)}(x,y)\le \liminf_{j\to\infty}F_{\mathrm{GC},N,\mathsf{c}}^\mathrm{OT}(\mu_{k_j}),\end{equation}
	where for the second inequality we applied \eqref{limitn1}. The \eqref{competineq} thus implies that \eqref{grail2bis} holds.
\end{proof}
We now show
\begin{lemma}\label{existgc} (Existence of an optimal solution for $F^{\mathrm{OT}}_{\mathrm{GC},N,\mathsf{c}}$)
Let $\mathsf{c}$ be as in \eqref{ass_c_main} and let $\mu\in {\mathcal P}({\mathbb{R}}^d)$ such that the integrals required to define the quantities involved in the definition of $F^{\mathrm{OT}}_{\mathrm{GC},N,\mathsf{c}}(\mu)$ are finite. Then $F^{\mathrm{OT}}_{\mathrm{GC},N,\mathsf{c}}(\mu)$ has at least one solution.
\end{lemma}
\begin{proof}
It suffices to take a sequence of competitors $((\lambda_{n,k})_{n\ge 0}, (\gamma_{n,k})_{n\ge 2})$ to $F^{\mathrm{OT}}_{\mathrm{GC},N,\mathsf{c}}(\mu)$ such that 
$$\sum_{n\ge 2}\lambda_{n,k}F^{\mathrm{OT}}_n(\mu_{n,k})\to F^{\mathrm{OT}}_{\mathrm{GC},N,\mathsf{c}}(\mu),$$ 
and then proceed similarly to the proof in Lemma \ref{gcotlim}: there exists a tight subsequence converging to a competitor $((\lambda_n^*)_{n\ge 0},(\gamma_n^*)_{n\ge 2})$, and the thesis follows by the same arguments as in Lemma \ref{gcotlim}. This will allow us to say that
$$F^{\mathrm{OT}}_{\mathrm{GC},N,\mathsf{c}}(\mu)=\liminf_{k\to\infty}\sum_{n\ge 2}\lambda_{n,k}F^{\mathrm{OT}}_n(\mu_{n,k})\ge F^{\mathrm{OT}}_{\mathrm{GC},N,\mathsf{c}}(\mu^*),$$
which will imply in particular that $((\lambda_n^*)_{n\ge 0},(\gamma_n^*)_{n\ge 2})$ is a minimizer.
\end{proof}

\section*{Acknowledgments}
We thank M. Lewin for discussing with us in January 2016 in IHP and in February 2017. Moreover, we are grateful to him for pointing out to us in January 2016  the useful tools from our Propositions \ref{subadd3} and \ref{unif}, and for giving to us at the time the slides of his November 2015 talk. We also thank M. Lewin for useful comments on a
preliminary version of the present work.

\par We are indebted to Gero Friesecke for suggesting to us the use of the Fefferman-de Llave decomposition, which in turn led us to the Fefferman-Gregg decomposition which we extended and used in our manuscript. Our heartfelt gratitude goes to David Brydges, whom we cannot thank enough for the numerous suggestions of tools involving positive-definiteness and else, and whose tireless encouragement helped us throughout.

\par Both authors thankfully acknowledge the support of a Royal Society International Exchanges Grant, and of an IHP Research-in-Pairs grant. Both authors also thank the organisers of the 2014 Fields Institute Thematic Program on
'Variational Problems in Physics, Economics and Geometry', where this research was first started. MP was funded by an EPDI fellowship and by the Forschungsinstitut f\"ur Mathematik at ETH Z\"urich, during different phases of this project.

\par Declarations of interest: none.

\end{document}